\pdfoutput=1
\documentclass[11pt, letterpaper]{article}
\usepackage{blindtext}
\usepackage{hyperref}
\hypersetup{
	colorlinks=true,
	linkcolor=blue!70!black,
	citecolor=blue!70!black,
	urlcolor=blue!70!black
}

\usepackage[margin=1in]{geometry}

\usepackage{amsfonts}
\usepackage{amsmath}
\usepackage{amssymb}
\usepackage{amsthm}
\usepackage{booktabs}
\usepackage{algorithm}
\usepackage[lined,boxed,ruled,norelsize,algo2e,linesnumbered,noend]{algorithm2e}
\usepackage{bbm}
\usepackage{bbold}
\usepackage{enumitem}

\usepackage{graphicx}
\graphicspath{{./figs/}}

\usepackage{cleveref}
\usepackage{thm-restate}
\usepackage{xcolor}

\definecolor{burntorange}{rgb}{0.8, 0.33, 0.0}

\newtheorem{theorem}{Theorem}

\newtheorem{lemma}[theorem]{Lemma}
\newtheorem{corollary}[theorem]{Corollary}

\newtheorem{definition}[theorem]{Definition}
\newtheorem{proposition}[theorem]{Proposition}
\newtheorem{fact}[theorem]{Fact}

\newtheorem{remark}[theorem]{Remark}

\renewcommand{\epsilon}{\varepsilon}

\newcommand{\defeq}{\stackrel{\textup{\tiny def}}{=}}
\newcommand{\norm}[1]{\left\lVert#1\right\rVert}

\newcommand{\normop}[1]{\left\lVert#1\right\rVert_{\textup{op}}}
\newcommand{\normsop}[1]{\lVert#1\rVert_{\textup{op}}}

\newcommand{\eps}{\epsilon}
\newcommand{\lam}{\lambda}

\newcommand{\R}{\mathbb{R}}

\newcommand{\N}{\mathbb{N}}

\newcommand{\gC}{\mathcal{C}}

\newcommand{\gH}{\mathcal{H}}
\newcommand{\gA}{\mathcal{A}}

\newcommand{\gB}{\mathcal{B}}

\newcommand{\diag}[1]{\textbf{\textup{diag}}\left(#1\right)}

\DeclareMathOperator*{\pr}{\mathbb{P}}
\newcommand{\E}{\mathbb{E}}

\newcommand{\Vol}{\textup{vol}}

\newcommand{\id}{\mathbf{I}}

\newcommand{\im}{\textup{im}}

\newcommand{\Dmax}{\Delta_{\textup{max}}}
\newcommand{\Dmin}{\Delta_{\textup{min}}}

\newcommand{\tO}{\widetilde{O}}

\newcommand{\nnz}{\textup{nnz}}

\newcommand{\Par}[1]{\left(#1\right)}

\newcommand{\Brace}[1]{\left\{#1\right\}}
\newcommand{\Abs}[1]{\left|#1\right|}

\newcommand{\ddi}{\boldsymbol{\mathrm{d}}^{(\textup{i})}}
\newcommand{\ddo}{\boldsymbol{\mathrm{d}}^{(\textup{o})}}
\newcommand{\ddpi}{[\boldsymbol{\mathrm{d}}']^{(\textup{i})}}
\newcommand{\ddpo}{[\boldsymbol{\mathrm{d}}']^{(\textup{o})}}

\newcommand{\Sym}{\mathbb{S}}

\newcommand{\poly}{\textup{poly}}
\newcommand{\polylog}{\textup{polylog}}
\newcommand{\polyloglog}{\textup{polyloglog}}
\newcommand{\bO}{\Breve{O}}

\newcommand{\supp}{\textup{supp}}
\newcommand{\vol}{\textup{Vol}}

\newcommand{\spn}{\textup{span}}

\newcommand{\ceil}[1]{\left\lceil#1\right\rceil}
\newcommand{\floor}[1]{\left\lfloor#1\right\rfloor}

\def\pleq{\preccurlyeq}
\def\pgeq{\succcurlyeq}

\renewcommand\AA{\boldsymbol{\mathrm{{A}}}}
\newcommand\BB{\boldsymbol{\mathrm{{B}}}}
\newcommand\CC{\boldsymbol{\mathrm{{C}}}}
\newcommand\DD{\boldsymbol{\mathrm{{D}}}}

\newcommand\FF{\boldsymbol{\mathrm{{F}}}}

\newcommand\HH{\boldsymbol{\mathrm{{H}}}}
\newcommand\II{\boldsymbol{\mathrm{{I}}}}

\newcommand\MM{\boldsymbol{\mathrm{{M}}}}

\newcommand\LL{\boldsymbol{\mathrm{{L}}}}
\newcommand\PP{\boldsymbol{\mathrm{{P}}}}
\newcommand\QQ{\boldsymbol{\mathrm{{Q}}}}
\newcommand\RR{\boldsymbol{\mathrm{{R}}}}
\renewcommand\SS{\boldsymbol{\mathrm{{S}}}}
\newcommand\TT{\boldsymbol{\mathrm{{T}}}}
\newcommand\UU{\boldsymbol{\mathrm{{U}}}}
\newcommand\WW{\boldsymbol{\mathrm{{W}}}}

\newcommand\XX{\boldsymbol{\mathrm{{X}}}}
\newcommand\YY{\boldsymbol{\mathrm{{Y}}}}
\newcommand\ZZ{\boldsymbol{\mathrm{{Z}}}}

\newcommand\DDi{\boldsymbol{\mathrm{{D}}}^{(\textup{i})}}
\newcommand\DDo{\boldsymbol{\mathrm{{D}}}^{(\textup{o})}}
\newcommand\DDpo{[\boldsymbol{\mathrm{{D}}}']^{(\textup{o})}}
\newcommand\DDpi{[\boldsymbol{\mathrm{{D}}}']^{(\textup{i})}}

\newcommand\PPi{\boldsymbol{\mathrm{\Pi}}}

\newcommand{\tA}{\widetilde{\AA}}

\newcommand{\tG}{\widetilde{G}}

\renewcommand\aa{\boldsymbol{\mathrm{a}}}
\newcommand\bb{\boldsymbol{\mathrm{b}}}
\newcommand\cc{\boldsymbol{\mathrm{c}}}
\newcommand\dd{\boldsymbol{\mathrm{d}}}
\newcommand\ee{\boldsymbol{\mathrm{e}}}
\newcommand\ff{\boldsymbol{\mathrm{f}}}

\newcommand\kk{\boldsymbol{\mathrm{k}}}

\newcommand\rr{\boldsymbol{\mathrm{r}}}

\newcommand\uu{\boldsymbol{\mathrm{u}}}
\newcommand\vv{\boldsymbol{\mathrm{v}}}
\newcommand\ww{\boldsymbol{\mathrm{w}}}
\newcommand\yy{\boldsymbol{\mathrm{y}}}
\newcommand\zz{\boldsymbol{\mathrm{z}}}
\newcommand\xx{\boldsymbol{\mathrm{x}}}
\newcommand\ppi{\boldsymbol{\mathrm{\pi}}}

\newcommand{\tkk}{\tilde{\kk}}

\newcommand\vone{\boldsymbol{\mathrm{1}}}
\newcommand\vzero{\boldsymbol{\mathrm{0}}}
\newcommand\vLL{\vec{\LL}}
\newcommand\vMM{\vec{\MM}}
\newcommand{\vG}{\vec{G}}
\newcommand{\vH}{\vec{H}}
\newcommand{\vS}{\vec{S}}

\newcommand{\vC}{\vec{C}}
\newcommand{\vB}{\vec{B}}
\newcommand{\rev}{\textup{rev}}
\newcommand{\und}{\textup{und}}
\newcommand{\blift}{\textup{blift}}

\newcommand{\ER}{\textup{ER}}

\newcommand{\ffs}{\ff_{\star}}

\newcommand{\te}{\tilde{e}}
\newcommand{\tf}{\tilde{f}}

\newcommand{\hxx}{\hat{\xx}}
\newcommand{\hyy}{\hat{\yy}}

\newcommand\ccpc{C_{\textup{cpc}}}
\newcommand\csbg{C_{\textup{sbg}}}
\newcommand\cbal{C_{\textup{bal}}}
\newcommand{\cadk}{C_{\textup{ADK}}}
\newcommand{\cded}{C_{\textup{DED}}}
\newcommand{\cds}{C_{\textup{DS}}}
\newcommand{\css}{C_{\textup{SS}}}
\newcommand{\bv}{\bar{v}}
\newcommand{\bV}{\bar{V}}

\DontPrintSemicolon
\SetFuncSty{textsc}
\SetKwProg{Pro}{Procedure}{}{}
\SetKwFunction{RO}{Rounding}
\SetKwFunction{PA}{Patching}
\SetKwFunction{PAE}{PatchingExternal}
\SetKwFunction{PAI}{PatchingInternal}
\SetKwFunction{PAS}{PatchingStar}
\SetKwFunction{SDS}{SparsifyDirectedSpectral}
\SetKwFunction{SS}{SparsifySubgraph}
\SetKwFunction{SDC}{SparsifyDiCut}
\SetKwFunction{SDM}{SparsifyDiCutMSF}
\SetKwFunction{SDMO}{SparsifyDiCutMSFOnce}
\SetKwFunction{EPA}{ExpanderDecompADK}
\SetKwFunction{EP}{ExpanderDecomp}
\SetKwFunction{DSP}{DegreeSparsifyPreprocess}
\SetKwFunction{VS}{VertexSample}
\SetKwFunction{SA}{SubsetSample}
\SetKwFunction{DSD}{DegreeSparsifyDelete}
\SetKwFunction{PR}{PreconRichardson}
\newcommand{\ROalgo}{\textsc{Rounding}}
\newcommand{\PAalgo}{\textsc{Patching}}
\newcommand{\PAEalgo}{\textsc{PatchingExternal}}
\newcommand{\PAIalgo}{\textsc{PatchingInternal}}
\newcommand{\PASalgo}{\textsc{PatchingStar}}
\newcommand{\SDSalgo}{\textsc{SparsifyDirectedSpectral}}
\newcommand{\SSalgo}{\textsc{SparsifySubgraph}}
\newcommand{\SDCalgo}{\textsc{SparsifyDiCut}}
\newcommand{\SDMalgo}{\textsc{SparsifyDiCutMSF}}
\newcommand{\SDMOalgo}{\textsc{SparsifyDiCutMSFOnce}}
\newcommand{\EPAalgo}{\textsc{ExpanderDecompADK}}

\newcommand{\DSPalgo}{\textsc{DegreeSparsifyPreprocess}}

\newcommand{\SAalgo}{\textsc{SubsetSample}}
\newcommand{\DSDalgo}{\textsc{DegreeSparsifyDelete}}
\newcommand{\PRalgo}{\textsc{PreconRichardson}}

\newcommand{\codeInput}{\textbf{Input:} }
 \SetKwComment{Comment}{$\triangleright$\ }{}
\SetCommentSty{mycommfont}

\renewcommand{\eps}{\varepsilon}

\DeclareMathOperator{\Sc}{\mathtt{Sc}}

\usepackage[backend=biber, isbn=false, style=alphabetic, backref=true, doi=false, url=false, maxcitenames=10, mincitenames=5, maxalphanames=10, maxbibnames=10, minbibnames=5, minalphanames=5, defernumbers=true, sortlocale=en_US]{biblatex}
\title{Fully Dynamic Spectral and Cut Sparsifiers for Directed Graphs}
\author{Yibin Zhao \\University of Toronto \\ \texttt{ybzhao@cs.toronto.edu}}
\date{}

\begin{document}
\maketitle

\begin{abstract}
Recent years have seen extensive research on directed graph sparsification.
In this work, we initiate the study of fast fully dynamic spectral and cut
sparsification algorithms for directed graphs.

We introduce a new notion of spectral sparsification called degree-balance
preserving spectral approximation, which maintains the difference between the
in-degree and out-degree of each vertex.
The approximation error is measured with respect to the corresponding undirected
Laplacian.
This notion is equivalent to direct Eulerian spectral approximation when the
input graph is Eulerian.
Our algorithm achieves an amortized update time of $O(\varepsilon^{-2} \cdot
\text{polylog}(n))$ and produces a sparsifier of size $O(\varepsilon^{-2} n \cdot
\text{polylog}(n))$.
Additionally, we present an algorithm that maintains a constant-factor
approximation sparsifier of size $O(n \cdot \text{polylog}(n))$ against an adaptive
adversary for $O(\text{polylog}(n))$-partially symmetrized graphs, a notion
introduced in [Kyng-Meierhans-Probst Gutenberg '22].
A $\beta$-partial symmetrization of a directed graph $\vG$ is the union of $\vG$
and $\beta \cdot G$, where $G$ is the corresponding undirected graph of $\vG$.
This algorithm also achieves a polylogarithmic amortized update time.

Moreover, we develop a fully dynamic algorithm for maintaining a cut sparsifier
for $\beta$-balanced directed graphs, where the ratio between weighted incoming
and outgoing edges of any cut is at most $\beta$.
This algorithm explicitly maintains a cut sparsifier of size
$O(\varepsilon^{-2}\beta n \cdot \text{polylog}(n))$ in worst-case update time
$O(\varepsilon^{-2}\beta \cdot \text{polylog}(n))$.
\end{abstract}

\section{Introduction} \label{sec:intro}
Designing spectral algorithms for directed graphs is a major frontier in algorithmic spectral graph theory.
These algorithms have found numerous applications, ranging from fast algorithms
for processing Markov chains \cite{CohenKPPSV16,AhmadinejadJSS19} to
deterministic low-space computation \cite{AhmadinejadKMPS20}.
Spectral sparsification of directed (Eulerian) graphs has emerged as a central topic in
this area, with considerable recent advances
\cite{CohenKPPRSV17,ChuGPSSW18,AhmadinejadPPSV23,SachdevaTZ24,JambulapatiSSTZ25,LauWZ25}.
In particular, faster spectral sparsification algorithms that produce sparser directed
Eulerian sparsifiers directly lead to faster computation of approximate
Personalized PageRank vectors \cite{CohenKPPSV16}, stationary distributions,
hitting times, escape probabilities of random walks \cite{CohenKPPSV16}, and
approximate Perron vectors \cite{AhmadinejadJSS19} through efficient directed
Eulerian Laplacian solvers \cite{CohenKPPRSV17,PengS22}.

In this work, we initiate the study of spectral and cut sparsification for
dynamically changing directed graphs.
We believe that our dynamic algorithms will serve as the foundation for
efficient dynamic computations of the aforementioned problems.

Unlike undirected spectral sparsification, developing a useful notion of
directed spectral or cut sparsification has been particularly challenging.
Consider the complete directed bipartite graph where every node on one side of
the bipartition has a directed edge to every node on the other side.
Any sparsifier that approximately preserves all directed cuts in such a graph
cannot delete any edges.
For cut sparsification, Cen, Cheng, Panigrahi, and Sun~\cite{CenCPS21} restricted
the sparsification guarantees to only apply to cuts with values close to their
corresponding cuts in the reverse direction.

\begin{definition}[$\beta$-balanced directed cut approximation]
    \label{def:baldicuappr}
    $\vH=(V,E_{\vH},\ww_{\vH})$ is a \emph{$(\beta,\eps)$-balanced directed cut
    approximation} of $\vG = (V,E,\ww)$ if for all non-trivial cuts $(U,V
    \setminus U)$ satisfying $\frac{1}{\beta} \ww(V \setminus U, U) \le
    \ww(U,V \setminus U) \le \beta \ww(V \setminus U, U)$, $\vH$ satisfies that
    \[
        (1-\eps) \ww(U,V \setminus U) \le \ww_{\vH}(U, V \setminus U) \le
        (1+\eps) \ww(U, V \setminus U).
    \]
\end{definition}

This notion of $\beta$-balanced cut sparsification is motivated by applications in
solving flows and cuts problems \cite{EneMPS16,KargerL15}, as the residual graphs
that appear in the intermediate steps of the maxflow algorithms are naturally
balanced.
These sparsifiers have found further application in recent work by Goranci,
Henzinger, Räcke, and Sricharan \cite{GoranciHRS25}, who used them to achieve
the first dynamic incremental maxflow algorithm with polylogarithmic amortized
update time for dense graphs ($m = \Omega(n^2)$).

Cohen, Kelner, Peebles, Peng, Rao, Sidford, and Vladu~\cite{CohenKPPRSV17}
addressed the directed spectral sparsification
problem by focusing on directed Eulerian graphs (where every vertex has equal
weighted in-degree and out-degree).
This requirement might seem overly restrictive, since all directed cuts of an
Eulerian graph must have a balance factor $\beta=1$.
However, Cohen et al.~demonstrated that their sparsification approach suffices
for developing fast solvers for all directed Laplacian linear systems (even
those not corresponding to an Eulerian graph).
In this paper, we propose a relaxed notion of directed spectral approximation
closely related to that of \cite{CohenKPPRSV17}, where the sparsifier only needs
to maintain the degree balances.

\begin{definition}[Degree balance preserving directed spectral approximation]
    \label{def:dispecappr}
    $\vH=(V,E_{\vH},\ww_{\vH})$ is a \emph{$\eps$-degree balance preserving
    directed spectral approximation} of $\vG = (V,E,\ww)$ if for the
    corresponding undirected $G \defeq \und(\vG)$,
    \begin{equation}\label{eq:dispec_def}
        \forall \xx,\yy \in \R^V, \quad
        \Abs{\xx^\top (\vLL_{\vG} - \vLL_{\vH}) \yy}
        \le \eps \sqrt{\xx^\top \LL_G \xx \cdot \yy^\top \LL_G \yy}.
    \end{equation}
\end{definition}

This definition of directed spectral approximation has been implicitly used in
almost all previous works in the directed Laplacian literature (e.g.
\cite{CohenKPPRSV17,CohenKKPPRSV18,ChuGPSSW18,KyngSPG22,SachdevaTZ24,JambulapatiSSTZ25}).
We note that the degree balance preserving requirement is implicitly ensured by
\eqref{eq:dispec_def}: if the degree balance is not preserved, then there are
some vectors $\xx,\yy$ such that the left hand side of \eqref{eq:dispec_def} is
non-zero while the right hand side is zero (see \Cref{lemma:dispec_equiv}).
The relaxation to any directed graph is necessary for dynamically changing
graphs, as modifying from one directed Eulerian graph to another requires at
least 3 individual edge updates.
Whenever a dynamic directed graph $\vG$ becomes Eulerian, our degree balance
preservation requirement naturally ensures that the dynamically maintained
sparsifier is also Eulerian.

For both directed balanced cut sparsification and degree balance preserving
spectral sparsification, we ask the natural question of whether a fully dynamic
algorithm with polylogarithmic update time exists for maintaining a nearly-linear
sized sparsifier (i.e., the number of edges in the sparsifier is at most $O(n
\cdot \poly(\eps^{-1},\log n)$) for directed graphs.
Previous work by Abraham, Durfee, Koutis, Krinninger, and
Peng~\cite{AbrahamDKKP16} established the existence of such dynamic algorithms
for undirected graphs.
In this work, we answer this question affirmatively for the directed case.

\subsection{Our results}
Our first result is a fully dynamic algorithm for degree balance preserving
directed spectral sparsification %
with polylogarithmic amortized update time.

\begin{theorem}[Explicit dynamic directed spectral sparsifier] \label{thm:dynspec_star_inf}
    Given a directed graph $\vG$ on vertices $V$ and with polynomially bounded edge weights, we
    can maintain \emph{explicitly} a graph $\vH$ on vertices $V \cup X$
    such that the Schur complement $\Sc(\vH,V)$ onto the original vertices is a
    $\eps$-degree balance preserving directed sparsifier.
    $\vH$ has size $\tO(\eps^{-2} n)$~\footnotemark~and satisfies $|X| = \tO(n)$.
    The algorithm has amortized recourse and update
    time $\tO(\eps^{-2})$ per edge insertion or deletion and works
    against an oblivious adversary.
    \footnotetext{We use $\tO(\cdot)$ to hide polylogarithmic factors in $n$.}
\end{theorem}

Our algorithm explicitly maintains a sparsifier at the expense of $\tO(n)$ extra
vertices.
We demonstrate in \Cref{lemma:schurprecon} (\Cref{sec:spectral}) that these
extra vertices do not affect the usefulness of our theorem.
Crucially, having an explicit sparsifier with low recourse enables further 
dynamic algorithms using our algorithm as a dynamic subroutine.
A version of this theorem that exactly preserves both in-degrees and out-degrees
is presented in \Cref{thm:dynspecstar_for_deg}.

We also present two different algorithms that maintain \emph{implicit}
sparsifiers on the original set of vertices in \Cref{sec:spectral}.
Specifically, \Cref{thm:dynspecext_for} gives an algorithm that
supports fast querying of the edges and the entire sparsifier.
The other algorithm (see \Cref{thm:dynspecint_for} in \Cref{sec:spectral})
supports maintenance of a subgraph sparsifier.

As observed by \cite{BernsteinvdBPGNSSS22} for dynamic undirected spectral
sparsifiers, dynamic algorithms that maintain implicit sparsifiers already
support important applications.
For example, one can maintain a dynamic row and column diagonally dominant
(RCDD) matrix\footnotemark~solver through a reduction to directed Eulerian
spectral sparsifier by \cite{CohenKPPSV16} and by running a static directed
Eulerian Laplacian solver (e.g., \cite{PengS22,JambulapatiSSTZ25}) on top of the
sparsifier.
\footnotetext{A RCDD matrix is a real M-matrix with non-negative row and column
sums. An update on a single off-diagonal entry of a RCDD matrix only incurs 3
edge updates in its related directed Eulerian Laplacian matrix.}

Our next result, parallel to spectral sparsification, is a simple fully dynamic
balanced directed cut sparsification algorithm %
with polylogarithmic worst-case update time.
\begin{theorem} \label{thm:dyndicut_worst_inf}
    Given a directed graph with polynomially bounded edge weights, we
    can maintain explicitly a $(\beta,\eps)$-balanced directed cut sparsifier
    of size $\tO(\eps^{-2} \beta n)$ with worst-case update time
    $\tO(\eps^{-2} \beta) $ per edge insertion or deletion against an
    oblivious adversary.
\end{theorem}

Unlike spectral sparsification, cut sparsification does not require degree
fixing.
As a result, our data structure can maintain a cut sparsifier explicitly without
extra vertices.
Additionally, we present an algorithm with improved amortized recourse and
update time (see \Cref{thm:dyndicut_for} in \Cref{sec:dicut}).
In achieving \Cref{thm:dyndicut_worst_inf}, we also prove that
independent edge sampling based on inverse undirected edge connectivity is
sufficient for balanced directed cut approximation.

\begin{theorem} \label{thm:statdicut_inf}
    There is a sparsification algorithm that, given a weighted directed graph
    and the edge connectivities of its corresponding undirected graph, computes
    a $(\beta,\eps)$-balanced directed cut sparsifier with $O(\eps^{-2} \beta n
    \log n)$ edges in linear time.
\end{theorem}

Finally, we present a fully dynamic directed spectral sparsification algorithm
that works against an adaptive adversary. 
This result requires more detailed explanation, which we provide below.

We first remind readers that our degree balance preserving sparsifier requires
the approximation factor $\eps < 1$.
This requirement is necessary and warranted -- for any directed Eulerian graph,
an empty graph on the same set of vertices can serve as a constant good
approximation.
To the best of our knowledge, there exists no notion of directed spectral
approximation that extends beyond constant approximation.
This limitation is particularly problematic since the state-of-the-art adaptive
undirected cut and spectral sparsifier from \cite{BernsteinvdBPGNSSS22} require
the multiplicative approximation factor to be at least $O(\log n)$.

Kyng, Meierhans, and Probst Gutenberg circumvented this issue by considering the
sparsification problem on a related graph called a \emph{partial symmetrization}.
Specifically, a $\beta$-partial symmetrization of a directed graph $\vG$ is
defined by $\vG^{(\beta)}\defeq \beta \cdot G \cup \vG$, which combines the
directed graph with a copy of its corresponding undirected graph $G \defeq
\und(\vG)$ scaled by a factor of $\beta$.
They showed that a partially symmetrized Eulerian graph serves as a good
\emph{approximate pseudoinverse} of the original directed graph (see \Cref{lemma:quad0to1}).
\begin{definition}[Approximate pseudoinverse, \cite{CohenKPPSV16}]
    \label{def:app_pinv}
    For $\eps \ge 0$, a square matrix $\ZZ$ is an $\eps$-approximate pseudoinverse of
    square matrix $\MM$ with respect to a Hermitian PSD matrix $\UU$ if
    $\ker(\UU) \subseteq \ker(\MM) = \ker(\MM^\top) = \ker(\ZZ) =
    \ker(\ZZ^\top)$, and
    \[
        \norm{\PP_{\im(\MM)} - \ZZ\MM}_{\UU \to \UU} \le \eps,
    \]
    where $\PP_{\im(\MM)}$ is the identity matrix on the image of $\MM$.
\end{definition}
Moreover, when $\beta$ is large, the directed portion of $\vG^{(\beta)}$ can be
sparsified much more aggressively under the notion of degree balance preserving
spectral approximation, since
\[
    \normop{\LL_{G^{(\beta)}}^{\frac \dag 2} (\vLL_{\vG} - \vLL_{\vH})
    \LL_{G^{(\beta)}}^{\frac \dag 2}}
    =
    \frac{1}{2\beta+1} \normop{\LL_G^{\frac \dag 2} (\vLL_{\vG} - \vLL_{\vH})
    \LL_G^{\frac \dag 2}}.
\]
Further sparsification can also be performed on the undirected portion of the
graph $\vG^{(\beta)}$ to ultimately produce a sparse approximation.
These layers of approximation allow them to build a sparse preconditioner chain,
which can be applied to solve directed Eulerian Laplacian systems using
iterative solvers.
See \Cref{sec:adaptive} for more details.
Our dynamic algorithm efficiently maintains these sparsification layers under
adaptive edge insertions and deletions.

\begin{theorem}[Adaptive directed spectral sparsification quadruple]
    \label{thm:dyn_quad_adp_inf}
    Given a directed graph $\vG$ on vertices $V$ and with polynomially bounded edge weights, we can
    maintain \emph{explicitly} a set of 3 directed graphs $\vG_1 =
    \vG^{(\beta)},\vG_2',\vG_3'$ with $\vG_0 = \vG$ and $V(\vG_2') = V \cup X$
    and $V(\vG_3') = V \cup X \cup Y$.
    For $\vG_2 \defeq \Sc(\vG_2',V)$ and $\vG_3 \defeq \Sc(\vG_3',V)$, they
    satisfy that 
    \begin{enumerate}
    \item when $\vG$ is Eulerian $\vG_i$ is a
        $1-\frac{1}{\polylog(n)}$-approximate pseudoinverse of $\vG_{i-1}$ for all
        $i=1,2,3$ and degree balance preserving with respect to $\vG$,
    \item $\vG_3'$ has size $\tO(n)$ and $\vG_2'$ has size $O(m)$.
    \end{enumerate}
    The algorithm works against an adaptive adversary, has preprocessing time
    $\tO(m)$ and amortized update time $\tO(1)$.
\end{theorem}

\subsection{Overview}

\paragraph{Directed spectral sparsification}
For directed spectral sparsification, a major challenge compared to undirected
spectral sparsification is the degree balance preservation requirement.
Sampling edges independently, as in undirected spectral sparsification
algorithms~\cite{SpielmanS08}, may create degree imbalances.
Recent approaches in the static setting instead use independent sampling of
graph structures that inherently preserve degree balances.
Chu, Gao, Peng, Sachdeva, Sawlani, and Wang \cite{ChuGPSSW18} used short-cycle
sampling while Jambulapati, Sachdeva, Sidford, Tian, and Zhao
\cite{JambulapatiSSTZ25} employed electrical circulation sampling.
These structures are particularly difficult to maintain in a dynamic setting.
They prohibit one-shot sparsification, requiring instead a half-sparsification
framework where the number of edges reduces by roughly half per iteration.
For example, when removing half of the edges from a short cycle, one must keep
and double the weights of the remaining edges to preserve degree balance.
Consequently, a potential dynamic algorithm would need to maintain
half-sparsifiers for $O(\log n)$ levels.
Moreover, a single edge update causes an average recourse of at least
$\Omega(\log n)$ per level because most cycles have length $\Omega(\log n)$
\cite{ChuGPSSW18,LiuSY19,ParterY19}.
This recourse then quickly propagates, resulting in prohibitively large total
update time.

To avoid these issues, we instead adopt the approach of Cohen et
al.~\cite{CohenKPPRSV17}, where the sparsification algorithm performs
independent edge sampling on expanders and computes a ``patching'' graph to fix
degree imbalances caused by the sampling.
Cohen et al.~made the key observation that sampling a directed edge with
probability inversely proportional to both the in-degree of its tail and the
out-degree of its head in an ``expander'' (i.e., a directed graph whose
corresponding undirected graph is an expander) effectively controls spectral
error and degree imbalance.
They then showed that a small patching of size $O(n)$ can be computed for such
expander by greedily matching remaining imbalanced degrees.

We show in \Cref{sec:spectral} that a patching similar to the greedy patching of
\cite{CohenKPPRSV17} can be implicitly maintained by a tree data structure that
allows $\tO(1)$ time querying of any edge.
By treating each in-degree as an interval of that length and placing these
intervals side-by-side starting from 0 (and doing the same for out-degrees), we
can compute a patching by simply checking interval overlaps.
Specifically, we set the weight of a directed edge $\ww_{(u,v)} = |I^{(o)}_u
\cap I^{(i)}_v|$ for each pair of $u,v \in V$ with in-degree interval
$I^{(i)}_v$ and out-degree interval $I^{(o)}_u$.
These intervals can be deterministically maintained in $\tO(1)$ update
time using two segment trees.
This approach also ensures that the number of edges in the patching remains at
most $O(n)$.
We call this an \emph{external} patching since it introduces directed edges not
originally present in the graph.

To avoid having extra edges in a sparsifier, we introduce another patching
algorithm that slightly adjusts the edge weights of the sampled subgraph using
electrical flow.
Here, we observe that after the independent edge samplings, the corresponding
undirected graph of these samples remains an expander, since it forms an
undirected spectral sparsifier of the original expander.
When degree imbalances are small, an electrical routing of the degree demands
through an expander incurs only small maximum congestion on the edges
\cite{FlorescuKPGS24}, enabling us to maintain a subgraph.
This electrical routing can be computed quickly using an undirected Laplacian
solver (e.g.~\cite{JambulapatiS21}). 
A rounding algorithm, such as a tree routing algorithm (see
Algorithm~\ref{alg:round} and \cite{KelnerOSZ13,JambulapatiSSTZ25}), is used to
handle the small errors that arise from approximate undirected Laplacian
solving.
We call this algorithm an \emph{internal} patching algorithm.
For our dynamic algorithm, we only update the demand vector associated with the
degree imbalance and compute an electric flow when querying the entire graph.

A significant problem with both patching approaches is that the patching edges
cannot be maintained explicitly.
For external patching, modifying the degree of a vertex whose interval appears
early in the ordering changes almost all edge weights in the patching.
It can be worse for internal patching: the entire electrical flow is changed by
a single update.
To solve this problem, we observe that by adding auxiliary vertices, we can
effectively patch degree imbalances using just a single directed star graph with
one extra vertex.
This approach allows us to maintain the spectral sparsifier explicitly and
efficiently, as a single degree change only requires adjusting one edge on the
degree imbalance.
We remark that there is no tangible drawback for having a few auxiliary vertices
in a sparsifier.
For Eulerian graphs specifically, sparsifiers constructed this way still serve
as effective preconditioners for solving Eulerian systems (see
\Cref{lemma:schurprecon} in \Cref{ssec:dispec_prelim}).
There are also applications of undirected sparsifiers with auxiliary vertices
\cite{LiS18,ForsterGLPSY21}.

The framework of \cite{BernsteinvdBPGNSSS22} effectively reduces the fully
dynamic spectral sparsification problem to the same problem on pruned expanders
(see \cite{SaranurakW19}) undergoing decremental (edge deletion only) updates.
Recall that the algorithm of \cite{CohenKPPRSV17} samples edges independently
with probability inversely proportional to their incident vertex degrees.
For an oblivious adversary, we only need to resample edges incident to the two
vertices whose degrees are affected by an edge deletion.
Using an oversampling parameter $\rho = \eps^{-2} \cdot \poly(\phi^{-1},\log
n)$, we can apply a fast sampling algorithm
(e.g.~\cite{Knuth97,Devroye06,BringmannP12}) to achieve recourse and update time
roughly equal to the combinatorial degree $\tilde{O}(\rho)$ of the sampled
vertices. 

For exact degree-preserving sparsification where both in-degrees and out-degrees
must be preserved exactly (rather than just preserving their difference), there
are issues with our external and star patching schemes.
Our algorithm operates on directed bipartite graphs (see
\Cref{lemma:blift_spectral}) with two copies of each vertex.
Since these patching schemes may introduce edges not originally in the graph,
collapsing the two copies could create a self-loop.
To prevent this, we introduce a top-level decomposition of the complete
bipartite edge set into $O(\log n)$ disjoint subsets, ensuring no two copies of
the same vertex appear in any subset.
This approach guarantees that no self-loops are generated when collapsing each
subset, thus maintaining exact degree-preserving guarantees.

\paragraph{Balanced directed cut sparsification}
To facilitate dynamic balanced directed cut sparsification, we devise a novel
static sparsification algorithm.
Our algorithm samples each directed edge independently with probability
proportional to the inverse of its undirected connectivity, paralleling the
undirected cut sparsification algorithm developed in \cite{FunHHP11}.
We note that a similar algorithm for sparsifying balanced directed cuts was
recently discovered independently by Goranci, Henzinger, Räcke, and Sricharan
\cite{GoranciHRS25}~\footnotemark.
They showed (Theorem 2, \cite{GoranciHRS25}) how to compute such sparsifiers
using any undirected edge connectivity parameter (edge strength, connectivity,
and Nagamochi-Ibaraki Index \cite{NagamochiI92}) with slightly worse sparsity
guarantees compared to ours that specifically uses undirected edge
connectivity.
\footnotetext{This paper was originally submitted to FOCS for review. We were
unaware of \cite{GoranciHRS25} when writing and only learned about it later
through reviewer comments.}

Any $\beta$-balanced directed cut that crosses an edge with high undirected
connectivity $k$ has size at least $\frac{k}{\beta+1}$.
Intuitively, such an edge would have a low contribution to the overall size of a
cut, making it less important.
Our static algorithm simply samples edges with probability inversely
proportional to their undirected edge connectivity.
We remark that the cut approximation quality only has a square root dependency
on undirected cut (see \Cref{thm:dyndicut_worst_for}).
This is because our sampling algorithm is unbiased with expectation for a
directed cut precisely equal to its weighted cut value.

Another major reason we choose edge connectivities for our probability measure
is that, compared to undirected edge strengths used in the algorithm by
\cite{CenCPS21}, they are significantly easier to certify.
We can easily certify multiple edge-disjoint paths between incident vertices
using a bundle of approximate maximum weighted spanning forests (MSF).
This allows us to directly apply the dynamic sparsification framework from
\cite{AbrahamDKKP16}, which maintains such bundle of approximate MSFs with low
worst-case update time and recourse using fully dynamic MSF algorithms from
\cite{KapronKM13,GibbKKT15}.

For a directed edge in an expander, its undirected connectivity simply equals
the smaller degree of its incident vertices multiplied by the expansion factor
$\phi$.
This enables us to sample our directed cut sparsifier using the same algorithm
as directed spectral sparsification without requiring any patching.
This approach yields slightly better update time under amortization using the
framework from \cite{BernsteinvdBPGNSSS22}.

\paragraph{Directed spectral sparsification against an adaptive adversary}
As previously noted, directed spectral sparsification against an adaptive
adversary presents significant challenges.
Our approach instead utilizes the sparsification quadruple construction from
\cite{KyngSPG22} to address this task.
Specifically, 
\begin{enumerate}
    \item $\vG_1 = \vG_0^{(\beta)}$ is the $\beta$-partial symmetrization of
        $\vG_0$,
    \item $\vG_2 = \beta \cdot G \cup \vH_2$ is a sparsification of the directed
        portion $\vG$ of $\vG_1$,
    \item $\vG_3 = \beta \cdot H_3 \cup \vH_2$, the final sparsifier, further
        sparsifies the undirected portion of $\beta \cdot G$ of $\vG_2$.
\end{enumerate}

For adaptive sparsification of the undirected portion (i.e., from $\vG_3$ to
$\vG_2$), we again adopt a framework from \cite{BernsteinvdBPGNSSS22} to reduce
the problem to undirected spectral sparsification on decremental expanders with
almost uniform degrees.
By observing that a degree-preserving undirected cut sparsifier of an expander
is also a spectral sparsifier \cite{ChuzhoyGLNPS20,BernsteinvdBPGNSSS22}, we can
create a dynamic algorithm on decremental expanders by combining an adaptive cut
sparsifier from \cite{BernsteinvdBPGNSSS22} with the degree-preserving variant
of our star patching scheme described above.

Regarding the directed portion, we demonstrate that when $\beta$ is sufficiently
large, any directed graph that preserves degrees with respect to $\vG$ can serve
as an effective degree-preserving directed spectral sparsifier.
This allows us to directly implement star patchings as sparsifiers.
Since our star patching algorithm is deterministic, we can utilize the previous
reduction (decremental expander without uniform degree requirement) from the
oblivious case for adaptive directed spectral sparsification.
Our final algorithm for maintaining $\vG_3$ simply runs one copy of each
adaptive algorithm described above and combines their sparsifiers (with
appropriate scaling).

One might notice the lack of an adaptive algorithm for balanced directed cut
sparsification.
We should first point out that instead of having two parameters $(\eps,\beta)$ for
the notion of directed cut approximation (see also \Cref{def:baldicuappr_gen}),
one can consolidate them into a single parameter of $\frac{\eps}{\sqrt{\beta}}$.
Similar to spectral sparsification, the state-of-the-art adaptive undirected cut
sparsifier from \cite{BernsteinvdBPGNSSS22} requires the multiplicative
approximation factor to be at least $O(\log n)$.
Then, if we consider the larger direction of a directed cut that carries at
least half of the weight of its corresponding undirected cut, we can only
guarantee that $\frac{\eps}{\sqrt{\beta}} = O(\log n)$.
Otherwise, we could improve the undirected sparsification.
This tradeoff between $\eps$ and $\beta$ is not particularly appealing to us.

\subsection{Related works}

\paragraph{Static directed graph sparsification}
The notion of Eulerian Laplacian spectral sparsification was introduced by
\cite{CohenKPPRSV17} to construct fast directed Eulerian Laplacian linear system solvers.
These solvers have applications in various computational problems related to
random walks on general directed graphs, through reductions given in
\cite{CohenKPPSV16,AhmadinejadJSS19}.
Cohen et al.~\cite{CohenKPPRSV17} gave the first $\tO(m)$ time algorithm for
constructing Eulerian sparsifiers with $\tO(\eps^{-2} n)$ edges, based on
expander decompositions.

As an alternative approach, Chu, Gao, Peng, Sachdeva, Sawlani, and Wang
\cite{ChuGPSSW18} introduced the short cycle decomposition technique and used
it to give a polynomial-time algorithm for computing sparsifiers with
significantly improved sparsity.
This result was later improved by Sachdeva, Thudi, and Zhao \cite{SachdevaTZ24}
using improved analysis of short-cycle sampling and made more efficient using
improved short-cycle decomposition algorithms \cite{LiuSY19,ParterY19}.
They also demonstrated an improved existential result on the sparsity of
Eulerian sparsifiers by leveraging techniques from matrix discrepancy theory
\cite{BansalJM23}.
This discrepancy approach was earlier introduced by Jambulapti, Reis, and Tian
\cite{JambulapatiRT23} for constructing undirected degree-preserving spectral
sparsifiers with optimal sparsity in almost linear time.
This notion of sparsification served as an intermediary between undirected and
directed Eulerian sparsification \cite{ChuGPSSW18}.

Recently, Jambulapati, Sachdeva, Sidford, Tian, and Zhao
\cite{JambulapatiSSTZ25} achieved state-of-the-art Eulerian sparsification
results for both efficient sparsification and existential sparsity using a
combination of techniques: effective resistance decomposition and sampling by
electrical circulations.
When combined with the reduction by Peng and Song \cite{PengS22}, they provided
the state-of-the-art Eulerian Laplacian solver that runs in $\bO(m \log^3 n +
n\log^6 n)$\footnotemark~time.
\footnotetext{We use $\bO(\cdot)$ to hide polyloglog factors in $n$.}
Lau, Wang, and Zhou \cite{LauWZ25} gave a derandomized Eulerian sparsification
algorithm using the ``deterministic discrepancy walk'' approach from
\cite{PesentiV23}.
They extended their sparsification result to satisfy a stronger notion of
directed spectral approximation known as ``singular value (SV) approximation"
\cite{AhmadinejadPPSV23}.

For directed cut sparsification, the notion of balanced cut approximation was
introduced by Cen et al.~\cite{CenCPS21} for sparsification of directed graphs
where each cut has approximately the same weight as the cut in the reverse
direction.
Efficient cut sparsification algorithms have applications in minimum cut
approximations for such balanced graphs \cite{EneMPS16}.
Cen et al.\ achieved a sparsification algorithm for constructing cut sparsifiers
with $O(\epsilon^{-2} \beta n \log n)$ edges for balance factor $\beta$ by
sampling using the inverse undirected edge strength, introduced in
\cite{BenczurK96} for undirected cut sparsification.
This upper bound was later proven to be nearly tight by Chen, Li, Ling, Tai,
Woodruff, and Zhang \cite{ChengLLTWZ24}, who gave a bit complexity lower bound
of $\Omega(\epsilon^{-2} \beta n)$.
\Cref{thm:statdicut_inf} complements the results of \cite{CenCPS21} by showing
that sampling with probability inversely proportional to undirected edge
connectivities, similar to \cite{FunHHP11} for undirected cut sparsification,
achieves the same sparsity guarantee.

Very recently, Goranci, Henzinger, Räcke, and Sricharan \cite{GoranciHRS25}
independently discovered a similar generalization to balanced cut approximation
using the framework of \cite{FunHHP11}.
They proved a stronger result showing that any undirected edge connectivity
measure is suitable for sparsification of balanced directed cuts.
One key difference is that their general result has a slightly worse sparsity
guarantee compared to ours that specifically uses undirected edge connectivity.
Their sparsification algorithm works in conjunction with a dynamic NI index
algorithm to dynamically sparsify residual networks in incremental maxflow
problems.

\paragraph{Dynamic graph sparsification}
Dynamic graph sparsification problems have been studied extensively for
undirected graphs.
There is a long line of work
\cite{AusielloFI06,Elkin11,BaswanaKS12,ForsterG19,BernsteinFH21,BernsteinvdBPGNSSS22,ChenKLPPGS22,vdBrandCKLPPGSS23,ChuzhoyP25}
on maintaining spanners -- graph sparsifiers that approximately preserve
distances -- with various assumptions about updates (amortized or worst-case)
and adversaries (oblivious or adaptive).

Abraham, Durfee, Koutis, Krinninger, and Peng \cite{AbrahamDKKP16} initiated the
study of dynamic undirected spectral and cut sparsifications.
They developed fully dynamic spectral sparsifiers with $\tO(\eps^{-2} n)$ edges
and $\tO(\eps^{-2})$ amortized update time against oblivious adversaries by
adapting the static spectral sparsification algorithm of Koutis and Xu
\cite{KoutisX16}.
Their algorithm maintains a bundle of dynamic spanners \cite{BaswanaKS12} that
can certify edges with small effective resistances, the central value for
undirected spectral sparsification.
By swapping spanners with a bundle of dynamic (approximate) maximum spanning
forests \cite{HolmRW15,KapronKM13,GibbKKT15} that could certify edge
connectivities, they also showed fully dynamic algorithms for cut sparsification
with worst-case update time of $\tO(\eps^{-2})$ against an oblivious adversary.

These results were subsequently strengthened by Bernstein, van den Brand, Probst
Gutenberg, Nanongkai, Saranurak, Sidford, and Sun \cite{BernsteinvdBPGNSSS22} to
either work against adaptive adversaries or have worst-case update time
guarantees.
At the core of these results is a black-box reduction that reduces fully dynamic
sparsification problems to decremental sparsification problems on almost uniform
degree expanders.
A version of their framework serves as the basis of our dynamic algorithms.
To handle probabilistic dependencies introduced by adaptive adversaries,
Bernstein et al.\ developed a technique called \emph{proactive} sampling, where
an edge update incurs multiple re-samplings with increasing intervals.
One downside is that their adaptive algorithms have a significant trade-off
between update time and approximation factor, leading to at least
polylogarithmic approximations for any subpolynomial update time.

Recently, Khanna, Li, and Putterman \cite{KhannaLP25}, and Goranci and Momeni
\cite{GoranciM25} independently discovered fully dynamic hypergraph spectral and
cut sparsification algorithms by adapting the dynamic framework of
\cite{AbrahamDKKP16}.

Dynamic undirected spectral vertex sparsifiers have also been studied, with
applications in dynamic undirected Laplacian solvers \cite{DurfeeGGP19}, dynamic
effective resistance sparsifiers \cite{DurfeeGGP19,ChenGHPS20}, and static flow
problems \cite{ChenGHPS20,GaoLP21,vdBrandGJLLPS22}.

\paragraph{Dynamic expander decomposition} 
In the static setting, undirected expander decomposition algorithms have been
instrumental in various algorithmic applications, most notably the first
nearly-linear time Laplacian solver \cite{SpielmanT04}, the first almost-linear
time directed Laplacian solver \cite{CohenKPPRSV17}, and the first almost-linear
time algorithm for maxflow or min-cost flow in directed graphs
\cite{ChenKLPPGS22}.
In dynamic settings, there has been significant amount of works on generating
different variants of expander decompositions — standard undirected
\cite{NanongkaiSWN17,SaranurakW19}, almost uniform degrees
\cite{BernsteinvdBPGNSSS22}, boundary-linked \cite{GoranciRST21}, and directed
\cite{BernsteinPGS20,HuaKPGW23,SulserPG25}.
These underlie a wide range of recent advancements in dynamic algorithms such as
dynamic connectivity \cite{NanongkaiSWN17,ChuzhoyGLNPS20,JinS22}, single-source
shortest paths \cite{ChuzhoyK19,BernsteinPGS20,BernsteinPGS22}, and flow and cut
problems \cite{BernsteinvdBPGNSSS22,GoranciRST21,JinST24}.

\section{Preliminaries} \label{sec:prelims}

\paragraph{General notation.} All $\log$ are base $2$ unless otherwise
specified.
We denote $\log_e$ by $\ln$.
When discussing a graph clear from context with $n$ vertices and edge weight
ratio bounded by $W$. 
We use the $\bO$ notation to hide $\polyloglog(nW)$ factors for brevity (in
runtimes only).
We let $[n] \defeq \{i \in \N \mid 1 \le i \le n\}$.

\paragraph{Vectors and Matrices.}
Vectors are denoted in lower-case boldface.
 $\vzero_d$ and $\vone_d$ are the all-zeroes and all-ones vector respectively of dimension $d$.
 $\ee_i$ denote the $i^{\text{th}}$ basis vector.
  $\uu \circ \vv$ denotes the entrywise product of $\uu, \vv$ of equal dimension.

Matrices are denoted in upper-case boldface.
We refer to the $i^{\text{th}}$ row and $j^{\text{th}}$ column of matrix $\MM$
by $\MM_{i:}$ and $\MM_{:j}$ respectively.
We use $[\vv]_i$ to index into the $i^{\text{th}}$ coordinate of vector $\vv$,
and let $[\MM]_{i:} \defeq \MM_{i:}$, $[\MM]_{:j} \defeq \MM_{:j}$, and
$[\MM]_{ij} \defeq \MM_{ij}$ in contexts where $\vv$, $\MM$ have subscripts.

$\id_d$ is the $d \times d$ identity matrix.
For $\vv \in \R^d$, $\diag{\vv}$ denotes the associated diagonal $d \times d$
matrix.
For linear subspace $S$ of $\R^d$, $\dim(S)$ is its dimension and $\PP_S$ is the
orthogonal projection matrix onto $S$.
We let $\ker(\MM)$, $\im(\MM)$, and $\MM^{\dagger}$ denote the kernel, image and
pseudoinverse of $\MM$.
We denote the operator norm (largest singular value) of matrix $\MM$ by
$\normop{\MM}$.
The number of nonzero entries of a matrix $\MM$ (resp.\ vector $\vv$) is denoted
$\nnz(\MM)$ (resp.\ $\nnz(\vv)$), and the subset of indices with nonzero entries
is $\supp(\MM)$ (resp.\ $\supp(\vv)$).

We use $\preceq$ to denote the Loewner partial order on $\Sym^d$, the symmetric
$d \times d$ matrices.
For $\MM \in \Sym^d$ and $i \in [d]$, we let $\lam_i(\MM)$ denote the
$i^{\text{th}}$ smallest eigenvalue of $\MM$, so $\lam_1(\MM) \le \lam_2(\MM)
\le \ldots \le \lam_d(\MM)$.
For positive semidefinite $\AA \in \Sym^d$, we define the seminorm induced by
$\AA$ by $\|\xx\|_{\AA}^2 \defeq \xx^\top \AA \xx$.

\paragraph{Graphs.} 
All graphs throughout this paper are assumed to be simple without loss of
generality, as collapsing parallel multi-edges does not affect (undirected or
directed) graph Laplacians.
Moreover, a dynamic update of adding a parallel edge can be simulated by one
edge deletion and one edge insertion with combined weight.
We denote undirected weighted graphs without an arrow and directed weighted
graphs with an arrow, i.e., $G = (V, E, \ww)$ is an undirected graph with
vertices $V$, edges $E$, and weights $\ww \in \R_{\ge 0}^E$, and $\vec{G}$ is a
directed graph.
A directed Eulerian graph is a directed graph where weighted in-degree equals weighted out-degree for every vertex.
We refer to the vertex set and edge set of a graph $G$ (resp.\ $\vec{G}$) by $V(G)$ and $E(G)$ (resp.\ $V(\vec{G})$ and $E(\vec{G})$).
We associate a directed edge $e$ from $u$ to $v$ with the tuple $(u, v)$, and an undirected edge with $(u, v)$ and $(v, u)$ interchangeably.
We define $h(e) = u$ and $t(e) = v$ to be the head and tail of a directed edge $e = (u, v)$.
We let $\rev(\vG)$ denote the directed graph with all edge orientations reversed from $\vG$, and $\und(\vG)$ denote the undirected graph which removes orientations (both keeping the same weights).

We say $H$ is a subgraph of $G$ if the edges and vertices of $H$ are subsets of
the edges and vertices of $G$ (with the same weights), denoting $H = G_{F}$ if
$E(H) = F$, and defining the same notion for directed graphs.
For $U \subseteq V$, we let $G[U]$ denote the induced subgraph of $G$ on $U$
(i.e., keeping all of the edges within $U$).
When $V$ is a set of vertices, we say $\{V_i\}_{i \in [I]}$ is a partition of
$V$ if $\bigcup_{i \in [I]} V_i = V$, and all $V_i$ are disjoint.
We say $\{G_j\}_{j \in [J]}$ are a family of edge-disjoint subgraphs of $G = (V,
E, \ww)$ if all $E(G_j)$ are disjoint, and for all $j \in [J]$, $V(G_j)
\subseteq V$, $E(G_j) \subseteq E$, and every edge weight in $G_j$ is the same
as its weight in $G$.

\paragraph{Graph matrices.} 
For a graph with edges $E$ and vertices $V$, we let $\BB \in \{-1, 0, 1\}^{E
\times V}$ be its edge-vertex transfer matrix, so that when $\vG$ is directed
and $e = (u, v)$, $\BB_{e:}$ is $2$-sparse with $\BB_{eu} = 1$, $\BB_{ev} = -1$
(for undirected graphs, we fix an arbitrary consistent orientation).
For $u, v \in V$, we define $\bb_{(u, v)} \defeq \ee_u - \ee_v$.
When $\BB$ is the transfer matrix associated with graph $G = (V, E, \ww)$
(resp.\ $\vec{G}$), we say $\xx$ is a circulation in $G$ (resp.\ $\vec{G}$) if
$\BB^\top \xx = \vzero_V$; when $G$ (resp.\ $\vec{G}$) is clear we simply say
$\xx$ is a circulation.
We let $\HH, \TT \in \{0, 1\}^{E \times V}$ indicate the heads and tails of each
edge, i.e., have one nonzero entry per row indicating the relevant head or tail
vertex for each edge, respectively, so that $\BB = \HH - \TT$.
When clear from context that $\ww$ are edge weights, we let $\WW \defeq
\diag{\ww}$.
For undirected $G = (V, E, \ww)$ with transfer matrix $\BB$, the Laplacian
 matrix of $G$ is $\LL \defeq \BB^\top \WW \BB$.
For directed $\vec{G} = (V, E, \ww)$, the directed Laplacian matrix of $\vec{G}$
is $\vec{\LL} \defeq \BB^\top \WW \HH$.
To disambiguate, we use $\LL_G$, $\HH_G$, $\TT_G$, $\BB_G$, etc.\ to denote
matrices associated with a graph $G$ when convenient.

Note that $\vec{\LL}^\top \vone_V = \vzero_V$ for any directed Laplacian $\vec{\LL}$.
If $\vG$ is Eulerian, then its directed Laplacian also satisfies $\vLL \vone_V =
\vzero_V$ and $\ww$ is a circulation in $\vG$ (i.e., $\BB^\top \ww =
\vzero_V$).
Note that for a directed graph $\vG = (V,E,\ww)$ and its corresponding
undirected graph $G \defeq \und(\vG)$, the undirected Laplacian is $\LL_G =
\BB^\top \WW \BB$, and the reversed directed Laplacian is $\vLL_{\rev(\vG)} =
-\BB^\top \WW \TT$.

We let $\PPi_V$ denote the Laplacian of the unweighted complete graph on $V$,
i.e., $\PPi_V \defeq \id_V - \frac 1 {|V|} \vone_V\vone_V^\top$.
Note that $\PPi_V$ is the orthogonal projection on the the subspace spanned by
the vector that is $1$ in the coordinates of $V$ and $0$ elsewhere.

For a subset of vertices $C \subseteq V$ and $F = V \setminus C$, the Schur complement of a directed
Laplacian $\vLL$ onto $C$ is defined by $\Sc(\vLL,C) \defeq \vLL_{CC} - \vLL_{CF}
\vLL_{FF}^\dag \vLL_{FC}$, which satisfies the block LU factorization
\begin{equation} \label{eq:lufact}
    \vLL = \begin{pmatrix}
        \vLL_{FF} & \vLL_{FC} \\
        \vLL_{CF} & \vLL_{CC}
    \end{pmatrix}
    =
    \begin{pmatrix}
        \II_F & 0 \\
        \vLL_{CF}\vLL_{FF}^\dag & \II_C
    \end{pmatrix}
    \begin{pmatrix}
        \vLL_{FF} & 0 \\
        0 & \Sc(\vLL,C)
    \end{pmatrix}
    \begin{pmatrix}
        \II_F & \vLL_{FF}^\dag\vLL_{FC} \\
        0 & \II_C
    \end{pmatrix}.
\end{equation}
A Schur complement of a directed graph Laplacian remains a directed graph
Laplacian.
We say that a directed graph $\vH$ is the Schur complement of $\vG$ onto a
subset $C \subseteq V$, denoted by $\vH = \Sc(\vG,C)$, if $\vLL_{\vH} =
\Sc(\vLL,C)$.
The Schur complement of an Eulerian graph is Eulerian.

\paragraph{Cut, Conductance and Expanders.}
Given a directed graph $\vG = (V,E,\ww)$, for a subset $S \subseteq V$ of
vertices, the volume of $S$ is the same as $\Vol_G(S) \defeq \sum_{v \in S}
\sum_{e \in E(G) : e \ni v} \ww(e)$ the volume of $S$ in the corresponding
undirected graph $G = \und(\vG)$.
For another subset $T \subseteq V$, we let $E(S,T)$ be the set of edges
satisfying that for $e \in E(S,T)$ it has $h(e) \in S$ and $t(e) \in T$ and
denote by $\ww(S,T) \defeq \sum_{e \in E(S,T)} \ww(e)$ the total weight of these
edges.
We say that a subgraph $\vC$, often referred to as a subset of edges, is a
directed cut (di-cut) of $\vG$ if there exists some subset $S \subseteq V$
satisfying $S,V \setminus S \ne \emptyset$ such that $E(\vC) = E(S,V\setminus
S)$.
The corresponding (undirected) cut $C = \und(\vC)$ of $\vC$ is defined by $E(C)
= E_G(S,V\setminus S) = E(S,V\setminus S) \cup E(V\setminus S, S)$ the set of
cut edges in $G$.

We say that the %
undirected conductance of subset $S$ is %
\[
    \Phi_{G}(S) = \frac{\ww_G(S,V\setminus S)}{\min\Brace{\vol(S),\vol(V
    \setminus S)}},
\]
where we note that $\ww_G(S,V \setminus S) = \ww(S, V\setminus S) +
\ww(V\setminus S, S)$.
The conductance of an undirected graph $G$, denoted by $\phi(G)$, is the minimum
conductance $\phi_{G}(S)$ of any cut $(S,V\setminus S)$.

\begin{definition}
    For any $\phi \in (0,1)$, %
    we say that an undirected graph $G$ is a $\phi$-expander if $\Phi(G) \ge
    \Phi$.
\end{definition}

\section{A fully dynamic directed spectral sparsifier} \label{sec:spectral}

In this section, we show our main results for dynamically maintaining directed
spectral sparsifiers under the notion of degree balance preserving directed
spectral approximation (\Cref{def:dispecappr}).
We formally state the guarantees of our algorithms in
\Cref{thm:dynspecstar_for,thm:dynspecext_for,thm:dynspecint_for}.

\Cref{thm:dynspecstar_for} is the formalized version of
\Cref{thm:dynspec_star_inf} where we spell our all the polylogarithmic
dependencies.
The explicit sparsifier maintained using this algorithm requires some extra
vertices that are not originally in the graph.
We show in \Cref{lemma:schurprecon} that having extra vertices does not have a
negatively affect on using such sparsifier for preconditioning.
\begin{theorem}[Formal version of \Cref{thm:dynspec_star_inf}]
    \label{thm:dynspecstar_for}
    There exists a fully dynamic algorithm that given a weighted directed graph
    $\vG$ on vertices $V$ with bounded weight ratio $W$ undergoing oblivious
    edge insertions and deletions and an $\eps \in (0,1)$,
    maintains \emph{explicitly} a graph $\vH$ on vertices $V \cup X$ with $X$
    disjoint from $V$ such that with high probability $\Sc(\vH,V)$ is an
    $\eps$-degree balance preserving directed spectral sparsifier of $\vG$.
    The algorithm has preprocessing time $O(m \log^7 n)$ and amortized update time
    $O(\eps^{-2} \log^{13} n)$. 
    The graph $\vH$ has size $O(\eps^{-2} n \log^{11} n \log W)$ and extra
    number of vertices $|X| \le O(n \log^2 n \log W)$.
\end{theorem}

Suppose one requires the absence of extra vertices, \Cref{thm:dynspecext_for}
gives an algorithm that maintains implicitly a sparsifier on the original set of
vertices.
\begin{theorem}%
    \label{thm:dynspecext_for}
    There exists a fully dynamic algorithm for maintaining implicitly an
    $\eps$-degree balance preserving directed spectral sparsifier $\vH$ of a
    weighted directed graph $\vG$ undergoing oblivious edge insertions and
    deletions with high probability for any $\eps \in (0,1)$.
    The algorithm has preprocessing time $O(m \log^7 n)$, amortized update time
    $O(\eps^{-2} \log^{13} n)$. 
    The sparsifier $\vH$ has size $O(\eps^{-2} n \log^{11} n \log W)$, and
    supports edge query in time $O(\log^3 n \log W)$ and graph query in time
    $O(\eps^{-2} n \log^{12} n \log W)$.
\end{theorem}

If one further requires that the sparsifier remains a subgraph,
\Cref{thm:dynspecext_for} gives an algorithm that maintains such subgraph
sparsifier in the expense that the sparsifier only support fast querying of the
entire graph.
\begin{theorem}%
    \label{thm:dynspecint_for}
    There exists a fully dynamic algorithm for maintaining implicitly an
    $\eps$-degree balance preserving directed spectral sparsifier $\vH$ of a
    weighted directed graph $\vG$ undergoing oblivious edge insertions and
    deletions with high probability for any $\eps \in (0,1)$.
    The algorithm has preprocessing time $O(m \log^7 n)$, amortized update time
    $O(\eps^{-2} \log^{22} n)$. 
    The sparsifier $\vH$ is a \emph{reweighted subgraph} of $\vG$ and has size
    $O(\eps^{-2} n \log^{22} n \log W)$, and supports graph query in time
    $O(\eps^{-2} \log^{21} n)$.
\end{theorem}

To prove these theorems, we show two dynamic versions of the sampling and
patching algorithm developed in \cite{CohenKPPRSV17} for static directed
Eulerian spectral sparsification.
The sampling algorithm provides approximation with respect to the weighted in
and out degrees of the graph.
The classic Cheeger's inequality (\Cref{lemma:cheeger}) then turns such
approximation into spectral approximation if the corresponding undirected graph
is an expander.
We adapt a standard dynamic expander decomposition framework from
\cite{BernsteinvdBPGNSSS22} that allows us to focus on devising a decremental
sparsification algorithm for directed graphs with corresponding undirected
expanders.
Each edge update in an expander only requires us to perform resampling for a
small set of edges, giving us our fast update time.

The sampling procedure above leads to undesired small imbalance in the degrees.
A degree-fixing patching, which is a small set of weighted edges, is required
for the degree balance preserving requirement.
We show in \Cref{ssec:dispec_prelim} that such a requirement is necessary.
It is computationally expensive to explicitly maintain the edges for the
patching; a single update can incur a recourse of $\Omega(n)$ for just the patching,
regardless of which type of patching we use.
We instead resort to maintaining patchings implicitly in our data structures and
allowing access to the dynamic sparsifier through query.
There are two types of queries we consider: an \emph{edge query} where given a
pair of vertices $u,v$, we are to return its current weight in the graph; a
\emph{graph query} where we are to return the entire sparsifier in some
explicit representation.

We organize the rest of this section as follows.
We start by giving some preliminaries on our \Cref{def:dispecappr} of directed
spectral sparsification in \Cref{ssec:dispec_prelim}. 
In \Cref{ssec:expander}, we state a few useful properties and computational
results regarding undirected expanders.
In \Cref{ssec:dispec_static}, we present static directed spectral sparsification
algorithms that serves as the bases of our dynamic algorithms.
We recall a standard reduction due to \cite{BernsteinvdBPGNSSS22} from fully
dynamic graph sparsifier to decremental sparsifier on expanders using dynamic
expander decomposition in \Cref{ssec:dynexpdecomp}.
Finally, we provide a dynamic degree balance preserving spectral sparsifier on
decremental expanders in \Cref{ssec:dyndec_spec} and prove our main results of
this section \Cref{thm:dynspecstar_for,thm:dynspecext_for,thm:dynspecint_for}.
Additionally, we provide strictly degree-preserving sparsification algorithms in
\Cref{ssec:degpre}.

\subsection{Directed spectral approximation preliminaries}
\label{ssec:dispec_prelim}
We provide some useful properties of our degree balance preserving directed
spectral sparsification here.
For this subsection, we assume $\vG = (V,E,\ww)$ and $\vG' = (V,E,\ww')$ where
$E = V \times V$ is the complete set of edges.
An edge $e \in E$ is present in $\vG$ if and only if $\ww_e > 0$.
We write $G \defeq \und(\vG), G' \defeq \und(\vG')$, the corresponding
undirected graphs of $\vG,\vG'$.
Let $\BB = \HH - \TT$ be the edge-vertex transfer matrix of $E$.
For a subset $U \subseteq V$, we denote $\BB_U = \HH_U - \TT_U$ the edge vertex
transfer matrix of $E_U \defeq U \times U$, the complete set of directed edges
induced on $U$.
Similarly, we let $\ww_U \defeq [\ww]_{E_U}$.

The following lemma states that \Cref{def:dispecappr} is well-defined linear
algebraically.
We defer the proof to \Cref{app:specproof}.
\begin{lemma} \label{lemma:dispec_ker}
    Suppose $G = \und(\vG)$ is connected under the edge weights $\ww$.
    If $\ww'$ satisfies that $\BB^\top \ww = \BB^\top \ww'$, then
    \[
        \ker(\vLL_{\vG} - \vLL_{\vG'}) \supseteq \ker(\LL_G), \quad
        \ker((\vLL_{\vG} - \vLL_{\vG'})^\top) \supseteq \ker(\LL_G^\top).
    \]
\end{lemma}

We say that $\vG'$ is degree balance preserving w.r.t.~$\vG$ if 
for each undirected connected component $U \subseteq V$ of $G$, it satisfies that $\BB_U
\ww_U = \BB_U \ww_U'$ (i.e., degree balance preserving within each component),
$[\ww']_{U \times (V \setminus U) \cup (V \setminus U) \times U} = \vzero$
(i.e., $\vG'$ cannot have an edge crossing connected components of $G$).
We show the following equivalent definition of \Cref{def:dispecappr} using the
degree preserving condition, justifying our naming for \Cref{def:dispecappr}.
\begin{lemma} \label{lemma:dispec_equiv}
    $\vG'$ is an $\eps$-degree balance preserving directed spectral
    sparsification of $\vG$ (i.e., satisfying \eqref{eq:dispec_def}) if and only
    if $\vG'$ is degree balance preserving w.r.t.~$\vG$ and 
    \begin{equation} \label{eq:dispec_def_alg}
        \normop{\LL_G^{\frac \dagger 2}\Par{\vLL_{\vG} - \vLL_{\vH}}\LL_G^{\frac
        \dagger 2}} \le \eps.
    \end{equation}
\end{lemma}
\begin{proof}
    Assume \eqref{eq:dispec_def} holds true.
    Then \eqref{eq:dispec_def_alg} follows directly.
    We now show that $\vG'$ must be degree balance preserving.
    For any connected component $U$ of $G$, we get by \eqref{eq:dispec_def} that
    for every
    \[
        (\vLL_{\vG} - \vLL_{\vG'}) \vone_U = \vzero_V, \quad 
        (\vLL_{\vG} - \vLL_{\vG'})^\top \vone_U = \vzero_V, \quad 
    \]
    as $\LL_G \vone_U = \LL_G^\top \vone_U = \vzero_V$.
    We get from the second identity
    \begin{align*}
        (\vLL_{\vG} - \vLL_{\vG'})^\top \vone_U
        & = \HH^\top (\WW - \WW') \BB \vone_U
        = \HH^\top [(\ww_{U \times V \setminus U} - \ww'_{U \times V \setminus
        U}) -  (\ww_{V \setminus U \times U} - \ww'_{V \setminus U \times U})]
        \\
        &=
        \HH^\top \ww'_{V \setminus U \times U} - \HH^\top \ww'_{U \times V \setminus U}
        = \vzero_V
    \end{align*}
    where the second last equality follows by $\ww_{U \times V \setminus U} =
    \vzero$ and $\ww_{V \setminus U \times U} = \vzero$.
    Notice that $\HH^\top \ww'_{V \setminus U \times U} \ge \vzero_V$, $\HH^\top
    \ww'_{U \times V \setminus U} \ge \vzero_V$ and $\supp(\HH^\top \HH^\top
    \ww'_{V \setminus U \times U}) \subseteq V \setminus U$, $\supp(\HH^\top
    \ww'_{U \times V \setminus U}) \subseteq U$.
    We must have $\ww'_{V \setminus U \times U} = \vzero$ and $\ww'_{U \times V
    \setminus U} = \vzero$, i.e., no edge of $\vG'$ crosses $U$.
    For the first identity,
    \begin{align*}
        (\vLL_{\vG} - \vLL_{\vG'}) \vone_U
        & = \BB^\top (\WW - \WW') \HH \vone_U
        = \BB^\top (\ww_{U \times V} - \ww'_{U \times V})
        \\
        &= \BB_U^\top (\ww_U - \ww'_U) 
        = \vzero_V,
    \end{align*}
    where the second last equality follows by $\ww_{U \times V \setminus U} =
    \vzero$ and $\ww'_{U \times V \setminus U} = \vzero$, giving us the
    degree balance preserving condition within $U$.
    
    We now prove the reverse direction.
    From \eqref{eq:dispec_def_alg}, we get
    \[
        \forall \xx,\yy \in \R^V, \quad
        \Abs{\xx^\top \PPi (\vLL_{\vG} - \vLL_{\vG'}) \PPi \yy}
        \le \eps \sqrt{\xx^\top \LL_G \xx \cdot \yy^\top \LL_G \yy},
    \]
    where $\PPi = \sum_{\mbox{component } U \subseteq V} \II_U - \frac{1}{|U|}
    \vone_U \vone_U^\top = \sum_{\mbox{component } U \subseteq V} \PPi_U$ is the
    orthogonal projection matrix on the image of $\LL_G$.
    It then suffices to show for all component $U$,
    \[
        (\vLL_{\vG} - \vLL_{\vG'}) \vone_U = \vzero_V, \quad
        (\vLL_{\vG} - \vLL_{\vG'})^\top \vone_U = \vzero_V.
    \]
    Since $\vG'$ is degree balance preserving w.r.t.~$\vG$, we are guaranteed 
    $\BB^\top_U \ww_U = \BB^\top_U \ww_U'$, $\ww_{V\setminus U, U} = \vzero$ and
    $\ww_{U \times V \setminus U} = \vzero$.
    We get the desired identities using the same linear algebraic derivations
    from above.
\end{proof}

We prove in \Cref{lemma:dispec_to_spec} that \Cref{def:dispecappr} generalizes
the notion of undirected spectral sparsification.
The proof requires the following fact from \cite{JambulapatiSSTZ25}.
The proof of \Cref{lemma:dispec_to_spec} is deferred to \Cref{app:specproof}.
\begin{fact} \label{fact:dispec_top}
    If $\xx \in \R^E$ is a circulation, i.e., $\BB^\top \xx = \vzero_V$, then
    \[
        \HH^\top \XX \HH = \TT^\top \XX \TT, \quad
        \BB^\top \XX \HH = -\TT^\top \XX \BB.
    \]
\end{fact}

\begin{lemma} \label{lemma:dispec_to_spec}
    Suppose directed graphs $\vG = (V,E,\ww)$ and $\vG' = (V,E,\ww')$ satisfy
    $\vG'$ is an $\eps$-degree balance preserving directed spectral
    sparsification of $\vG$.
    Then, their corresponding undirected graphs $G= \und(\vG)$ and $G'=
    \und(\vG')$ satisfy
    \[
        (1-2\eps) \LL_G \preceq \LL_{G'} \preceq (1+2\eps) \LL_G.
    \]
\end{lemma}

Recall the \emph{union property} of graph sparsification problems.
Let $\gH(\vG,\eps)$ be the set of $\eps$-approximations of $\vG$ under some
definition of approximation for weighted directed (or undirected) graph.
Then, we say that $\gH$ satisfies union property if for any $G = \bigcup_{i =
1}^k s_i \cdot G_i$ with $s_1, \ldots, s_k \in \R_{\ge 0}$ it satisfies that for
any $\vH_i \in \gH(\vG_i,\eps)$, we have
\[
    \bigcup_{i=1}^k s_i \cdot \vH_i \in \vH(\vG, \eps).
\]
Our notion of directed spectral sparsification \Cref{def:dispecappr} naturally
satisfies the union property.
We defer its proof to \Cref{app:specproof}.
\begin{lemma}[Union property] \label{lemma:dispec_union}
    Suppose directed graph $\vG = \bigcup_{i=1}^k s_i \cdot \vG_i$ for
    some $k$ and that $s_1,\ldots,s_k \in \R_{\ge 0}$. 
    Then, suppose $\vG_i'$ is an $\eps$-degree balance preserving directed
    spectral approximation to $\vG_i$ for every $i \in [k]$, it follows that
    $\bigcup_{i=1}^k s_i \vG_i'$ is an $\eps$-degree balance preserving directed
    spectral approximation to $\vG$.
\end{lemma}

We show also the \emph{contraction property} where contraction of vertices does
not affect the quality of our degree-balance preserving approximation.
\begin{lemma}[Contraction property] \label{lemma:dispec_contra}
    For a directed graph $\vG=(V,E,\ww)$, suppose $\vH$ is an $\eps$-degree
    balance preserving directed spectral approximation to $\vG$.
    Then, for any subset $W \subseteq V$, let $\vG'$,$\vH'$ be the respective
    directed graphs resulting from contracting $W$, then $\vH'$ is an
    $\eps$-degree balance preserving directed spectral approximation to $\vG'$.
\end{lemma}
\begin{proof}
    The operator norms bound follows simply by noticing that vertex contraction
    is equivalent to restricting vectors into a subspace $S$ with the constraints
    that for vector $\xx \in S$, it has $\xx_u = \xx_v$ for all $u,v\in W$.
    Let $w$ be the resulting vertex of the contraction in both $\vG',\vH'$ and
    set $V' = V\setminus W \cup \{w\}$ be the vertex set.
    For any $v \ne w$, its out-degree in $\vH'$ is 
    \begin{align*}
        \sum_{u \in V': u \ne v} [\ww_{\vH'}]_{(v,u)}
        =
        \sum_{u \in V \setminus W} [\ww_{\vH}]_{(v,u)}
        +
        [\ww_{\vH'}]_{(v,w)}
        =
        \sum_{u \in V \setminus W: u \ne v} [\ww_{\vH}]_{(v,u)}
        +
        \sum_{x \in W} [\ww_{\vH}]_{(v,x)}
        =
        [\ddo_{\vH}]_v.%
    \end{align*}
    By the same argument, we also have 
    \[
        [\ddo_{\vG}]_v = [\ddo_{\vG'}]_v,
        \quad
        [\ddi_{\vH'}]_v = [\ddi_{\vH}]_v, \quad [\ddi_{\vG}]_v =
        [\ddi_{\vG'}]_v,
    \]
    which guarantee
    \[
        [\ddo_{\vH'} - \ddi_{\vH'}]_v = [\ddo_{\vH} - \ddi_{\vH}]_v
        = [\ddo_{\vG} - \ddi_{\vG}]_v
        = [\ddo_{\vG'} - \ddi_{\vG'}]_v.
    \]
    All there left to consider is in and out degree of $w$.
    For the out degree, we have
    \begin{align*}
        \ww_{\vH'}(w,V' \setminus w)
        =
        \sum_{u \in V: u \ne w} [\ww_{\vH'}]_{(w,u)}
        =
        \sum_{x \in W}\sum_{u \in V \setminus W} [\ww_{\vH}]_{(x,u)}
        =
        \ww_{\vH}(W,V\setminus W)
    \end{align*}
    Using the same argument, we get
    \[
        \ww_{\vH'}(V'\setminus w,w) = \ww_{\vH}(V\setminus W,W),
        \ww_{\vG'}(w, V'\setminus w) = \ww_{\vG}(W, V\setminus W),
        \ww_{\vG'}(V'\setminus w,w) = \ww_{\vG}(V\setminus W,W).
    \]
    It then suffices to show
    \[
        \ww_{\vH}(W,V\setminus W) - \ww_{\vH}(V\setminus W,W) =
        \ww_{\vG}(W,V\setminus W) - \ww_{\vG}(V\setminus W,W).
    \]
    Notice that
    \begin{align*}
        \ww_{\vG}(W,V\setminus W) - \ww_{\vG}(V\setminus W,W) 
        &= 
        (\sum_{x \in W} [\ddo_{\vG}]_x - \ww_{\vG}(W,W)) - 
        (\sum_{x \in W} [\ddi_{\vG}]_x - \ww_{\vG}(W,W))
        \\
        &=
        \sum_{x \in W} [\ddo_{\vG} - \ddi_{\vG}]_x,
    \end{align*}
    and 
    \[ 
        \ww_{\vH}(W,V\setminus W) - \ww_{\vH}(V\setminus W,W) = \sum_{x \in W}
        [\ddo_{\vH} - \ddi_{\vH}]_x
        =
        \sum_{x \in W} [\ddo_{\vG} - \ddi_{\vG}]_x,
    \]
    we can conclude $\vH'$ is degree balance preserving.
\end{proof}

We now demonstrate in that having extra vertices in a degree balance preserving
sparsifier does not affect most potential applications of it.
In particular, \Cref{lemma:schurprecon} states that as long as the Schur
complement of a graph $\vH$ is a good degree balance preserving sparsifier of an
Eulerian graph $\vG$, then $\vLL_{\vH}^\dag$ serves as a good preconditioner of
$\vLL_{\vG}$.
\begin{lemma} \label{lemma:schurprecon}
    Let $\vG=(V,E,\ww)$ and $\vH=(V \cup X,E_{\vH},\ww_{\vH})$ be directed
    graphs such that the Schur complement $\Sc(\vH,V)$ is a $\eps$-degree
    balance preserving approximation of $\vG$ for $\eps \in (0,\frac{1}{4})$.
    If $\vG$ is Eulerian, then
    \[
        \norm{\PP_V - \PP_V\vLL_{\vH}^\dag \vLL_{\vG}}_{\LL_G \to \LL_G}
        \le 8\eps.
    \]
\end{lemma}

To prove \Cref{lemma:schurprecon}, we recall the following linear algebraic
lemma.
\begin{lemma}[Lemma B.9, \cite{CohenKPPRSV17}] \label{lemma:eulinvbound}
    Suppose $\vG$ is Eulerian and let $G = \und(\vG)$, then
    \[
        \LL_G \pleq 4 \vLL_{\vG} \LL_G^\dag \vLL_{\vG}.
    \]
    Furthermore, for any matrix $\AA$ satisfying $\ker(\AA) = \ker(\AA^\top)
    \supseteq \ker(\vLL_{\vG})$, we have
    \[
        \norm{\AA}_{\LL_G \to \LL_G} \le 2 \normop{\LL_G^{\frac \dag 2}
        \vLL_{\vG} \AA \LL_G^{\frac \dag 2}}.
    \]
\end{lemma}
Note that we have some extra factors in \Cref{lemma:eulinvbound} since we define
$\und(\vG) = \vG \cup \rev(\vG)$ instead of $\frac{1}{2}(\vG \cup \rev(\vG))$.

\begin{proof}[Proof of \Cref{lemma:schurprecon}]
    Assume w.l.o.g.~that $\vG$ is strongly connected.
    Then, we have by degree balance preserving that $\Sc(\vH,V)$ is also
    strongly connected on $V$, and $\Sc(\vLL_{\vH},V)$ has the same left and
    right kernel as $\vLL_{\vG}$, i.e., the span of $\vone_V$.

    We claim that $\PP_V \vLL_{\vH}^\dag \PP_V = \Sc(\vLL_{\vH},V)^\dag$.
    Let $\vS = \Sc(\vH,V)$ and $S = \und(\vS)$.
    We have by \Cref{lemma:dispec_to_spec} that $(1-2\eps) \LL_G \pleq  \LL_S
    \pleq (1+2\eps) \LL_G$.
    Then, for any matrix $\AA$,
    \begin{align*}
        \norm{\AA}_{\LL_G \to \LL_G}
        &=
        \normop{\LL_G^{\frac 1 2} \AA \LL_G^{\frac \dag 2}}
        =
        \normop{\LL_G^{\frac 1 2}\LL_S^{\frac \dag 2} \cdot 
        \LL_S^{\frac 1 2} \AA \LL_S^{\frac \dag 2} \cdot 
        \LL_S^{\frac 1 2} \LL_G^{\frac \dag 2}}
        \le
        \normop{\LL_G^{\frac 1 2}\LL_S^{\frac \dag 2}} \cdot 
        \normop{\LL_S^{\frac 1 2} \AA \LL_S^{\frac \dag 2}} \cdot 
        \normop{\LL_S^{\frac 1 2} \LL_G^{\frac \dag 2}}
        \\
        &\le
        \sqrt{\frac{1+2\eps}{1-2\eps}} \norm{\AA}_{\LL_S \to \LL_S}.
    \end{align*}
    Then, using \Cref{lemma:eulinvbound},
    \begin{align*}
        \norm{\PP_V - \PP_V \vLL_{\vH}^\dag \vLL_{\vG}}_{\LL_G\to
        \LL_G}
        &\le
        \sqrt{\frac{1+2\eps}{1-2\eps}} \cdot
        \norm{\PP_V - \PP_V \vLL_{\vH}^\dag \vLL_{\vG}}_{\LL_S\to \LL_S}
        \\
        &\le
        2\sqrt{\frac{1+2\eps}{1-2\eps}} \cdot
        \normop{\LL_S^{\frac 1 2} \vLL_{\vS} (\PP_V - \PP_V\vLL_{\vH}^\dag \PP_V
        \vLL_{\vG}) \LL_S^{\frac \dag 2}}
        \\
        &=
        2\sqrt{\frac{1+2\eps}{1-2\eps}} \cdot
        \normop{\LL_S^{\frac 1 2} (\vLL_{\vS} - \vLL_{\vG}) \LL_S^{\frac \dag 2}}
        \\
        &\le
        2\eps \frac{(1+2\eps)^{1/2}}{(1-2\eps)^{3/2}}.
    \end{align*}
    When $\eps < \frac{1}{4}$, the upperbound can be loosely simplified too $8\eps$.

    We now prove $\PP_V \vLL_{\vH}^\dag \PP_V = \Sc(\vLL_{\vH},V)^\dag$.
    Notice first that the left and right kernels conditions are satisfied by
    the projections $\PP_V$.
    For simplicity, we omit $\vH$ in the subscripts.
    Since $\vLL\vLL^\dag\vLL = \vLL$, by \eqref{eq:lufact} and the fact that the
    lower and upper triangular matrices are invertible, we have
    \begin{align*}
        \begin{pmatrix}
            \vLL_{XX} & \vLL_{XV}\\
            0 & \Sc(\vLL,V)
        \end{pmatrix}
        \vLL^\dag
        \begin{pmatrix}
            \vLL_{XX} & 0\\
            \vLL_{VX} & \Sc(\vLL,V)
        \end{pmatrix}
        =
        \begin{pmatrix}
            \vLL_{XX} & 0\\
            0 & \Sc(\vLL,V)
        \end{pmatrix}.
    \end{align*}
    Applying $\II_V$ to both sides and using the fact that $\Sc(\vLL,V) =
    \PP_V\Sc(\vLL,V)\PP_V$, we get
    \[
        \Sc(\vLL,V) \PP_V \vLL^\dag \PP_V \Sc(\vLL,V) 
        = \Sc(\vLL,V) \II_V \vLL^\dag \II_V \Sc(\vLL,V) = 
        \Sc(\vLL,V).
    \]
    Suppose $\PP_V \vLL^\dag \PP_V \ne \Sc(\vLL,V)^\dag$, then there is some
    non-trivial matrix $\ZZ$ such that $\PP_V \vLL^\dag \PP_V = \Sc(\vLL,V)^\dag
    + \ZZ$ and $\PP_V \ZZ \PP_V = \ZZ$.
    Then, we must have $\Sc(\vLL,V)\ZZ \Sc(\vLL,V) = 0$.
    Applying $\Sc(\vLL,V)^\dag$ on both sides and using that fact that
    $\Sc(\vH,V)$ is Eulerian and strongly connected, we get a contradiction by
    \[
        \Sc(\vLL,V)^\dag \Sc(\vLL,V) \ZZ \Sc(\vLL,V) \Sc(\vLL,V)^\dag 
        =
        \PP_V \ZZ \PP_V = \ZZ = 0.
    \]
    Thus, we have proven our claim.
\end{proof}

We also give a basic helper result which we use to ensure
degree-preserving properties, a strictly stronger property of degree balance
preserving properties, of our algorithms by working with bipartite lifts.
This primitive is from Section 9.1 of \cite{JambulapatiSSTZ25}.
Given a directed graph $\vG = (V,E,\ww)$, we let the directed graph
$\vG^\uparrow \defeq \blift(\vG)$ be its bipartite lift, which is defined so
that $V_{\vG^\uparrow} = V \cup \bar{V}$ where $\bar{V}$ is a copy of $V$, and
$E_{\vG^\uparrow} = \{f = (u,\bar{v}) | (u,v) \in E\}$ with $\ww_{(u,\bar{v})} =
\ww_{(u,v)}$.
Notice that our definition gives a canonical bijection between
$E_{\vG^\uparrow}$ and $E_{\vG}$.

\begin{lemma} \label{lemma:blift_spectral}
Let $\vG = (V,E,\ww)$ be a directed graph and let its bipartite lift be
$\vG^\uparrow \defeq (V\cup V',E^\uparrow,\ww) = \blift(\vG)$, with 
$G \defeq \und(\vG)$ and $G^\uparrow \defeq \und(\vG^\uparrow)$. 
Suppose that for some $\eps > 0$, $\ww' \in \R^E_{>0}$ satisfies
\[
    \BB_{\vG^\uparrow}^\top \ww' = \BB_{\vG^\uparrow}^\top \ww, \quad
    \normop{\LL_{G^\uparrow}^{\frac \dagger 2} \BB_{\vG^\uparrow}^\top (\WW' - \WW) 
    \HH_{\vG^\uparrow} \LL_{G^\uparrow}^{\frac \dagger 2}} 
    \le \eps.
\]
Then, letting $|\BB_{\vG}|$ apply the absolute value entry-wise,
\[
    \BB_{\vG}^\top \ww' = \BB_{\vG}^\top \ww,\quad |\BB_{\vG}|^\top \ww' = |\BB_{\vG}|^\top \ww, \quad
    \normop{\LL_G^{\frac \dagger 2} \BB_{\vG}^\top (\WW' - \WW) \HH_{\vG} 
    \LL_G^{\frac \dagger 2}} 
    \le \eps.
\]
\end{lemma}
\Cref{lemma:blift_spectral} entails that working with bipartite lifts allows us
to construct strictly degree-preserving directed spectral sparsifiers.
Working with bipartite lifts also has benefits in computing degree fixing
patchings for our randomized sampling sparsification algorithm, as we will show
in \Cref{ssec:dispec_static}.
We provide a proof of this lemma in \Cref{app:specproof}.
For the rest of this section, we almost exclusively considers such directed
bipartite graphs.

We show in the next \Cref{lemma:blift_schur} that adding extra vertices and
taking the Schur complement onto the original vertices under bipartite lift does
not affect approximation guarantees on the unlifted graph.
\begin{lemma} \label{lemma:blift_schur}
Let $\vG = (V,E,\ww)$ be a directed graph and let its bipartite lift be
$\vG^\uparrow \defeq (V\cup V',E^\uparrow,\ww) = \blift(\vG)$, with 
$G \defeq \und(\vG)$ and $G^\uparrow \defeq \und(\vG^\uparrow)$. 
Suppose $\vH = (V \cup V' \cup X, E_{\vH},\ww_{\vH})$ is a directed graph such
that $\Sc(\vH,V \cup V')$ is an $\eps$-degree-preserving spectral approximation of
$\vG^\uparrow$.
Let $\vH^\downarrow$ be the resulting directed graph of $\vH$ with each pair of
corresponding vertices $\{v,v'\}$ contracted.
Then, $\Sc(\vH^\downarrow, V)$ is an $\eps$-degree-preserving spectral
approximation of $\vG$.
\end{lemma}
\begin{proof}
    We claim that $\Sc(\vH^\downarrow,V)$ is
    equivalent to the graph resulted from contracting all pairs of $\{v,v'\}$ on
    $\Sc(\vH,V \cup V')$.
    Suppose the claim is true, then \Cref{lemma:dispec_contra} ensures our
    spectral approximation guarantees.

    To prove the claim, we show the equivalence linear algebraically.
    The crux of the proof is the fact that these vertex contractions are, in
    fact, equivalent to taking Schur complements onto a linear subspace.
    Specifically, the subspace $S \subseteq \R^{V\cup V' \cup X}$ is defined by
    the matrix $\ZZ = \frac{1}{\sqrt{2}}(\vone_{\{v,v'\}})_{v \in V}$, i.e.,
    $\xx \in S$ if $\ZZ \xx = \xx$.
    Note that the columns of $\ZZ$ are independent and have unit $\ell_2$ norms
    and $\ZZ$ is non-zero only on columns indexed by $V$.
    Then, we get as required
    \begin{align*}
        \Sc(\vLL_{\vH^\downarrow},V) 
        &= \ZZ^\top \Sc\Par{\Sc(\vLL_{\vH},S),V \cup V'} \ZZ
        = \ZZ^\top \Sc(\Sc(\vLL_{\vH},V \cup V'),S) \ZZ
        \\
        &= \Sc(\Sc(\vLL_{\vH},V \cup V'),S).
    \end{align*}
\end{proof}

The next lemma is particularly useful for turning approximations in degrees into
a spectral approximation.
\begin{lemma}[\cite{CohenKPPRSV17}] \label{lemma:optoinf}
    For matrix $\MM \in \R^{n_1 \times n_2}$ and diagonal matrices $\DD_1 \in
    \R_{\ge 0}^{n_1 \times n_1}, \DD_2 \in \R_{\ge 0}^{n_2 \times n_2}$, we have
    \[
        \normop{\DD_1^{\frac \dag 2} \MM \DD_2^{\frac \dag 2}} \le 
        \max\Brace{\norm{\DD_1^\dag \MM}_\infty, \norm{\DD_2^\dag \MM^\top}_\infty}.
    \]
\end{lemma}

\subsection{Expanders and expander decompositions} \label{ssec:expander}
Most of our computations are performed on expanders.
In this subsection, we provide some useful properties of undirected expanders
and state the static expander decomposition results for our static
sparsification algorithms.

We state the well-known Cheeger's inequality that gives a spectral property of
expanders.
\begin{lemma}[Cheeger's inequality] \label{lemma:cheeger}
If $G$ is a $\phi$-expander with Laplacian $\LL_G$ and degrees $\DD_G$, then
\[
    \lambda_2\Par{\DD_G^{\frac \dagger 2} \LL_G \DD_G^{\frac \dagger 2}} 
    \geq \frac {\phi^2} 2.
\]
\end{lemma}

One key useful property of expanders is their strong routing guarantees.
\begin{lemma}[$\ell_\infty$ expander routing] \label{lemma:route_exp}
    If undirected weighted graph $G=(V,E,\ww)$ is a $\phi$-expander, then for
    any demand $\dd \in \R^V$ satisfying $\dd^\top \vone_V = 0$ there exists a
    flow $f \in \R^E$ satisfying $\BB^\top \ff = \dd$, where $\BB$ is the
    edge-vertex transfer matrix of $G$, such that
    \[
        \norm{\WW^{-1} \ff}_\infty \le \phi^{-1},
    \]
    where $\WW = \diag{\ww}$.
\end{lemma}
For a vector $\dd \in \R^V$ defined on a set of vertices $V$, we say that it is
a demand vector if $\dd^\top \vone_V =0$.
The following \Cref{lemma:eroute_cong} gives a simple way to construct a
relatively good routing on expanders.
\begin{lemma}[Electrical routing, \cite{FlorescuKPGS24}]
    \label{lemma:eroute_cong}
    For a $\phi$-expander graph $G=(V,E,\ww)$ with edge-vertex transfer matrix
    $\BB$ and Laplacian $\LL$, the electrical routing $\WW \BB \LL^\dag$
    has competitive ratio at most $3\cdot \frac{\log(2m)}{\phi}$.
    That is, for any demand vector $\dd \in \R^n$, it satisfies that
    \[
        \norm{\BB \LL^\dag \dd}_\infty \le 
        \frac{3\log(2m)}{\phi} \norm{\WW^{-1} \ff}_\infty
    \]
    for any $\ff \in \R^E$ satisfying $\BB^\top \ff = \dd$.
\end{lemma}

One particularly useful tool in the design of fast combinatorial or spectral
graph algorithms is expander decomposition.
Note that for the purpose of this paper, we assume that the subgraphs in the
decomposition contain the \emph{entire set of edges}. 
\begin{definition}[Expander decomposition]\label{def:exp_partition}
We call $(\{G_i\}_{i \in [I]})$ a 
\emph{$(\phi, J)$-expander decomposition}
if $\{G_i\}_{i \in [I]}$ are edge-disjoint subgraphs of $G = (V, E, \ww)$
satisfying $\bigcup_i G_i = G$, and the following hold.
\begin{enumerate}
    \item \label{item:exp:partition:weight} \emph{Bounded weight ratio}: For all
        $i \in [I]$, $\frac{\max_{e \in E(G_i)}\ww_e}{\min_{e \in E(G_i)}\ww_e}
        \le 2$.
    \item \label{item:exp:partition:phi} \emph{Conductance}: For all $i \in [I]$, $\Phi(G_i) \ge \phi$.
    \item \label{item:exp:partition:vertex} \emph{Vertex coverage}: There is a
        $J$-partition $\bigcup_{j \in [J]} S_j = [I]$ such that every vertex $v
        \in V(G)$ appears at most once in the subgraphs of each $S_j$.
\end{enumerate}
\end{definition}

We consider the following expander decomposition scheme using the
state-of-the-art algorithm in \cite{AgassyDK23}.
\begin{proposition}\label{prop:ex_partition}
    There is an algorithm $\EPAalgo(G,\delta)$ that, given as input an undirected
    $G = (V,E,\ww)$ with $\frac{\max_{e\in \supp(\ww)}{\ww_e}}{\min_{e\in
    \supp(\ww)}{\ww_e}} \le W$, $r \geq 1$, and $\delta \in (0,1)$, computes in
    time
    \[
        O\Par{m \log^6(n)\log\Par{\frac n \delta}}
    \]
    a $(\cadk\log^{-2}(n), (\log m +1) (\log W + 3))$-expander decomposition of $G$ with
    probability $\ge 1-\delta$, for a universal constant $\cadk$.
    \footnote{The algorithm of \cite{AgassyDK23} is stated for $n^{-O(1)}$
        failure probabilities, but examining Section 5.3, the only place where
        randomness is used in the argument, shows that we can obtain failure
        probability $\delta$ at the stated overhead.
        The vertex coverage parameter $\log_r(W) + 3$ is due to bucketing the
        edges by weight.
        We note that there is no $\log_r(W)$ overhead in the runtime, as the
    edges in each piece are disjoint.}
\end{proposition}

\subsection{A simple static spectral sparsifier} \label{ssec:dispec_static}
We consider arguably the simplest directed spectral sparsification algorithm due
to \cite{CohenKPPRSV17}.
Their algorithm is based on a simple independent edge sampling scheme
\Cref{lemma:entrysample} within disjoint subgraphs that are expanders.
Such independent sampling guarantees good spectral approximation as well as good
approximation of the weighted degrees.
The algorithm then deterministically computes an external graph called a
patching that fixes the degrees.
Since it serves as the basis of our fully dynamic algorithm, we
present the algorithm as well as an analysis in this section.

\begin{algorithm2e}[t!] \label{alg:outersparsify_static}
    \caption{$\SDSalgo(\vG=(V,E,\ww),\eps,\delta,(\PAalgo,\eta))$}
    \codeInput $\vG=(V,E,\ww)$ a simple weighted directed bipartite graph with
    bipartition $C \cup R = V$, edges $E \subseteq C \times R$, $\ww_e \in
    [1,W]$ for all $e \in E$ and $G \defeq \und(\vG)$, $\eps,\delta \in (0,1)$,
    \PA a patching algorithm with congestion factor $\eta(\phi) \ge 1$ that is a
    function on $\phi$\;
    $\{G_i\}_I \gets \EPAalgo(G,\frac{\delta}{2})$, a $(\cadk\log^{-2}(n), 2\log
    n (\log W + 3))$-expander decomposition of $G$, and let $\{\vG_i\}_I$ be the
    corresponding decomposition of $\vG$\;
    \ForEach{$\vG_i$}{
         $\vH_i \gets \SSalgo(\vG_i, \eps, \frac{\delta}{2m},
         (\PAalgo, \eta(\cadk \log^{-2}(n))))$\; 
     }
     \Return{$\vH \gets \bigcup_{i\in [I]} \vH_i$}\; 
\end{algorithm2e}

\begin{theorem} \label{thm:dispec_ext_static}
    \SDS (Algorithm~\ref{alg:outersparsify_static}) satisfies that
    given $\vG,\eps,\delta$ as in the input conditions and a patching algorithm
    \PAE with $\eta = 1$, it returns a simple weighted directed graph $\vH$ such
    that with probability at least $1-\delta$, $\vH$ is a
    $\eps$-degree-preserving directed spectral sparsifier of $\vG$ with $m_{\vH}
    = O(\eps^{-2} n \log^9(n) \log(W) \log(\frac{n}{\delta}))$.
    The algorithm runs in time $O\Par{m \log^6(n)\log\Par{\frac n \delta} +
    m\log(n)\log(W)}$.
\end{theorem}

One drawback of their degree fixing patching is that it may use edges that
are not in the graph originally.
We present an alternative patching scheme that fixes the degrees using only
sampled edges.
This naturally ensures the final sparsifier remains a subgraph.
\begin{theorem} \label{thm:dispec_int_static}
    \SDS (Algorithm~\ref{alg:outersparsify_static}) satisfies that
    given $\vG,\eps,\delta$ as in the input conditions and a patching algorithm
    \PAI with $\eta(\phi) = 200\phi^{-2}\log(2n)$, it returns a simple weighted
    directed graph $\vH$ such that with probability at least $1-\delta$, $\vH$
    is an $\eps$-degree-preserving directed spectral sparsifier of $\vG$ with
    $m_{\vH} = O(\eps^{-2} n \log^{19}(n)\log(W) \log(\frac{n}{\delta}))$.
    The algorithm runs in time
    \[
        \bO\Par{m \log^6(n)\log\Par{\frac n \delta} + m\log(n)\log(W) + 
        \eps^{-2}n\log^{19}(n)\log(W)\log^2(\frac{n}{\delta})}.
    \]
\end{theorem}

For our explicit dynamic directed spectral sparsifier, we fix degree imbalances
using directed star graphs with extra vertices.
Intuitively, a star graph corresponds to adding a product-weighted directed
biclique graph on the original vertices with degree imbalances.
\begin{theorem} \label{thm:dispec_star_static}
    \SDS (Algorithm~\ref{alg:outersparsify_static}) satisfies that given
    $\vG,\eps,\delta$ as in the input conditions with $\vG$ defined on vertices
    $V$ and a patching algorithm \PAS with $\eta=1$, it returns a simple
    weighted directed graph $\vH$ on vertices $V \cup X$ satisfying $X$ is an
    independent set and $|X| \le O(n \log n \log W)$ such that with probability
    at least $1-\delta$, $\Sc(\vH,V)$ is an $\eps$-degree-preserving directed
    spectral sparsifier of $\vG$ and $m_{\vH} = O(\eps^{-2} n \log^9(n) \log(W)
    \log(\frac{n}{\delta}))$.
    The algorithm runs in time 
    \[
        O\Par{m \log^6(n)\log\Par{\frac n \delta}\ + m\log(n)\log(W)}.
    \]
\end{theorem}
Note that algorithm \SDS works on bipartite lifted graph using the
connection in \Cref{lemma:blift_spectral}.
While adding vertices breaks the bipartite structure of the graph, it does not
affect our spectral guarantees.
For $\vH$ defined on $V\cup V' \cup X$ with $\Sc(\vH,V\cup V')$ a sparsifier of
$\vG^\uparrow$, let $\vH'$ be the resulting graph of contracting $\{v,v'\}$ for
each pair of corresponding vertices $v \in V, v' \in V'$.
One can observe that the Schur complement $\Sc(\vH',V)$ is equivalent to
directed graph resulted from taking the same contractions on $\Sc(\vH,V \cup
V')$ (see \Cref{lemma:blift_schur}).
\Cref{lemma:dispec_contra} ensures that $\Sc(\vH',V)$ remains a sparsifier of
$\vG$.

The following \Cref{lemma:entrysample} is the key lemma from
\cite{CohenKPPRSV17} that enables our results.
\begin{lemma} \label{lemma:entrysample}
    Let $\AA \in \R_{\ge 0}^{n_1 \times n_2}$ be a matrix with $m \defeq
    \nnz(\AA), n \defeq (n_1 + n_2)$ and parameters $\eps,\delta \in (0,1)$. 
    Let the row and column sums be $\rr = \AA \vone_{d_2}$ and $\cc = \AA^\top
    \vone_{d_1}$ with $\RR \defeq \diag{\rr}$ and $\CC \defeq \diag{\cc}$.
    Consider a random matrix $\tA \in \R_{\ge 0}^{n_1 \times n_2}$ that
    independently for each nonzero $\AA_{ij}$
    \[
        \tA_{ij} = 
        \begin{cases}
            \frac{1}{p_{ij}} \AA_{ij} &\mbox{with probability~} p_{ij} \ge 
            \min{1, \rho \AA_{ij}\Par{\frac{1}{\sqrt{\rr_i \cc_j}}}} \\
            0 &\mbox{otherwise,}
        \end{cases}
    \]
    where $\rho = \css \cdot \eps^{-2} \log \frac{3n}{\delta}$.
    Then, it satisfies that with probability $\ge 1-\delta$, the following holds
    \begin{equation} \label{eq:degappr}
        \normop{\RR^{\frac \dag 2} (\tA - \AA) \CC^{\frac \dag 2}} 
        \le
        \eps, \quad
        \norm{\RR^{\dag} (\tA - \AA) \vone_{n_2}}_\infty
        \le
        \eps, \quad
        \norm{\CC^{\dag} (\tA - \AA)^\top \vone_{n_1}}_\infty
        \le
        \eps.
    \end{equation}
\end{lemma}

\Cref{lemma:vertexflow} is useful for analyzing additional spectral error
introduced by a patching.
\begin{lemma} \label{lemma:vertexflow}
    For a directed bipartite graph $G$ with bipartition $V = C \cup R$ and $E
    \subseteq C \times R$, let $\rr \in \R_{\ge 0}^{R},\cc \in \R_{\ge 0}^{C}$
    be vertex weightings on $R$ and $C$ respectively and let $\BB = \HH - \TT$
    be the edge-vertex transfer matrix of $G$.
    Suppose there exists a flow 
    $\ff \in \R^E$ satisfying
    \[
        \norm{\CC^\dag \HH^\top \ff}_\infty \le \eta, \quad
        \norm{\RR^\dag \TT^\top \ff}_\infty \le \eta,
    \]
    where $\eta \ge 0$, then
    \[
        \normop{\RR^{\frac \dag 2} \TT^\top \FF \HH \CC^{\frac \dag 2}} \le \eta.
    \]
\end{lemma}
\begin{proof}
    Notice that $\HH^\top \ff = \HH^\top \FF \vone_E = \HH^\top \FF \TT \vone_V$
    and $\TT^\top \ff = \TT^\top \FF \vone_E = \TT^\top \FF \vone_V$.
    Then, the two vertex congestion inequalities and \Cref{lemma:optoinf}
    directly gives the desired operator norm bound.
\end{proof}

The following \Cref{lemma:extpatch} shows means to compute a small external
patching that fixes degree imbalances.
This patching algorithm was originally used in \cite{CohenKPPRSV17}.
\begin{lemma} \label{lemma:extpatch}
    For a directed bipartite graph $G$ with bipartition $V = C \cup R$ with $E =
    C \times R$, let $\rr \in \R_{\ge 0}^{R},\cc \in \R_{\ge 0}^{C}$ be
    vertex weightings on $R$ and $C$ respectively and let $\BB = \HH - \TT$ be
    the edge-vertex transfer matrix of $G$.
    Suppose there are demands vectors $\dd_1 \in \R_{\ge}^{R},\dd_2 \in
    \R_{\ge}^{C}$ satisfying $\dd_1 \le \rr, \dd_2 \le \cc$ and $\|\dd_1\|_1 =
    \|\dd_2 \|_1$, then there exists an algorithm that computes
    in time $O(\nnz(\dd_1)+\nnz(\dd_2))$ a non-negative flow 
    $\ff \in \R_{\ge 0}^{E}$ such that
    \[
        \HH^\top \ff = \dd_1, \quad
        \TT^\top \ff = \dd_2, \quad
        \normop{\RR^{\frac \dag 2} \TT^\top \FF \HH \CC^{\frac \dag 2}} \le 1,
    \]
    where $\FF = \diag{\ff}$ and $\nnz(\ff) \le (\nnz(\dd_1) + \nnz(\dd_2))$.
\end{lemma}
\begin{proof}
    We construct $\ff$ as follows by greedily reducing the demands
    $\dd_1,\dd_2$.
    Fix an arbitrary order on the vertices.
    For the earliest unsatisfied demand $[\dd_1]_u > 0$ where $[\dd_1]_w = 0$ for
    all $w < u$ and the earliest unsatisfied demand $[\dd_2]_v > 0$, we set
    $\ff_{(u,v)} = \min\Par{[\dd_1]_u, [\dd_2]_v}$ and reduce both
    $[\dd_1]_u,[\dd_2]_v$ by the same amount.
    We repeat the process until all demands are satisfied.
    The first two equalities are immediate from this construct.
    As $\rr \ge \dd_1$ and $\cc \ge \dd_2$, applying \Cref{lemma:vertexflow}
    then gives the operator norm inequality.
    Finally, the number of nonzero and runtime follow by noticing that at
    least one of $[\dd_1]_u$ and $[\dd_2]_v$ is satisfied each iteration.
\end{proof}

\begin{algorithm2e}[t!] \label{alg:sparsify_static}
    \caption{$\SSalgo(\vG=(V,E,\ww),\eps,\delta,(\PAalgo,\eta))$}
    \codeInput $\vG=(V,E,\ww)$ a simple weighted directed bipartite graph with
    bipartition $C \cup R = V$, edges $E \subseteq C \times R$ and $G \defeq
    \und(\vG)$ a $\phi$-expander, $\eps,\delta \in (0,1)$, \PA a patching
    algorithm with congestion factor $\eta \ge 1$ that depends on $\phi$\;
    $\rho \defeq 400 \css \cdot \eps^{-2} \phi^{-4} \eta^2 \log
    \frac{8n}{\delta}$\;
    Compute $\ddi,\ddo$ the in and out-degree vectors of $\vG$\;
    \ForEach{$e = (u,v) \in E(\vH)$}{
        With probability $p_e = \rho \cdot
        \ww_e \Par{\frac{1}{[\ddo]_u}+\frac{1}{[\ddi]_v}}$, set
        $\ww_e' \gets \frac{1}{p_e} \ww_e$; otherwise $\ww_e' \gets 0$\;
    }
    \Return{$\vH \gets \PAalgo(\vG,\ww',(1+\frac{\eps\phi^2}{16\eta})^{-1},
    \frac{\delta}{2})$}\; 
\end{algorithm2e}

\begin{algorithm2e}[t!] \label{alg:patching_ext}
    \caption{$\PAEalgo(\vG=(V,E,\ww),\ww',\xi)$}
    \codeInput $\vG=(V,E,\ww)$ a simple weighted directed bipartite graph 
    with bipartition $V = C \cup R$,$E \subseteq C \times R$ and edge-vertex
    transfer matrix $\BB = \HH - \TT$, $\ww'$ a reweighting of the edges, $\xi
    \in (0,1]$\;
    Let $\vH$ be a directed graph with the same vertex and edge sets as $\vG$\;
    $\dd_1 = \HH^\top (\ww - \xi\ww')$, $\dd_2 = \TT^\top (\ww - \xi\ww')$\;
    Compute flow vector $\ff \in \R_{\ge 0}^{C \times R}$ for $\dd_1,\dd_2$ by
    \Cref{lemma:extpatch}\;
    For each $(u,v) \in \supp(\ff)$, $E_{\vH} \gets E_{\vH} \cup (u,v)$ and 
    $[\ww_{\vH}]_{(u,v)} \gets \xi \ww'_{(u,v)} + \ff_{(u,v)}$\;
    \Return{$\vH = (V,E_{\vH},\ww_{\vH})$}\;
\end{algorithm2e}

\begin{lemma}[External patching] \label{lemma:extpatch_algo}
    Given $\vG$, $\ww'$ as in the input conditions and $\xi$ satisfying that
    $\xi\ddpo \le \ddo$ and $\xi\ddpi \le \ddi$ for $\ddpi,\ddpo$
    the in and out degrees of $\vG'=(V,E,\ww')$, \PAE 
    (Algorithm~\ref{alg:patching_ext}) returns in time $O(m)$ a simple weighted
    directed graph $\vH$ such that
    \begin{enumerate}
        \item \label{item:extpatch_algo1}
            $\ddo= \ddo_{\vH}$ and $\ddi = \ddi_{\vH}$,
        \item \label{item:extpatch_algo2}
            \[
                \normop{[\DDi]^{\frac \dag 2} (\vLL_{\vH} - \vLL_{\vG'})
                [\DDo]^{\frac \dag 2}}
                \le (1-\xi) + \max\Brace{\norm{\vone_R - [\DDi]^\dag
                \ddpi}_\infty, \norm{\vone_C - [\DDo]^\dag
                \ddpo}_\infty},
            \]
        \item \label{item:extpatch_algo3}
            $m_{\vH} \le \nnz(\ww') + n$.
    \end{enumerate}
\end{lemma}
\begin{proof}
    Let $\BB = \HH - \TT$ be the edge-vertex transfer matrix of $\vG$.
    We append zeros to $\ww$ and $\ww'$ so that they are defined on $E(\vH)$.
    Then, $\ddo = \HH^\top \ww$, $\ddi = \TT^\top \ww$. 
    Degree vectors $\ddpo,\ddo_{\vH},\ddpi,\ddi_{\vH}$ can be
    analogously expressed using $\ww'$ or $\ww_{\vH}$.
    Now, we get \Cref{item:extpatch_algo1} by
    \[
        \ddo_{\vH} = \HH^\top \ww_{\vH} = \HH^\top (\xi\ww'+\ff) = \HH^\top \ww' = \ddpo
    \]
    where the last equality follows by our definition of 
    $\dd_1 = \ddo - \xi\ddpo$ and the first two equalities of
    \Cref{lemma:extpatch}.
    The equality for the in degrees follows analogously.
    The sparsity guarantee \Cref{item:extpatch_algo3} follows directly by
    \Cref{lemma:extpatch}.
    The runtime follows by noting that each step, including the application of
    \Cref{lemma:extpatch}, runs in time at most $O(m)$.

    Consider \Cref{item:extpatch_algo2}.
    Let 
    \[
        \alpha \defeq \max\Brace{\norm{\vone_R - [\DDi]^\dag
        \ddpi}_\infty, \norm{\vone_C - [\DDo]^\dag \ddpo}_\infty}.
    \]
    We have $\dd_1 \le (1-(1-\alpha)\xi)\ddo \le (1-\xi+\alpha) \ddo$ and
    $\dd_2 \le (1-\xi+\alpha) \ddi$.
    The operator norm bound in \Cref{lemma:extpatch} then gives us the desired
    error bound.
\end{proof}

\begin{algorithm2e}[t!] \label{alg:patching_int}
    \caption{$\PAIalgo(\vG=(V,E,\ww),\ww',\xi,\delta)$}
    \codeInput $\vG=(V,E,\ww)$ a simple weighted directed bipartite graph 
    with bipartition $V = C \cup R$,$E \subseteq C \times R$ and edge-vertex
    transfer matrix $\BB = \HH - \TT$, $\ww'$ a reweighting of the edges with
    $\vG' \defeq (V,E,\ww')$, $\xi \in [\frac{1}{2},1]$, $\delta \in (0,1)$\;
    Let $\vH$ be directed graph with the same vertex and edge sets as $\vG$\;
    Let $T$ be a maximum spanning forest of $G \defeq \und(\vG)$\;
    Set $\dd_1,\dd_2 \in \R_{\ge 0}^V$ by $\dd_1 = \HH^\top (\ww - \xi\ww')$ and
    $\dd_2 = \TT^\top (\ww - \xi\ww')$\;
    Compute flow vector $\ff \in \R^E$ for $\dd_1,\dd_2$ by
    \Cref{lemma:eroute_approx} with $\zeta = \frac{1}{5n^2}$\;
    $\yy \gets \ROalgo(\vG',\ff,T)$ and let $\aa \gets \ff + \yy$\;
    For each $(u,v) \in \supp(\ff)$, $E_{\vH} \gets E_{\vH} \cup (u,v)$ and 
    $[\ww_{\vH}]_{(u,v)} \gets \xi \ww'_{(u,v)} + \aa_{(u,v)}$\;
    \Return{$\vH = (V,E_{\vH},\ww_{\vH})$}\;
\end{algorithm2e}

Since having extra edges in a sparsifier is typically undesired, we also present
an algorithm \PAI (Algorithm~\ref{alg:patching_int}) that computes a patching
using sampled edges and electrical routing.
Such a degree fixing patching, which we call an internal patching, guarantees
the sparsifier remains a subgraph of the original input graph.

To compute internal patchings, we need to call the state-of-the-art
undirected Laplacian solver.
\begin{proposition}[Theorem 1.6, \cite{JambulapatiS21}]\label{prop:js21}
	Let $\LL_G$ be the Laplacian of $G = (V, E, \ww)$. There is an algorithm which takes $\LL_G$, $\bb \in \R^V$, and $\delta, \xi \in (0, 1)$, and outputs $\xx$ such that with probability $\ge 1 - \delta$, $\xx$ is an $\xi$-approximate solution to $\LL_G \xx = \bb$, i.e.,
	\[\norm{\xx - \LL_G^\dagger \bb}_{\LL_G} \le \xi \norm{\LL_G^\dagger \bb}_{\LL_G},\]
	in time $\bO(|E| \cdot \log \frac{1}{\delta\xi})$. Moreover, the algorithm returns $\xx = \MM \bb$ where $\MM$ is a random linear operator constructed independently of $\bb$, such that the above guarantee holds with $1 - \delta$ for all $\bb$.
\end{proposition}

We state the guarantees on the edge congestion for approximate electrical
routing using the approximate Laplacian solver above.
\begin{lemma} \label{lemma:eroute_approx}
    There is an algorithm that given as input a $\phi$-expander graph
    $G=(V,E,\ww)$ with $\min_{e \in E} \ww_e \ge \vone_E$, a demand vector $\bb
    \in \R^V$ satisfying $\|\DD^{-1} \bb \|_\infty \le 1$,
    and a parameter $\zeta \in (0,1)$, 
    returns in time $\bO(m \log \frac{M}{\phi \zeta \delta})$ a random flow
    $\ff \in \R^E$ that depends linearly on $\bb$ satisfying,
    \begin{equation} \label{eq:flow_approx}
        \norm{\WW^{-1} \ff}_\infty \le \frac{4\log(2m)}{\phi^2}, \quad
        \min_{\zz : \BB \zz = \dd} \norm{\WW^{-1} (\ff - \zz)}_\infty \le \zeta, 
    \end{equation}
    with probability $\ge 1-\delta$, where $M = \|\ww \|_1$, $\DD =
    \diag{\dd}$ the weighted degree matrix.
\end{lemma}
\begin{proof}
    We will compute $\ff$ by computing an approximate solution $\xx$ of
    $\LL^\dag \bb$ using \Cref{prop:js21} with error parameter $\zeta'$ that we
    derive later.
    Then, $\ff$ can be expressed as $\WW\BB \xx$ for some $\xx \in \R^V$.
    Let $\ffs \defeq \WW \BB^\top \LL^\dag \bb$, then $\bb = \BB \ffs$.
    By the definition of approximate solutions (see \Cref{prop:js21}), we have
    \[
        \norm{\xx - \LL^\dag \bb}_{\LL} 
        \le \zeta' \cdot \norm{\LL^\dag \bb}_{\LL} 
        \le \zeta' \cdot \norm{\DD^{-1/2}\bb}_2 \cdot \sqrt{\norm{\PP
        \DD^{1/2}\LL^\dag \DD^{1/2}\PP}}
        \le \frac{\sqrt{2} \cdot \zeta'}{\phi} \cdot \norm{\DD^{-1/2}\bb}_2 
    \]
    where $\PP_d = \II_V - \frac{1}{\|\dd\|_1}\dd^{1/2} (\dd^{1/2})^\top$ is an
    orthogonal projection matrix and the last inequality follows by
    \Cref{lemma:cheeger}.
    We further have 
    \[
        \norm{\DD^{-1/2}\bb}_2 
        \le \sqrt{\norm{\DD^{-1}\bb}_\infty \cdot \norm{\bb}_1}
        \le \sqrt{2\norm{\ww}_1}
    \]
    where we used the assumption that $\bb \le \dd$ and the fact that $\|\dd\|_1
    = 2\|\ww\|_1 \defeq M$.
    We consider first the difference
    \begin{align*}
        \norm{\WW^{-1} (\ff - \ffs)}_\infty
        &\le \norm{\WW^{1/2}\BB (\xx - \LL^\dag \bb)}_2 \cdot 
        \sqrt{\norm{\ww^{-1}}_\infty}
        \\
        &\le \norm{(\xx - \LL^\dag \bb)}_{\LL}
        \le \zeta' \cdot \norm{\LL^\dag \bb}_{\LL}
        \le \frac{\sqrt{2} \cdot \zeta'}{\phi} \norm{\DD^{-1/2} \bb}_2
        \\
        &\le \frac{2 \sqrt{M} \cdot \zeta'}{\phi}.
    \end{align*}
    It then suffices to have $\zeta' \le \frac{\phi}{2\sqrt{M}} \zeta$ for
    the difference to be at most $\zeta$.
    By \Cref{lemma:route_exp,lemma:eroute_cong}, 
    we have $\norm{\WW^{-1} \ffs}_\infty \le \frac{3\log(2m)}{\phi^2}$.
    Note that this upperbound is larger than $3 > 3\zeta$ for any possible $m$
    and $\phi$.
    By choosing $\zeta' = \frac{\phi \zeta}{2\sqrt{M}}$, the runtime is
    dominated by the Laplacian solving step that runs in $\bO(m \log
    \frac{M}{\phi \zeta \delta})$.
\end{proof}

To deal with the small error in the routing due to approximate Laplacian
solving, we will use a simple algorithm \RO (Algorithm~\ref{alg:round}) for
rounding degrees using a maximum spanning forest tree.
This procedure is related to other tree-based rounding schemes
in the literature, see e.g., \cite{KelnerOSZ13,JambulapatiSSTZ25}.

\begin{algorithm2e}[ht!]\label{alg:round}
\caption{$\ROalgo(\vG, \dd, T)$}
\DontPrintSemicolon
\codeInput $\vG = (V, E, \ww)$, $\dd \in \R^V$, $T$ a maximum spanning forest
of $G \defeq \und(\vG)$ \;
$\yy \gets $ unique vector in $\R^E$ with $\supp(\yy) \subseteq E(T)$ and $\BB_{\vG}^\top \yy = \dd$\;
\Return{$\yy$}
\end{algorithm2e}

\begin{restatable}{lemma}{restaterounding}\label{lemma:rounding}
Given $\vG = (V,E,\ww)$, a maximum spanning forest subgraph $T$ of $G \defeq
\und(\vG)$, \RO (Algorithm~\ref{alg:round}) returns in $O(n)$ time $\yy \in
\R^{E}$ with $\supp(\yy) \subseteq T$ satisfying
\begin{enumerate}
    \item $\BB_{\vG}^\top \yy = \dd$,
    \item for any $\zz \in \R^{E}$ satisfying $\BB_{\vG}^\top \zz = \dd$,
        $\norm{\WW^{-1} \yy}_{\infty} \le \norm{\WW^{-1} \zz}_1$,
    \item further, $\normsop{\LL^{\dag/ 2}_G \BB^{\top}_{\vG}(\YY -
        \ZZ)\HH_{\vG} \LL^{\dag/ 2}_G} \leq n \norm{\WW^{-1} \zz}_1$.
\end{enumerate}
\end{restatable}
For completeness, we include a proof of this lemma in \Cref{app:rounding}.

\begin{lemma}[Internal patching by electrical routing] \label{lemma:intpatch}
    Given $\vG$, $\ww'$, $\xi$, $\delta$ as in the input conditions and
    satisfying that
    \begin{enumerate}[label=(\alph*)]
        \item $\frac{\max_{e \in E} \ww_e}{\min_{e \in E} \ww_e} \le 2$ and
            $\ww_f' \ge \ww_f$ for any $f \in \supp(\ww_f')$,
        \item $G'=\und(\vG')$ is a $\phi$-expander, 
        \item $\xi\ddpo \le \ddo$ and $\xi\ddpi \le \ddi$ for
            $\ddpi,\ddpo$ the in and out degrees of $\vG'$, 
        \item $\|[\DDo]^\dag (\ddo - \xi \ddpo)\|_\infty, \|[\DDi]^\dag (\ddi -
            \xi \ddpi)\|_\infty \le \frac{\phi^2}{100 \log(2n)} \eps$ for some
            $\eps \in (0,1)$,
    \end{enumerate}
    \PAI (Algorithm~\ref{alg:patching_int}) returns in time $\bO(\nnz(\ww') \log
    \frac{n}{\delta})$ reweighted simple directed subgraph $\vH$ of $\vG'$ such
    that with probability at least $1-\delta$
    \begin{enumerate}
        \item \label{item:intpatch_algo1}
            $\ddo =\ddo_{\vH}$ and
            $\ddi =\ddi_{\vH}$,
        \item \label{item:intpatch_algo2}
            $\normop{[\DDi]^{\frac \dag 2} (\vLL_{\vH} - \vLL_{\vG'}) [\DDo]^{\frac
            \dag 2}} \le (2-2\xi+\frac{1}{10}\eps),$
        \item \label{item:intpatch_algo3}
            $m_{\vH} \le \nnz(\ww')$,
            and for every $e \in \supp(\ww')$, 
            $\ww_{\vH}(e) \in (\xi \pm\frac{1}{4}\eps) \ww'(e)$.
    \end{enumerate}
\end{lemma}
\begin{proof}
    We consider first the approximate routing $\ff$ computed by the algorithm in
    \Cref{lemma:eroute_cong}.
    By conditions (c) and (d), both $\norm{[\DDo]^\dag \dd_1}_\infty,
    \norm{[[\DDi]^\dag \dd_2}_\infty \le \frac{\phi^2 \eps}{100 \log(2n)}$.
    We also have that, entry-wise $\xi\ddpo \in (1\pm \frac{\phi^2 \eps}{100 \log(2n)}) \ddo$ and 
    $\xi\ddpi \in (1\pm \frac{\phi^2 \eps}{100 \log(2n)}) \ddi$.
    Then,
    \[
        \norm{(\DDpo)^\dag \dd_1}_\infty, \norm{(\DDpi)^\dag \dd_2}_\infty 
        \le \frac{1}{\xi}(1 - \frac{\phi^2 \eps}{100 \log(2n)})^{-1} \frac{\phi^2 \eps}{100 \log(2n)}
        \le \frac{\phi^2 \eps}{90 \xi \log(2n)},
    \]
    where we note that $\frac{\phi^2}{100 \log(2n)} < \frac{1}{10}$ for any
    valid $\phi$ and $n$.
    We also get by (a) that the reweighting $\ww'$ satisfies  
    \[
        \|\ww'\|_1 
        \le (1+ \frac{\phi^2 \eps}{100 \log(2n)})\frac{1}{\xi} \|\ww \|_1
        \le \frac{6m}{\xi}  \min_{e \in E} \ww_e
        \le \frac{6m}{\xi}  \min_{e \in E'} \ww_e',
    \]
    where $E' = \supp(\ww')$.
    To apply \Cref{lemma:eroute_approx}, we first scale all $\ww'$,
    $\dd_1,\dd_2$ by the $\frac{1}{\min_{e \in E'} \ww'}$.
    Note that this scaling does not change the flow as $\DDpi$ and $\DDpo$
    also scales by the same factor.
    We then scale $\dd_1$ and $\dd_2$ by $\frac{90\xi \log(2n)}{\phi^2 \eps}$
    and rescale the computed flow back by $\frac{\phi^2 \eps}{90\xi \log(2n)}$.
    Now, the final flow $\ff$ has $\supp(\ff) \subseteq E'$ and satisfies
    \[
        \norm{(\WW')^\dag \ff}_\infty \le \frac{\eps}{10\xi} \le \frac{\eps}{5}, \quad
        \min_{\zz : \HH_{\vG'}^\top \zz = \dd_1, \TT_{\vG'}^\top \zz = \dd_2} 
        \norm{(\WW')^\dag (\ff-\zz)}_\infty 
        \le \zeta \frac{\phi^2 \eps}{90\xi \log(2n)}
        \le \frac{\zeta \eps}{4},
    \]
    by the input condition that $\xi \ge \frac{1}{2}$.
    We can now consider the rounding quality.
    By the second inequality above, there exists some flow $\xx$ supported on
    $E'$ satisfying $\HH_{\vG'}^\top \xx = \dd_1$ and $\TT_{\vG'}^\top \xx =
    \dd_2$ such that $\norm{(\WW')^\dag \xx}_1 \le \frac{m' \zeta \eps}{4} \le
    \frac{n^2 \zeta \eps}{4}$.
    Thus, the second guarantee of \Cref{lemma:rounding} and our choice of $\zeta
    = \frac{1}{5n^2}$ gives $\norm{(\WW')^\dag \yy}_\infty \le \frac{n^2 \zeta
    \eps}{4} \le \frac{\eps}{20}$ with $\yy$ is supported on $E'$.
    When combined with the first guarantee of \Cref{lemma:rounding} and the
    bipartite condition of $\vG$, we have
    \begin{gather*}
        \HH^\top \aa = \HH^\top \yy + \HH^\top \ff = \dd_1, \quad
        \TT^\top \aa = \TT^\top \yy + \TT^\top \ff = \dd_2, 
        \\
        \norm{(\WW')^\dag \aa}_\infty 
        \le \norm{(\WW')^\dag \ff}_\infty + \norm{(\WW')^\dag \yy}_\infty
        \le \frac{1}{4}\eps,
    \end{gather*}
    giving us both \Cref{item:intpatch_algo1,item:intpatch_algo3}.
    As for \Cref{item:intpatch_algo2}, notice that
    \begin{align*}
        \normop{[\DDi]^{\frac \dag 2} (\vLL_{\vH} - \vLL_{\vG'}) [\DDo]^{\frac \dag 2}} 
        &=
        \normop{[\DDi]^{\frac \dag 2} \BB^\top ((1-\xi)\WW' - \AA) \HH [\DDo]^{\frac \dag 2}} 
        \\
        &=
        \normop{[\DDi]^{\frac \dag 2} \TT^\top ((1-\xi)\WW' - \AA) \HH [\DDo]^{\frac \dag 2}} 
        \\
        &\le 
        (1-\xi) \normop{[\DDi]^{\frac \dag 2} \TT^\top \WW'\HH [\DDo]^{\frac \dag 2}} 
        + \normop{[\DDi]^{\frac \dag 2} \TT^\top \AA \HH [\DDo]^{\frac \dag 2}}
    \end{align*}
    Since $\ddpo \le \frac{1}{\xi}\ddo \le 2\ddo$ and $\ddpi \le
    \frac{1}{2}\ddi$, \Cref{lemma:optoinf} gives that the first term is at most
    $2(1-\xi)$.
    Similarly, by $\HH^\top \aa = \dd_1 \le \frac{\eps}{10} \ddo$,
    $\TT^\top \aa = \dd_2 \le \frac{\eps}{10} \ddi$, we get from
    \Cref{lemma:optoinf} that the second term is at most $\frac{1}{10}\eps$.

    Finally, the runtime is dominated by the application of
    \Cref{lemma:eroute_approx}, which runs in $\bO(m' \log \frac{n}{\delta})$
    time with $m' = |E'| = \nnz(\ww')$.
\end{proof}

\begin{algorithm2e}[t!] \label{alg:patching_star}
    \caption{$\PASalgo(\vG=(V,E,\ww),\ww',\xi)$}
    \codeInput $\vG=(V,E,\ww)$ a simple weighted directed bipartite graph 
    with bipartition $V = C \cup R$,$E \subseteq C \times R$ and edge-vertex
    transfer matrix $\BB = \HH - \TT$, $\ww'$ a reweighting of the edges, $\xi
    \in (0,1]$\;
    Let $\vH$ be $\vG$ but with one extra vertex $x$.
    $\dd_1 = \HH^\top (\ww - \xi\ww')$, $\dd_2 = \TT^\top (\ww - \xi\ww')$\;
    For each $v \in C$ with $[\dd_1]_v \ne 0$, $E_{\vH} \gets E_{\vH} \cup
    (v,x)$, and $[\ww_{\vH}]_{(v,x)} \gets [\dd_1]_v$\;
    For each $v \in R$ with $[\dd_2]_v \ne 0$, $E_{\vH} \gets E_{\vH} \cup
    (x,v)$, and $[\ww_{\vH}]_{(x,v)} \gets [\dd_2]_v$\;
    \Return{$\vH = (V \cup \{x\}, E_{\vH},\ww_{\vH})$}\;
\end{algorithm2e}

Finally, we consider patching using directed star graphs.
\begin{lemma}[Star patching] \label{lemma:starpatch_algo}
    Given $\vG$, $\ww'$ as in the input conditions and $\xi$ satisfying that
    $\xi\ddpo \le \ddo$ and $\xi\ddpi \le \ddi$ for $\ddpi,\ddpo$
    the in and out degrees of $\vG'=(V,E,\ww')$, \PAS 
    (Algorithm~\ref{alg:patching_star}) returns in time $O(m)$ a simple weighted
    directed graph $\vH$ such that
    \begin{enumerate}
        \item \label{item:starpatch_algo1}
            $\ddo= [\ddo_{\vH}]_V$, $\ddi = [\ddi_{\vH}]_V$, and for the Schur
            complement $\vS = \Sc(\vH,V)$, $\ddo_{\vS} = \ddo_{\vH}$,
            $\ddi_{\vS} = \ddi_{\vH}$,
        \item \label{item:starpatch_algo2}
            \begin{align*}
                &\normop{[\DDi]^{\frac \dag 2} \Par{\Sc(\vLL_{\vH},V) - \vLL_{\vG'}} [\DDo]^{\frac \dag 2}}
                \\
                \le& 
                (1-\xi) + \max\Brace{\norm{\vone_R - [\DDi]^\dag
                \ddpi}_\infty, \norm{\vone_C - [\DDo]^\dag
                \ddpo}_\infty},
            \end{align*}
        \item \label{item:starpatch_algo3}
            $m_{\vH} \le \nnz(\ww') + n$ and $|V_{\vH}| = n + 1$.
    \end{enumerate}
\end{lemma}
\begin{proof}
    It's easy to see by our construction that the in and out degrees of $\vH$ on
    $V$ is the same as the in and out degrees of $\vG$ respectively.
    Notice that the Schur complement $\vS = \Sc(\vH,V)$ is the union of $\xi\vG'$
    and a complete directed bipartite graph.
    Linear algebraically, this is equivalently to
    \[
        \Sc(\vLL,V) = \vLL_{\vS} = \xi \cdot \BB^\top \WW' \HH + 
        \BB^\top \FF \HH,
    \]
    where $\FF = \diag{\ff}$ is defined by $\ff_{(u,v)} = \frac{[\dd_1]_u
    [\dd_2]_v}{\|\dd_1\|_1}$ and $\BB=\HH-\TT$ is the edge-vertex transfer
    matrix of the complete directed bipartite graph from $C$ to $R$.
    Then, the out-degrees of $\vS$ is
    \begin{align*}
        \ddo_{\vS} 
        &= \HH^\top(\xi \ww' + \ff) = 
        \xi \HH^\top \ww' + \frac{1}{\|\dd_2\|_1} \HH^\top \dd_1 \dd_2^\top
        \vone_V
        \\
        &=
        \xi \HH^\top \ww' + \HH^\top \dd_1
        =
        \xi \HH^\top \ww' + \HH^\top (\ww - \xi \ww')
        = \ddo,
    \end{align*}
    where we note that $\|\dd_1\|_1 = \|\dd_2\|_1$.
    An analogous argument gives us $\ddi_{\vS} = \ddi$, completing our proof of
    \Cref{item:starpatch_algo1}.

    Consider \Cref{item:starpatch_algo2}.
    Let 
    \[
        \alpha \defeq \max\Brace{\norm{\vone_R - [\DDi]^\dag
        \ddpi}_\infty, \norm{\vone_C - [\DDo]^\dag \ddpo}_\infty}.
    \]
    We have $\dd_1 \le (1-(1-\alpha)\xi)\ddo \le (1-\xi+\alpha) \ddo$ and
    $\dd_2 \le (1-\xi+\alpha) \ddi$.
    Then, the operator norm bound on the Schur complement $\vLL_{\vS}$ then
    follows directly by \Cref{lemma:vertexflow} using our $\ff$ above.

    The sparsity and vertex number guarantees \Cref{item:extpatch_algo3} follows
    directly from our algorithm.
\end{proof}

With all the lemmas on patching, we can obtain the following three
sparsification guarantees of \SS (Algorithm~\ref{alg:sparsify_static}) using
different patching methods above.

\begin{lemma} \label{lemma:subspar_int}
    \SS (Algorithm~\ref{alg:sparsify_static}) satisfies that
    given $\vG,\eps,\delta$ as in the input conditions and a patching algorithm
    \PAI with $\eta = 200 \phi^{-2}\log(2n)$, 
    it returns $\vH$ a reweighted subgraph of $\vG$ such that with probability
    at least $1-\delta$,
    $\vH$ is an $\eps$-degree-preserving directed spectral sparsifier
    of $\vG$ with $m_{\vH} = O(\eps^{-2} \phi^{-8} n \log^2(n) \log
    (\frac{n}{\delta}))$.
    The algorithm runs in time $O(m + \eps^{-2} \phi^{-8} n \log^2(n)
    \log^2(\frac{n}{\delta}) \polyloglog(n))$.
\end{lemma}
\begin{proof}
    By the sampling algorithm and \Cref{lemma:entrysample}, we have with
    probability $\ge 1-\frac{3\delta}{8}$ that $\vG' = (V,E,\ww')$ satisfies
    \begin{gather*}
        \normop{[\DDi]^{\frac \dag 2} \TT^\top (\WW'-\WW)\HH [\DDo]^{\frac \dag 2}} 
        \le
        \frac{\eps\phi^2}{20\eta},
        \\
        \norm{[\DDi]^{-1} (\ddpi - \ddi)}_\infty
        \le
        \frac{\eps\phi^2}{20\eta}, \quad
        \norm{[\DDo]^{-1} (\ddpo - \ddo)}_\infty
        \le
        \frac{\eps\phi^2}{20\eta}.
    \end{gather*}
    We condition on its success for the rest of the proof.
    Notice that
    \begin{align*}
        \normop{[\DDo]^{\frac \dag 2} \HH^\top (\WW'-\WW) \HH [\DDo]^{\frac \dag 2}}
        &= \normop{[\DDo]^{\frac \dag 2} (\DDpo - \DDo) [\DDo]^{\frac \dag 2}}
        \\
        &= \norm{[\DDo]^\dag (\ddpo - \ddo)}_\infty
        \le \frac{\eps\phi^2}{20\eta},
    \end{align*}
    we then have
    \begin{align*}
        \normop{\DD^{\frac \dag 2} (\vLL_{\vG'} - \vLL_{\vG}) \DD^{\frac \dag 2}} 
        &\le
        \normop{[\DDi]^{\frac \dag 2} \TT^\top (\WW'-\WW)\HH [\DDo]^{\frac \dag 2}} 
        + \normop{[\DDo]^{\frac \dag 2} \HH^\top (\WW'-\WW)\HH [\DDo]^{\frac \dag 2}} 
        \\
        &\le
        \frac{\eps\phi^2}{10\eta},
    \end{align*}
    where $\DD = \diag{\dd}$ for $\dd \in \R^{C\cup R}$ the degrees of $G =
    \und(\vG)$ satisfying $\dd = \ddi + \ddo$ by our bipartite requirments.
    By \Cref{fact:dispec_top},
    \[
        \vLL_{\rev(\vG)} - \vLL_{\rev(\vG')}
        = -\BB^\top (\WW - \WW') \TT  
        = \HH^\top (\WW - \WW') \BB
        = (\vLL_{\vG} - \vLL_{\vG'})^\top,
    \]
    and we get undirected approximation for $G' \defeq \und(\vG')$
    \begin{align*}
        \normop{\DD^{\frac \dag 2} (\LL_{G'} - \LL_{G}) \DD^{\frac \dag 2}} 
        =&
        \normop{\DD^{\frac \dag 2} (\vLL_{\vG'}+\vLL_{\rev(\vG')} -
        \vLL_{\vG}-\vLL_{\rev(\vG)}) \DD^{\frac \dag 2}} 
        \\
        &\le 
        2 \normop{\DD^{\frac \dag 2} (\vLL_{\vG'} - \vLL_{\vG}) \DD^{\frac \dag 2}} 
        \le
        \frac{\eps\phi^2}{5\eta}.
    \end{align*}
    Cheeger's inequality (\Cref{lemma:cheeger}) gives $G'$ is a
    $\frac{\eps}{2\eta}$-spectral approximation of $G$.
    As spectral approximation also guarantees cut approximation, using the
    conditions $\eps < 1$, $\eta \ge 1$ and the fact that $G$ is a
    $\phi$-expander, $G'$ is at least a $\frac{1}{4}\phi$-expander.
    This gives us the condition (b) in \Cref{lemma:intpatch}.
    Condition (a) is guaranteed by our input assumptions and the fact that for
    any edge $e \in E$ that has a probability of being sampled less than 1, it
    must have $\ww_e' > \ww_e$ if sampled.
    Consider conditions (c) and (d) in \Cref{lemma:intpatch}.
    By our approximate out degrees guarantees above and the choice that
    $\xi = (1+\frac{\eps\phi^2}{16\eta})^{-1}$, we get
    \[
        \vzero \le \ddo - \xi \ddpo \le \frac{\eps \phi^2}{8\eta} \ddo
        \le \frac{(\phi/4)^2}{100 \log(2n)} (\phi^2\eps) \cdot \ddo
    \]
    by our choice of $\eta = 200 \phi^{-2} \log(2n)$.
    The same inequality holds for the in degrees as well.
    With all the conditions, \Cref{item:intpatch_algo1} first ensures that $\vH$
    is degree-preserving.
    Note that $G$ must be connected, as it is an expander.
    We get from our choice of $\xi$ that 
    $1-\xi \le \frac{1}{16}\frac{\eps\phi^2}{\eta} \le \frac{1}{16}(\phi^2\eps)$
    as $\eta \ge 1$.
    When combined with the spectral error bound and degree-preserving guarantees
    from above, \Cref{item:intpatch_algo2} of \Cref{lemma:intpatch} gives
    \begin{align*}
        &\normop{\DD^{\frac \dag 2} (\vLL_{\vH} - \vLL_{\vG}) \DD^{\frac \dag 2}}
        \\
        \le&
        \normop{[\DDo]^{\frac \dag 2} \HH^\top (\WW_{\vH} - \WW) \HH [\DDo]^{\frac \dag 2}}
        + \normop{[\DDi]^{\frac \dag 2} \TT^\top (\WW_{\vH} - \WW) \HH [\DDo]^{\frac \dag 2}}
        \\
        =&
        \normop{[\DDi]^{\frac \dag 2} (\vLL_{\vH} - \vLL_{\vG}) [\DDo]^{\frac \dag 2}}
        \le
        \normop{[\DDi]^{\frac \dag 2} (\vLL_{\vG'} - \vLL_{\vG}) [\DDo]^{\frac \dag 2}}
        + \normop{[\DDi]^{\frac \dag 2} (\vLL_{\vH} - \vLL_{\vG'}) [\DDo]^{\frac \dag 2}}
        \\
        \le& 
        \frac{\eps\phi^2}{10\eta} + \frac{2}{16}(\phi^2\eps) +
        \frac{1}{10}(\phi^2\eps)
        \le \frac{\phi^2\eps}{2}.
    \end{align*}
    Using Cheeger's inequality (\Cref{lemma:cheeger}) and degree-preserving
    guarantees, we get the final spectral error bound of
    \[
        \normop{\LL_G^{\frac \dag 2} (\vLL_{\vH} - \vLL_{\vG}) \LL_G^{\frac \dag 2}}
        \le 
        \frac{2}{\phi^2} \normop{\DD^{\frac \dag 2} (\vLL_{\vH} - \vLL_{\vG}) \DD^{\frac \dag 2}}
        \le
        \eps.
    \]
    We note that \Cref{item:intpatch_algo3} of \Cref{lemma:intpatch} guarantees
    no edge weight of $\vH$ is negative.

    We now consider the number of edges in $\vH$.
    A simple application of the multiplicative Chernoff's inequality ensures
    that the sampling algorithm gives $\nnz(\ww') = O(\eps^{-2} \phi^{-6}
    n\log^2(n) \log \frac{n}{\delta})$ with probability at least $1 -
    \frac{\delta}{8}$ for any valid $n$.
    When combined with \Cref{item:intpatch_algo3} of \Cref{lemma:intpatch}, we
    get the desired sparsity guarantee.
    As for the runtime, we get the first term $O(m)$ comes from the sampling
    step.
    The second term follows by the runtime of \PAI given in
    \Cref{lemma:intpatch} given our sparsity bound on $\vG'$ above.
    The final probability follows by a union bound on all the
    probablistic events mentioned above.
\end{proof}

\begin{lemma} \label{lemma:subspar_ext}
    \SS (Algorithm~\ref{alg:sparsify_static}) satisfies that
    given $\vG,\eps,\delta$ as in the input conditions and a patching algorithm
    \PAE with $\eta = 1$, it returns a simple weighted directed graph $\vH$ such
    that with probability at least $1-\delta$, $\vH$ is a
    $\eps$-degree-preserving directed spectral sparsifier of $\vG$ with $m_{\vH}
    = O(\eps^{-2} \phi^{-4} n \log \frac{n}{\delta})$.
    The algorithm runs in time $O(m)$.
\end{lemma}

\begin{lemma} \label{lemma:subspar_star}
    \SS (Algorithm~\ref{alg:sparsify_static}) satisfies that
    given $\vG,\eps,\delta$ as in the input conditions and a patching algorithm
    \PAS with $\eta = 1$, it returns a simple weighted directed graph $\vH$ on
    vertices $V \cup \{x\}$ such that with probability at least $1-\delta$,
    $\Sc(\vH,V)$ is a $\eps$-degree-preserving directed spectral sparsifier of
    $\vG$ with $m_{\vH} = O(\eps^{-2} \phi^{-4} n \log \frac{n}{\delta})$. 
    The algorithm runs in time $O(m)$.
\end{lemma}

The proofs of \Cref{lemma:subspar_ext,lemma:subspar_star} are omitted as they
follow almost identically to that of \Cref{lemma:subspar_int}.
We are now ready to prove
\Cref{thm:dispec_ext_static,thm:dispec_int_static,thm:dispec_star_static}.
Again, we omit the proof of \Cref{thm:dispec_ext_static} and focus on the ones
that are harder to prove.
\begin{proof}[Proof of \Cref{thm:dispec_int_static}]
    By the union property \Cref{lemma:dicut_union} and the error guarantees in
    \Cref{lemma:subspar_int}, we get the final graph $\vH$ is a
    $\eps$-degree-preserving directed spectral approximation to $\vG$.
    Each $G_i$ is a $1/O(\log^2 n)$ expander.
    Our parameter choices and \Cref{lemma:subspar_int} then gives 
    \[
        m_{\vH_i} = O(\eps^{-2} n_i \log^{18}(n) \log(\frac{n}{\delta})),
    \] 
    with $n_i$ the number of non-trivial vertices\footnotemark~in $\vG_i$.
    \footnotetext{A vertex is trivial if it has no incident edges.}
    The total number of edges depends linearly on the vertex converage $J =
    O(\log n \log W)$, giving us the final sparsity bound.
    Note that the total number of subgraph $I$ is at most the number of edges
    $m$. 
    This justifies our choice of $\frac{\delta}{4m}$ for the expander
    sparsification.
    As the expander decomposition fails with probability at most
    $\frac{\delta}{2}$, we get by a union bound that the overall success
    probability is at least $1-\delta$.
    The runtime bound results simply from a combination of the runtime of
    \Cref{prop:ex_partition} and total runtime of \Cref{lemma:subspar_int} for
    all expanders.
\end{proof}

\begin{proof}[Proof of \Cref{thm:dispec_star_static}]
    The proof follows similarly to that of \Cref{thm:dispec_int_static}.
    Note that the error analysis here should be carried out on the Schur
    complement $\Sc(\vH,V)$ for $\vH$ on vertices $V \cup X$.

    There is only one additional step in this proof we need to consider: the
    size of $|X|$.
    Notice that we only need to add one additional vertex per expander by
    \Cref{lemma:subspar_star}.
    Then, $|X|$ can simply be bounded by the total number of expanders, which,
    by \Cref{prop:ex_partition}, is at most $O(n\log n \log W)$.
\end{proof}

\subsection{Reduction to decremental sparsifier on expanders}
\label{ssec:dynexpdecomp}

We recall a fully dynamic expander decomposition algorithm by
\cite{BernsteinvdBPGNSSS22} using expander pruning from \cite{SaranurakW19}.
Again, we replaces the use of the static expander decomposition algorithm from
\cite{SaranurakW19} by the state-of-the-art algorithm in \cite{AgassyDK23}.
\begin{lemma}[Expander pruning]
    \label{lemma:expprun}
    Let $G=(V,E,\ww)$ be a $\phi$-expander with $\frac{\max_{e \in
    E}{\ww_e}}{\min_{e \in E}{\ww_e}} \le 2$.
    There is a deterministic algorithm with access to adjacency lists of $G$
    such that given an online sequence of $k \le \frac{\phi m}{10}$ edges
    deletions in $G$, can maintain a pruned set $P \subset V$ satisfing the
    following properties.
    Let $G_i,P_i$ be the graph and set $P$ after $i$ edge deletions.
    We have for all $i$
    \begin{enumerate}
        \item $P_0 = \emptyset$, $P_i \subseteq P_{i+1}$,
        \item $\Vol_{G_i}(P_i) \le \frac{16i}{\phi} \min_{e \in E} \ww_e$,
        \item $G_i[V \setminus P_i]$ is a $\frac{1}{12}\phi$-expander,
        \item $|E(P_i, V \setminus P_i)| \le 4i$.
    \end{enumerate}
    The amortized update time is $O(\phi^{-2} \log m)$.
\end{lemma}

\begin{proposition}[Dynamic degree expander decomposition]
    \label{prop:dynexpdecomp}
    Given a weighted undirected dynamic graph $G$ with edge bounded edge weights
    ratio $\frac{\max_{e \in E} \ww_e}{\min_{e \in E} \ww_e} \le W$ at all time,
    there exists a dynamic algorithm against an adaptive adversary that
    preprocess in $O\Par{m \log^7 n}$ time,
    maintains with high probability a $(\cded\log^{-2}(n), (\log(W) + 3)(2\log n
    + 3)\log n)$-expander decomposition over all polynomially bounded number of
    updates of $G$ with high probability for some constant $\cded$.
    The algorithm supports both edge deletions and insertions with $O(\log^2 n)$
    amortized recourse and $O(\log^7 n)$ amortized time.

    After each update, the output consists of a list of potential changes to the
    decomposition: (i) edge deletions to some subgraph in the decomposition,
    (ii) removing some subgraph from the decomposition, and (iii) new subgraph
    added to the decomposition.
\end{proposition}

Crucially, the types of updates in \Cref{prop:dynexpdecomp} entails that any
expander in the dynamic decomposition only receives decremental update before it
becomes obsolete.
For the rest of this section, we focus on deriving a decremental sparsifier on
directed graphs with corresponding undirected graphs being expanders.
For completeness, we provide a proof of \Cref{prop:dynexpdecomp} in
\Cref{app:dynexpdecomp}.

We now recall the reduction to decremental algorithm on expanders by
\cite{BernsteinvdBPGNSSS22} using dynamic expander decomposition algorithm from
above.
The version we present here is a slight generalization: we remove the
perturbation requirement and allows for directed graphs.

\begin{lemma}[Blackbox reduction with dynamic expander decomposition] \label{lemma:red_dynexp}
    Assume $\gH$ is a graph sparsification problem that satisfies the union
    property, and there is a decremental algorithm $\gA$ for $\gH$ on
    (directed) graphs with edge weights ratio at most $2$ and satisfying that
    their corresponding undirected graph is a $\phi$-expander for $\phi \ge
    \cded \log^{-2} n$.
    Suppose $\gA$ preprocess in time $P(m,\eps) \ge m$, and maintains a
    sparsifier of size at most $S(n,\eps) \ge n$ with $N(n) \ge 0$ extra
    of vertices in $T(n,\eps)$ amortized update time.
    The maintained sparsifier is is either explicit, or is implicit and has
    query time $Q_1(n,\eps) \ge 1$ for an edge query or $Q_2(n,\eps) \ge
    S(n,\eps)$ for a graph query, where $P,N$ is superadditive
    in the first variable and $S,Q_2$ are subadditive in the first variable.

    There exists a fully dynamic algorithm $\gB$ for $\gH$ on
    weighted (directed) graphs such that given a dynamic weighted simple
    directed graph $\vG$ with weight ratios bounded by $W$, and a parameter
    $\eps$, $\gB$ preprocesses in time of $O(m \log^7 n + P(m,\eps))$ and
    maintains a sparsifier in $\gH(\vG,\eps)$ of size $O(S(n,\eps) \cdot \log^2
    n \log W)$ with $O(N(1) \cdot n \log^2 n \log W)$ extra number of vertices
    in amortized update time
    \[
        O\Par{\log^7 n + (T(n,\eps)+ \frac{P(n^2,\eps)}{n^2}) \log^3 n}.
    \]
    The sparsifier is explicit if $\gA$ maintains an explicit sparsifier.
    Otherwise, $\gB$ supports edge or graph query in time $O(Q_1(n,\eps) \cdot
    \log^2 n \log W)$ or $O(Q_2(n, \eps) \cdot \log^2 n \log W)$ respectively.
\end{lemma}
\begin{proof}
    Our algorithm maintains a $\vG = \bigcup_i \vG_i$ of the graph $\vG$ 
    through \Cref{prop:dynexpdecomp} such that each corresponding undirected
    graph $G_i = \und(\vG_i)$ is a $\phi = O(\log^{-2} n)$-expander with bounded
    weight ratio.
    For each subgraph in the decomposition, we maintain a sparsifier $\vH_i \in
    \gH(\vG_i, \eps)$ using $\gA$.
    The union property then ensures that $\bigcup_i \vH_i \in \gH(\vG, \eps)$.

    Consider edge updates on our data structure.
    By \Cref{prop:dynexpdecomp}, there are three potential changes. 
    (i) A set of edges is removed from some $G_i$.
    Then, the corresponding set of directed edges are removed from $\vG_i$ and
    These edge deletions are propagated to $\vH_i$ through $\gA$.
    (ii) Some $G_i$ is removed from the decomposition.
    We then also remove the corresponding directed graph $\vG_i$ and its
    maintained instance of $\gA$ with sparsifier $\vH_i$.
    (iii) A new subgraph $G_j$ is added to the decomposition. 
    We then add the corresponding $\vG_j$ to the directed graph decomposition
    and initialize a new instance of $\gA$ on $\vG_j$ with sparsifier $\vH_j$.

    For the total preprocessing time of $\gA$, we get $\sum_i P(|E(\vG_i)|,\eps)
    \le P(\sum_i |E(\vG_i)|,\eps) \le P(m)$ where we used $P(m,\cdot) \ge m$, its
    superadditive property, and the fact that $\vG = \bigcup_i \vG_i$.
    Combine with the preprocessing time of the dynamic expander decomposition
    algorithm in \Cref{prop:dynexpdecomp}, we the total preprocessing time of
    $O(m \log^7 n + P(m,\eps))$.
    Consider the amortized update time.
    Each update on the graph $\vG$ incurs an amortized recourse of $O(\log^3 n)$
    edge updates.
    Thus, the amortized time for updating all $\gA$ instances is bounded by
    $O(T(n,\eps) \cdot \log^3(n))$.
    Now, each time a subgraph $\vG_j$ is created as in (iii), it takes
    $P(|E(\vG_j)|,\eps)$ time to initialize $\gA$ on it.
    Notice that for a new subgraph $\vG_j$ to be formed, there needs to be at least
    $\frac{|E(\vG_j)|}{\log^3 n}$ updates on $\vG$.
    As such, the amortized cost of all initialization of $\gA$ can be bounded by
    $O\Par{\frac{P(n^2,\eps)}{n^2} \log^3 n}$, again using the assumption that
    $P(k,\cdot) \ge k$, its superadditivity, and the fact that $\vG$ is simple.
    Combine both amortized cost above with the amortized cost in
    \Cref{prop:dynexpdecomp}, we get a total amortized update time of
    \[
        O\Par{\log^7 n + (T(n,\eps)+ \frac{P(n^2,\eps)}{n^2}) \log^3 n}.
    \]
    Using the vertex coverage guarantees of the decomposition given by
    \Cref{prop:dynexpdecomp} and the subadditivity of $S$, the size the
    sparsifier $\vH$ is bounded by
    \[
        \sum_i S(|V(G_i)|,\eps) \le |J| \cdot S(n,\eps) \le O(S(n, \eps) \cdot
        \log^2 n \log W), 
    \]
    The number of extra vertices follows also by \Cref{prop:dynexpdecomp} and
    the superadditivity of $N$,
    \[
        \sum_i N(|V(G_i)|) \le n|J| \cdot N(1) \le O(N(1) \cdot n \log^2 n
        \log W).
    \]

    Finally, consider the query operations and their respectively complexity
    when the sparsifier is implicit.
    For edge queries, we note that our data structure can maintain for each
    vertex $v$ a set of subgraphs in the decomposition that contains its
    incident edges with no additional asymptotic cost in its amortized runtime.
    The vertex coverage guarantee then gives that there are at most $O(\log^2 n
    \log W)$ such subgraphs.
    Given a pair of vertices $u,v$, our algorithm propagates the query to only
    the incident subgraphs and aggregate the results after the query operation.
    The runtime is the then bounded by $O(Q_1(n,\eps) \cdot \log^2 n \log W)$.
    As for graph queries, we can simply perform the queries on each subgraph
    $\vG_i$ to obtain an explicit $\vH_i$ and aggregate all $\vH_i$'s.
    For a simple representation, such as weighted adjacency matrix, the total
    query time is
    \[
        \sum_i Q_2(|V(\vG_i)|,\eps) \le |J| \cdot Q_2(n, \eps) \le O(Q_2(n,
        \eps) \cdot \log^2 n \log W),
    \]
    once again by the subadditivity properties.
\end{proof}

\subsection{A simple decremental algorithm on expanders}
\label{ssec:dyndec_spec}
Finally, we present our main ingredient for a fully dynamic spectral sparsifier:
a decremental algorithm for maintaining a degree-preserving directed spectral
sparsifier for an expander.
We again give three results: the first one supports fast edge query but may
contain extra edges, the second result gives subgraph sparsifier, and the third
one gives explicit sparsifier with an extra vertex.

\begin{lemma} \label{lemma:decexpspec_ext}
    Suppose $\vG=(C \cup R, E, \ww)$ is a weighted directed bipartite
    graph with edges $E \subseteq C \times R$ and weight ratio $\frac{\max_{e
    \in E} \ww_e}{\min_{e \in E} \ww_e} \le 2$ undergoing edge deletions only.
    There is an algorithm that given $\eps,\phi \in (0,1)$ and an oblivious
    online sequence of edge deletions, preprocesses in time $O(m)$,
    maintains implicitly a weighted directed graph $\vH$ on vertices $C \cup R$
    with $E(\vH) \subseteq C \times R$ and $m_{\vH} = O(\eps^{-2} \phi^{-4} N
    \log n)$ where $N$ is the number of non-trivial vertices currently.
    The algorithm has worst-case update time $O(\eps^{-2} \phi^{-4} \log^2 n)$,
    and has query time of $O(\log n)$ for each edge query and $O(\eps^{-2}
    \phi^{-4} N \log^2 n)$ for each graph query.
    Further more, with high probability, whenever $G = \und(\vG)$ becomes a
    $\phi$-expander, $\vH$ is an $\eps$-degree-preserving directed spectral
    sparsifier of $\vG$.
\end{lemma}

\begin{lemma} \label{lemma:decexpspec_int}
    Under the same assumptions of \Cref{lemma:decexpspec_ext},
    there is a dynamic algorithm that preprocesses in time $O(m)$, maintains
    implicitly a reweighted directed subgraph $\vH$ of $\vG$ with 
    $m_{\vH} = O(\eps^{-2} \phi^{-8} N \log^3 n)$, where $N$ is the number of
    non-trivial vertices currently.
    The algorithm has worst-case update time $O(\eps^{-2} \phi^{-8} \log^3 n)$,
    and has query time of $\bO(\eps^{-2} \phi^{-8} n \log^4 n)$ for each graph
    query.
    Further more, with high probability whenever $G = \und(\vG)$ becomes a
    $\phi$-expander, $\vH$ is an $\eps$-degree-preserving directed spectral
    sparsifier of $\vG$.
\end{lemma}

\begin{lemma} \label{lemma:decexpspec_star}
    Under the same assumptions of \Cref{lemma:decexpspec_ext},
    there is a dynamic algorithm that preprocesses in time $O(m)$, maintains
    explicitly a directed graph $\vH$ on vertices $C \cup R \cup \{x\}$
    and $m_{\vH} = O(\eps^{-2} \phi^{-4} N \log n)$, where $N$ is the number of
    non-trivial vertices currently.
    The algorithm has worst-case recourse and update time $O(\eps^{-2} \phi^{-8}
    \log^2 n)$.
    Further more, with high probability whenever $G = \und(\vG)$ becomes a
    $\phi$-expander, $\Sc(\vH,C \cup R)$ is an $\eps$-degree-preserving directed
    spectral sparsifier of $\vG$.
\end{lemma}

In light of our static algorithm, our dynamic algorithm maintains a reweighted
subgraph $\vG'$ satisfying the approximation guarantees in \eqref{eq:degappr}
for its weighted adjacency matrix $\AA' = \TT^\top \WW' \HH$ with $\RR = \DDi$
and $\CC = \DDo$.
For brevity we say that $\vG'$ is an $\eps$-degree approximation of $\vG$.

\begin{algorithm2e}[ht!] \label{alg:dsp}
    \caption{$\DSPalgo(\vG=(V,E,\ww),\eps,\delta)$}
    \codeInput $\vG=(V,E,\ww)$ a simple weighted directed graph with edge
    weights ratio $\frac{\max_{e \in E} \ww_e}{L} \le 2$ for all $e \in E$ where
    we set $L \defeq \min_{e \in E} \ww_s$ the initial smallest edge weight,
    $\eps,\delta \in (0,1)$\;
    $\rho \defeq \cds \cdot \eps^{-2} \log(8n)$\;
    Initialize $\ww' \gets \vzero_E$ and compute $\ddi,\ddo$ the in and out-degree vectors of $\vG$\;
    \ForEach{$v \in V$}{
        $F_v \gets \SAalgo(E_v, \frac{2L\rho}{\dd_v})$\;
        $\ww_e' \gets \frac{1}{p_e} \ww_e$ for each $e \in F_v$, where $p_e = 1$
        if either $\dd_{h(e)} \le 2L\rho$ or $\dd_{t(e)} \le 2L\rho$ and  $p_e =
        4L^2\rho^2\frac{(\dd_{h(e)} + \dd_{t(e)} - 1)}{\dd_{h(e)}\dd_{h(e)}}$
        otherwise\;
    }
\end{algorithm2e}

Our dynamic algorithm has some slight deviations from the static setup.
Similar to the algorithm in \cite{BernsteinvdBPGNSSS22},
Algorithms~\ref{alg:dsp},~\ref{alg:dsd}) sample for each vertex a subset of its
incident edges using a fast uniform subset sampling algorithm
\Cref{prop:subset}.
An edge $e = (u,v)$ is present in the sparsifier as long as one of $u,v$ has it
sampled.
\begin{proposition}[Subset Sampling, \cite{Knuth97,Devroye06,BringmannP12}] \label{prop:subset}
    Given a universe $U$ of size $n$ and a sampling probability $p$, there is an
    algorithm that samples a set $S \subseteq U$ in time $O(\min(p n \log
    \frac{n}{\delta}, n))$ with probability at least $1-\delta$, where each element
    of $U$ is in $S$ with probability $p$ independently.
\end{proposition}

\begin{algorithm2e}[ht!] \label{alg:dsd}
    \caption{$\DSDalgo(e)$}
    Update $[\ddi]_v,[\ddo]_v$ for each $v \in e$\;
    \ForEach{$v \in e$ such that $v \in V_S$}{
        $F_v \gets \SAalgo(E_v, \frac{2L\rho}{\dd_v})$\;
        $\ww_e' \gets \frac{1}{p_e} \ww_e$ for each $e \in F_v$, where $p_e = 1$
        if either $\dd_{h(e)} \le 2L\rho$ or $\dd_{t(e)} \le 2L\rho$ and  $p_e =
        4L^2\rho^2\frac{(\dd_{h(e)} + \dd_{t(e)} - 1)}{\dd_{h(e)}\dd_{h(e)}}
$
        otherwise\;
    }
    Update $\ddi,\ddo$\;
\end{algorithm2e}

\begin{lemma} \label{lemma:decexpdeg}
    Under the same assumptions on $\vG,\eps,\phi$ as in
    \Cref{lemma:decexpspec_ext},
    there is an algorithm that given an oblivious online sequence of edge
    deletions, preprocesses in time $O(m)$, maintains explicitly a reweighted
    directed subgraph $\vG'$ that is an $\eps$-degree approximation of $\vG$ with
    $\nnz(\ww') = O(\eps^2 N \log n)$ where $N$ is the number of non-trivial
    vertices currently.
    The algorithm has worst-case recourse $O(\eps^{-2} \log n)$ and worst-case
    update time $O(\eps^{-2} \log^2 n)$ .
\end{lemma}

\begin{proof}%
    We use \DSP for preprocessing and \DSD to handle edge deletions.
    We start by proving our complexity guarantees.
    The runtime of the preprocessing steps is bounded the runtime of all the
    edge samplings plus an additional overhead of $O(m)$.
    Then, \Cref{prop:subset} gives that the total runtime for sampling is with
    high probability $O(\min(\rho n \log n, m))$ where we note that
    \begin{equation} \label{eq:degsum}
        \frac{|E_v|L}{\dd_v} \in (0.5,1) \sum_{e \ni v} \frac{\ww_e}{\dd_v} =
        (0.5,1)
    \end{equation}
    using our edge weights condition.
    For an edge deletion of $e$, our algorithm performs at most 2 vertex
    resamplings for its incident vertex $v \ni e$.
    Our high probability worst-case update time then follows by
    \Cref{prop:subset}, \eqref{eq:degsum} and a union bound over at most $n^2$
    total deletions.

    To bound the recourse, we show with high probability over
    the entire sequence of oblivious updates, all sampled sets $F_v$ has
    $|F_v| = O(\rho)$.
    Condition on previous updates, we can apply a multiplicative Chernoff bound
    to show a high probability bound that each time we perform a resample on
    some $v$, we are guaranteed $|F_v| = O(\rho)$ by \eqref{eq:degsum}.
    Since the updates are oblivious (thus do not depend on our random choice), a
    union bound over preprocessing and all $n^2$ updates gives with high
    probability all reweighting throughout the operations satisfies
    $|F_v| = O(\rho)$ for a large enough constant $\cds$.
    The recourse of each edge update depends on size of the previous
    reweighting and the new reweighting after resamplings.
    Condition on the success of the bound above, we get that the worst-case
    recourse is at most $O(\rho) = O(\eps^{-2} \log n)$ as required.
    Moreover, we also get that the total number of edges $\nnz(\ww') \le \sum_{v
    \in V} |F_v| = O(\nnz(\dd) \cdot \eps^{-2} \log n)$ where $\nnz(\dd)$
    is the number of non-trivial vertices.

    Finally, we show that with high probability $\vG' = (V,E,\ww')$ is a
    $\eps$-degree approximation of $\vG = (V,E,\ww)$ throughout all updates.
    Conceptually, we let an edge removal of $e$ be carried out by setting
    $\ww_e \gets 0$ instead.
    Let $e_t$ be the edge deleted at the $t$'th update.
    We define $\vG_t = (V,E,\ww_t), \vG_t=(V,E,\ww_t')$ be the graphs $\vG,\vG'$
    after $t$'th edge deletions with $\vG_0,\vG_0'$ the initial input and
    sparsifier after preprocessing.
    For an edge deletion $e_t = (u,v)$, for all vertices $w \ne u,v$, its degree
    satisfies $[\dd_t]_w = [\dd_{t-1}]_w$.
    Thus, the sampling probability for all edges $e \not\ni u,v$ remains
    unchanged.
    As all edge deletions are oblivious and do not depend on our random choices,
    we can reuse them for these non-incident edges.
    Our random sampling gives a that the probability of $e$ being sampled is 
    \[
        p_e = 
        1 - (1-\frac{2L\rho}{\dd_{h(e)}})(1-\frac{2L\rho}{\dd_{t(e)}}) = 
        2L\rho\frac{(\dd_{h(e)} + \dd_{t(e)}) - 2L\rho)}{\dd_{h(e)}\dd_{t(e)}}
        \ge 
        \frac{4L^2\rho^2}{\dd_{h(e)}\dd_{t(e)}}
    \]
    if both $\dd_{h(e)},\dd_{t(e)} > 2L\rho$; otherwise $p_e = 1$ which is as
    large as possible. 
    Note that $p_e$ is as defined in the algorithm.
    Then,
    \[
        [\ww_t(w)']_f = 
        \begin{cases}
            0 & \mbox{w.p.~}
            1-p_e,
            \\
            \frac{1}{p_e} [\ww_t]_f & \mbox{w.p.~}
            p_e
        \end{cases}
    \]
    By \Cref{lemma:entrysample} and a union bound over preprocessing and at most
    $n^2$ edge updates, by having an appropriately set constant $\cds$, we get
    the $\eps$-degree approximation with high probability.
\end{proof}

\Cref{lemma:decexpspec_int} follows directly by combining \Cref{lemma:intpatch}
and \Cref{lemma:decexpdeg} with approximation factor $\eps' = \Theta(\eps
\phi^{-4} \log n)$ in the same way as the scaling and patching step at the end
of \SS (Algorithm~\ref{alg:sparsify_static}).
Note that we only ever compute the internal patching when given a graph query.
We remark that it is possible to combine a dynamic undirected Laplacian solver,
(e.g. \cite{DurfeeGGP19,vdBrandGJLLPS22}), for sublinear time edge queries on
dynamic sparsifiers with subgraph guarantees. 

We now show that one can implicitly maintain an external patching as in
\Cref{lemma:vertexflow} that supports fast edge query.
\begin{lemma}[Dynamic external patching] \label{lemma:dynextpatch}
    Let dynamic vectors $\dd_1 \in \R_{\ge 0}^{V_1},\dd_2 \in \R_{\ge 0}^{V_2}$
    for disjoint sets $V_1,V_2$ where each update changes the value of one entry
    in $V_1 \cup V_2$.
    There is an deterministic algorithm that maintains implicitly a vector $\ff
    \in \R_{\ge 0}^{V_1 \times V_2}$ on the complete directed bipatite graph on
    edges $E = V_1 \times V_2$ with edge-vertex transfer matrix $\BB = \HH -
    \TT$ such that for every $i \in \{1,2\}$%
    it satisfies $[|\BB|^\top \ff]_{V_i} \le \dd_i,$
    with equal sign holds when $\|\ddi\|_1 = \min(\|\dd_1\|_1,\|\dd_2\|_1)$,
    $\nnz(\ff) \le n$ where $n = |V_1| + |V_2|$,
    and for $\FF = \diag{\ff}$,
    \[
        \normop{\DD_2^{\frac \dag 2} \TT^\top \FF \HH \DD_1^{\frac \dag 2}} \le
        1.
    \]
    The algorithm has a preprocessing time of $O(n)$ and worst-case
    update time of $O(\log n)$.
    The entire $\ff$ can be queried in time $O((\nnz(\dd_1)+\nnz(\dd_2))\log n)$.
    Further, given as query a pair of entries $u \in V_1, v \in V_2$, it returns
    $\ff_{(u,v)}$ in time $O(\log n)$.
\end{lemma}
\begin{proof}
    We fix an arbitrary order of $V_1$ and $V_2$.
    For each $i \in \{1,2\}$, our algorithm maintains a segment tree for $\dd_i$
    augmented to support querying the next non-zero index in $\dd_i$
    \footnote{This operation can also be implemented using a separate data
    structure, such as a balanced binary search tree.}.
    For preprocessing, we simply initialize the two segment trees, taking $O(n)$
    time.
    For each value update in $\dd_i$, we propagate the update to its
    corresponding segment tree, giving a worst-case update time of $O(\log n)$.

    For a vertex $v \in V_i$, we say the range of $v$ is the $[s_{v-1},s_v]$
    where $s_v$ is the total sum of $\dd_i$ up to $v$ and $s_{v-1}$ is the total
    sum of $\dd_i$ before $v$ (0 if $v$ is the first in order).
    The vector $\ff$ is defined by
    \[
        \ff_{(u,v)} = |[s_{u-1},s_u] \cap [s_{v-1},s_v]|.
    \]
    It is immediate from this definition that 
    $[|\BB|\ff]_{V_1} = \HH \ff \le \dd_1, [|\BB|\ff]_{V_2} = \TT\ff \le \dd_2$
    and with equal signs if $\|\dd_2 \|_1 \le \|\dd_1 \|_1$ or respectively
    $\|\dd_1 \|_1 \le \|\dd_2 \|_1$.
    \Cref{lemma:optoinf} then gives the operator norm bound.
    To see the number of non-zero entries of $\ff$, observe that when going
    through the ranges of the segment trees, each non-zero $\ff_{(u,v)}$ must
    occupy the remaining demand of either $u$ or $v$ up to that point.
    Thus, $\nnz(\ff) \le n$ as required.

    Consider the query operations.
    For a single edge query with $(u,v)$, we simply compute then entry
    $\ff_{(u,v)}$ as above using at most 4 sum queries on the two segment trees.
    These queries can be carried out in $O(\log n)$ time.
    As for querying the entire $\ff$, we simply traverses through orders of
    $V_1,V_2$ simultaneously, incrementing the index
    for $v_i$ with the smaller current sum while skipping zero entries.
    The algorithm only needs to traverse through the entire current support
    $\supp(\dd_1) \cup \supp(\dd_2)$, with one sum query and one finding the
    next non-zero entry query for each element.
    As both queries run in time $O(\log n)$, we get the required total query
    time.
\end{proof}

Similarly, \Cref{lemma:decexpspec_ext} follows by combining
\Cref{lemma:dynextpatch,lemma:extpatch_algo,lemma:decexpdeg} with
approximation factor $\eps' = \Theta(\eps \phi^{-2})$.

\begin{proof}[Proof of \Cref{lemma:decexpspec_ext}]
    Our algorithm is a dynamic version of \SS
    (Algorithm~\ref{alg:sparsify_static}).
    It maintains an $\eps'$-degree approximation $\vG'$ of $\vG$ with
    $\eps' = \Theta(\eps \phi^{-2})$ through \Cref{lemma:decexpdeg} and dynamic
    maintains an external patching using \PAE (Algorithm~\ref{alg:patching_ext})
    with \Cref{lemma:vertexflow} replaced by its dynamic version
    \Cref{lemma:dynextpatch}.
    The approximation and sparsity guarantees follows similarly to the
    that of the static algorithm \Cref{lemma:subspar_ext}.

    The update time consists of the update time of our degree approximation
    $\vG'$ and the total update time incurred by the changes to $\vG'$.
    Since there is a recourse of at most $O(\eps^{-2} \phi^{-4} \log n)$ and a update time
    of $O(\eps^{-2} \phi^{-4} \log^2 n)$ per update on $\vG$, the vectors
    $\dd_1,\dd_2$ has at most $O(\eps^{-2} \phi^{-4} \log n)$ entries change.
    \Cref{lemma:dynextpatch} then gives a worst-case update time of $O(\eps^{-2}
    \phi^{-4} \log^2 n)$ for the entire operation.
    For the queries, we simply propagate them to $\vG'$ and $\ff$ and aggregate
    them together using the fixed scalings in \PAE.
\end{proof}

Our proof of \Cref{lemma:decexpspec_star}, which we omit, also follows
almost identically to the proofs of
\Cref{lemma:decexpspec_int,lemma:decexpspec_ext} above.
Note that the static algorithm \PAS (Algorithm~\ref{alg:patching_star}) can be
trivially made dynamic with deterministic updates in $O(1)$ worst-case update
time.
This dynamic algorithm can also maintain the guarantees in
\Cref{lemma:starpatch_algo}.

We are now ready to prove our main theorems of this section.
Once again, we omit the proof to \Cref{thm:dynspecint_for} and focus only on
\Cref{thm:dynspecext_for,thm:dynspecstar_for}
\begin{proof} [Proof of \Cref{thm:dynspecstar_for,thm:dynspecext_for}]
    Our dynamic algorithm maintains an $\eps$-degree-preserving directed
    spectral sparsifier for the bipartite lift $\vG^\uparrow = \blift(\vG)$.
    \Cref{lemma:blift_spectral,lemma:blift_schur} then allows us to recover an
    $\eps$-degree balance preserving sparsifier, even for a graph with extra
    vertices, using the canonical bijection between the edges $E_{\vG}^\uparrow$
    and $E_{\vG}$.
    We can now assume w.l.o.g. that our algorithm is given instead a directed
    bipartite graph with bipartite $C \cup R = V$ and all directed edges from
    $C$ to $R$.
    
    Since our sparsification problem satisfies the union property
    \Cref{lemma:dispec_union}, we can apply the dynamic expander decomposition
    framework in \Cref{lemma:red_dynexp} using our decremental sparsifier on
    expander with dynamic external patchings \Cref{lemma:decexpspec_ext}. 
    We set the expansion to $\phi = 1/\Omega(\log^2 n)$.
    Then, both algorithms in
    \Cref{lemma:decexpspec_ext,lemma:decexpspec_star} has preprocessing time
    $P(m,\eps) = O(m)$, update time $T(N,\eps) \le O(\eps^{-2} \log^{10} n)$,
    sparsity $S(N,\eps) \le O(\eps^{-2} N \log^9 n)$.
    The only difference here for them is that the algorithm in
    \Cref{lemma:decexpspec_star} maintains an explicit sparsifier with
    $N(\cdot) = 1$ extra vertex, while the one in \Cref{lemma:decexpspec_ext}
    maintains an implicit sparsifier with $N(\cdot) = 0$ and has query time
    $Q_1(N,\eps) \le O(\log n)$ for an edge query and $Q_2(N,\eps) \le
    O(\eps^{-2} N \log^{10} n)$ for a graph query.
    Note that $N$ is the number of non-trivial vertices in the decremental
    directed graphs with corresponding undirected expander, and $n$ is the
    number of vertices in the entire graph $\vG$.
    Then, our overall guarantees are given by \Cref{lemma:red_dynexp}.
    To summarize, our dynamic algorithm has preprocessing time of $O(m \log^7
    n)$, amortized update time $O(\eps^{-2} \log^{13} n)$, number of edges in
    the sparsifier $O(\eps^{-2} n \log^{11} n \log W)$.
    For \Cref{thm:dynspecstar_for}, the graph $\vH$ is explicitly maintained with
    extra vertices $X$ satisfies $|X| \le O(n \log^2 n \log W)$.
    For \Cref{thm:dynspecext_for}, the sparsifier $\vH$ is implicitly maintained
    with edge query time $O(\log^3 n \log W)$, and graph query time $O(\eps^{-2}
    n \log^{12} n \log W)$.
\end{proof}

\subsection{Degree-preserving directed spectral sparsification}
\label{ssec:degpre}

In this section, we show to extend the degree-balance preserving guarantee from
the algorithms above to exact degree-preserving, i.e., the algorithm maintains a
sparsifier with exact same weighted in and out degrees as the input directed
graph.
Even though our algorithms are degree-preserving on bipartite lifted graphs, the
patching schemes, except internal patching, may introduce self-loops when
``unlifting'' the sparsifier back to the original set of vertices.
Specifically, \PAE (Algorithm~\ref{alg:patching_ext}) may introduce for $v \in
V$ an edge from $v$ to $v'$, where $v'$ is the corresponding copy of $v$ in the
lifted graph.
When contracting back $v$ and $v'$, this extra edge creates a self-loop.
This is not a problem for \PAI (Algorithm~\ref{alg:patching_int}), since the
algorithm only returns edges in the original lift, which cannot contain edges of
the form $(v,v')$.
A similar problem arises for our star patchings.
Eliminating the center of a star creates a biclique on all of its adjacent
vertices.
If there was an edge $(v,x)$ and $(x,v')$ for center $x$, then the biclique
contains the degenerative edge of $(v,v')$, making it not degree-preserving.

We show first that the external patching algorithm can be modified to satisfy
degree-preserving condition.
Note that we specifically require for a pair of corresponding vertices $(v,v')$,
the sum of their weighted degrees $\ddo_v + \ddi_{v'}$ is no more than the total
out demand (or in demand). 
The is crucial, since otherwise self-loop can be
unavoidable.
It is also a logical assumption -- otherwise there must be an edge
from $v$ to $v'$ in the first place.
\begin{lemma}[Dynamic external patching, degree-preserving] \label{lemma:dynextpatch_mod}
    Let dynamic vectors $\dd_1 \in \R_{\ge 0}^{V_1},\dd_2 \in \R_{\ge 0}^{V_2}$
    be defined on disjoint sets $V_1,V_2$ satisfying $\|\dd_1\|_1 = \|\dd_2\|_1$
    and $\max_{u \in V_i} [\dd_i]_u, \le \frac{1}{2}\|\dd_i\|_1$ for $i\in
    \{1,2\}$.
    Each update changes one entry in each $\dd_i$.
    Suppose we are given a fixed injective function $g
    : U_1 \mapsto V_2$ for subset $U_1 \subseteq V_1$ with $|U_1| \le |V_1|,
    |U_1| \le |V_2|$ satisfying always that $[\dd_1]_u + [\dd_2]_{g(u)} \le
    \|\dd_1\|_1$.
    There is a deterministic algorithm that maintains implicitly a vector $\ff
    \in \R_{\ge 0}^E$ on the complete directed graph on
    edges $E = V_1 \times V_2$ with edge-vertex transfer matrix $\BB = \HH -
    \TT$ such that for every $i \in \{1,2\}$%
    it satisfies $[|\BB|^\top \ff]_{V_i} = \dd_i$, $\ff_{(u,g(u))} = 0$ for all
    $u \in U_1$, $\nnz(\ff) \le 2n$ where $n = |V_1| + |V_2|$,
    and for $\FF = \diag{\ff}$,
    \[
        \normop{\DD_2^{\frac \dag 2} \TT^\top \FF \HH \DD_1^{\frac \dag 2}} \le
        1.
    \]
    The algorithm has a preprocessing time of $O(n)$ and worst-case
    update time of $O(\log n)$.
    The entire $\ff$ can be queried in time $O((\nnz(\dd_1)+\nnz(\dd_2))\log n)$.
    Further, given as query a pair of vertices $u \in V_1, v \in V_2$, it returns
    $\ff_{(u,v)}$ in time $O(\log n)$.
\end{lemma}
\begin{proof}%
    Similar to before (\Cref{lemma:dynextpatch}), we are using augmented segment
    trees for maintaining $\ff$.
    Here, we will instead use 4 segment trees.
    Let $\PP_g$ be the matrix of $g$ that maps $\PP_g \vone_u = \vone_{g(u)}$
    for each $u \in V_1$.

    Let $\bb \in \R_{\ge 0}^{V_1}$ be defined by $\bb_u = \min_([\dd_1]_u,
    [\dd_2]_{g(u)})$ for $u \in U_1$ and $\bb_u = 0$ otherwise.
    Let $w \in U_1$ be the vertex with $\bb_w \ge \bb_u$ for any $u \in U_1$.
    We define a vector $\aa$ to be the same as $\bb$ on all entries but $w$.
    For the entry on $w$, there are two cases to consider.
    When $\bb_w \le \frac{1}{2}\|\bb\|_1$, we set $\aa_w = \bb_w$ as well.
    Otherwise, we set $\aa_w = \|\bb\|_1 - \bb_w$.
    Notice that in both cases, 
    \[ 
        \max_{u \in V_1} \aa_u = \aa_w \le \frac{1}{2}\|\aa\|_1.
    \]
    We claim that 
    \[
        [\dd_1]_w + [\dd_2]_{g(w)} - 2\aa_w \le \|\dd_1\|_1 - \|\aa\|_1.
    \]
    Otherwise, $[\dd_1]_w + [\dd_2]_{g(w)} > \|\dd_1\|_1$, violating our
    assumption.
    Each update on $\dd_1,\dd_2$ incurs at most 3 entry changes on $\aa$ and a
    potential change of the maximum entry $w$.
    These changes can be tracked in $O(\log n)$ worst-case update time.

    Consider the first two segment trees, one for each $V_i$.
    Here, we do the same as in \Cref{lemma:dynextpatch} and store implicitly a
    flow $\ff_1$, but for demands $\dd_1 - \aa$ and $\dd_2 - \PP_g\aa$.
    For each $V_i$, we associated it with an order on which we build its segment
    tree.
    To avoid having $\ff_{(w,g(w))} > 0$, we pick $w$ to be the first in order for
    $V_1$ and $g(w)$ to be the last in order for $V_2$; the rest can be ordered
    arbitrarily.
    Note that $\dd_1-\aa \ge \vzero$, $\dd_2 - \PP_g\aa \ge \vzero$
    and $g(\supp(\dd_1 - \aa)) \cap \supp(\dd_2 - \PP_g\aa) \subseteq \{g(w)\}$
    by our construction.
    Then, our ordering and the claim above guarantees that 
    $[\ff_1]_{(u,g(u))} = 0$ for all $u \in U_1$.
    Since each update of $\dd_1,\dd_2$ incurs at most 3 changes to the entries
    of each of $\dd_1-\aa$ and $\dd_2 - \PP_g\aa$, 
    we get a worst-case update time of $O(\log n)$.
    Same as before, we have $\nnz(\ff_1) \le n$.

    For the next two segment trees, we, again, choose an specific ordering for
    each $V_i$.
    Note that the demands we are dealing with are $\aa$ and $\PP_g\aa$, which
    are identical under the permutation $\PP_g$.
    Now, we fix an order $\ppi_1$ for $V_1$ with the only requirement that
    the vertex $w$ is at first.
    For $V_2$, we set $[\ppi_2]_i = g([\ppi_1]_{i+1})$ for $i \le |U_1|-1$,
    $[\ppi_2]_{|U_1|} = \ppi_1 = w$, and the rest in arbitrary order after
    $|U_1|$.
    We observe that for any $k \le |U_1|-1$, 
    \[
        \sum_{i=1}^k \aa([\ppi_1]_i) \ge \sum_{i=1}^k (\PP_g\aa)([\ppi_2]_i),
    \]
    which can be obtained by noticing that
    \[
        \sum_{i=1}^k \aa([\ppi_1]_i) - \sum_{i=1}^k (\PP_g\aa)([\ppi_2]_i)
        =
        \aa([\ppi]_1) - \aa([\ppi]_k)
        \ge 0.
    \]
    Moreover, we have $\aa([\ppi_1]_1) \le \sum_{i=1}^{|U_1|-1}([\ppi_2]_i)$ by
    recalling that $\aa([\ppi_1]_1) = \aa_w \le \frac{1}{2} \|\aa\|_1$.
    Then, computing a flow $\ff_2$ for the demands as before would also
    guarantee that $[\ff_2]_{(u,g(u))} = 0$ for all $u \in U_1$.
    Recall from above that an update on $\dd_1,\dd_2$ incurs at most 3 changes
    on the entries of $\aa$.
    Our algorithm first propagate the changes on $\aa$ to the segment tree on
    $V_1$ with order $\ppi_1$.
    If the largest entry $w$ of $\aa$ is modified to $w' \ne w$,
    we swap the respective orders of $w,w'$ in the segment tree by reweighting
    the respective entries in time $O(\log n)$.
    For $\ppi_2$, we swap instead $[\ppi_2]_{|V_1|}$ and the order of $g(w')$ in
    $\ppi_2$ using the same operations in time $O(\log n)$.
    Thus, our update time for these two trees are also $O(\log n)$.
    We also have $\nnz(\ff_2) \le n$ as the argument in \Cref{lemma:dynextpatch}
    is not affected by the ordering of the vertices.

    Queries can be performed the same ways as in the proof of
    \Cref{lemma:dynextpatch}, but with a sum on the entries of $\ff_1$ and
    $\ff_2$.
\end{proof}

By replacing the algorithm in \Cref{lemma:dynextpatch} with that in
\Cref{lemma:dynextpatch_mod}, we get our degree-preserving dynamic algorithm.
\begin{theorem}[\Cref{thm:dynspecext_for}, degree-preserving]
    \label{thm:dynspecext_for_deg}
    There is a fully dynamic algorithm with the same guarantees in
    \Cref{thm:dynspecext_for} that additionally satisfies the sparsifier $\vH$
    is degree-preserving with respect to $\vH$.
\end{theorem}

We now turn our attention to the star patching algorithm.
To avoid self-loops, we instead consider using multiple stars.
We ensure for each star that the incident edges of the center do not contain any
pair of edges of the form $(v,x)$, $(x,v')$.
For vertex sets $V,V'$, and the complete set of bipartite edges $E = V \times
V'$, we devise a decomposition.
Firstly, we label the vertices of $V$ by $[n]$ where $n = |V|$ and denote by $l$ its
labeling.
For each $v'$ that corresponds to $v$, we assign it the same label.
Then, we defined for each integer $i \in [1, \ceil{\log_2(n)}]$ the edge set
$E_i$ by
\[
    (u,v') \in E_i \quad \text{if } \floor{\frac{l(u)}{2^{i-1}}} \ne
    \floor{\frac{l(v')}{2^{i-1}}}
    \text{ and } \floor{\frac{l(u)}{2^i}} = \floor{\frac{l(v')}{2^i}}.
\]
We then get the following \emph{disjoint} decomposition of $E$,
\[
    E = \{(v,v') | v \in V\} \cup \Par{\bigcup_{i = 1}^{\ceil{\log_2(n)}} E_i}.
\]
Then, to produce a degree-preserving directed spectral sparsifier, we first
decompose the bipartite lift $\vG^\uparrow = (V\cup V', F,\ww)$ into $\ceil{\log_2(n)}$ edges
disjoint subgraphs where each $\vG_i$ only contains the edges $F \cap
E_i$.
Our dynamic algorithm then proceeds the same as before
(\Cref{thm:dynspecstar_for}), with one copy of the algorithm on each $\vG_i$.
Now, suppose the algorithm receives an update (insertion or deletion) with edge
$e = (u,v')$, we first check for which $E_i$ the edge $e$ belongs to.
The update is then propagated to the copy on $\vG_i^\uparrow$ and processed.
Our final sparsifier $\vH$ is a union of our degree-preserving sparsifiers
$\vH_i^\uparrow$ of each $\vG_i^\uparrow$ unlifted back to the original set of
vertices.
The correctness of these operations are guaranteed by
\Cref{lemma:dispec_union,lemma:blift_schur}.
We summarize our conclusion above into the following statement.

\begin{theorem}[\Cref{thm:dynspecstar_for}, degree-preserving]
    \label{thm:dynspecstar_for_deg}
    There exists a fully dynamic algorithm that given a weighted directed graph
    $\vG$ on vertices $V$ undergoing oblivious edge insertions and deletions and
    an $\eps \in (0,1)$, 
    maintains \emph{explicitly} a graph $\vH$ on vertices $V \cup X$ with $X$
    disjoint from $V$ such that with high probability $\Sc(\vH,V)$ is an
    $\eps$-degree-preserving directed spectral sparsifier of $\vG$.
    The algorithm has preprocessing time $O(m \log^7 n)$ and amortized update time
    $O(\eps^{-2} \log^{13} n)$. 
    The sparsifier $\vH$ has size $O(\eps^{-2} n \log^{12} n \log W)$ and extra
    number of vertices $|X| \le O(n \log^3 n \log W)$.
\end{theorem}

Note that the main difference compared to \Cref{thm:dynspecstar_for}, apart from
the degree-preserving guarantee, is the slight increase in the size of $X$ and
the sparsity of the sparsifier $\vH$.

\section{A fully dynamic balanced directed cut sparsifier} \label{sec:dicut}
In this section we give fully dynamic algorithms for maintaining a
$\beta$-balanced directed cut sparsifier under edges updates.
We first define a stronger notion of directed cut approximation.
\begin{definition}[$\beta$-directed cut approximation]
    \label{def:baldicuappr_gen}
    For a weighted directed graph $\vG = (V,E,\ww)$ with its corresponding
    undirected graph $G \defeq \und(\vG)$, a weighted directed graph $H$ on $V$
    is said to $(\beta,\eps)$-directed cut (dicut) approximates $\vG$
    if for all directed cut $\vC = (U,V\setminus U)$ and its corresponding
    undirected cut $C$ in $G$
    \begin{equation} \label{eq:betacut_def}
        \ww_{\vG}(\vC) - \frac{\eps}{\sqrt{\beta+1}} \sqrt{\ww_{\vG}(\vC) \cdot \ww_G(C)}
        \le \ww_{\vH}(\vC) \le
        \ww_{\vG}(\vC) + \frac{\eps}{\sqrt{\beta+1}} \sqrt{\ww_{\vG}(\vC) \cdot
        \ww_G(C)},
    \end{equation}
    where parameters $\beta \ge 1$ and $\eps \in (0,1)$.
\end{definition}
The parameter $\beta$ quantifies the balance between the cut value of dicut
$\vC$ and $\rev(\vC)$.
\begin{definition}[$\beta$-balanced dicut]
    A directed cut $E(U,V\setminus U)$ is $\beta$-balanced if
    $\frac{1}{\beta} \cdot \ww(V\setminus U, U) \le \ww(U,V\setminus U) \le
    \beta \cdot \ww(V\setminus U, U)$.
\end{definition}
For a $\beta$-balanced dicut $\vC$, its corresponding undirected cut also
satisfies $\ww(C) \le (\beta+1) \ww(\vC)$, giving
\begin{equation} \label{eq:betabal_def}
    (1 - \eps) \ww_{\vG}(\vC) \le \ww_{\vH}(\vC) \le (1 + \eps) \ww_{\vG}(\vC).
\end{equation}
In the static setting, Cen, Cheng, Panigrahi, and Sun \cite{CenCPS21} initiated
the study of $\beta$-balance directed cut sparsification. %
Cen et al.~\cite{CenCPS21} achieved
$O(\eps^{-2}\beta n\log n)$ sparsity by independent edge sampling with
probability proportional to inverses of undirected \emph{edge strength} defined
originally by Bencz\'{u}r and Karger \cite{BenczurK96} for undirected cut
sparsification.
As it stands, there is currently no known algorithm for fast
maintenance of edge strength or fast certification of large edge strength under
edge updates.
We instead turn our attention to \emph{undirected edge connectivity} that is 
proven to be useful for undirected cut sparsification in the work
\cite{FunHHP11}.

\begin{definition}[Edge connectivity]
    For any pair of vertices $u,v \in V$ in a weighted \emph{undirected} graph
    $G=(V,E,\ww)$, the edge connectivity $\kk_{uv}$ between $u,v$ is defined as
    the minimum value of a cut that separates them.
    The connectivity of an edge $e = (u,v)$ is defined by $\kk_e = \kk_{uv}$.
\end{definition}

We show in \Cref{ssec:dicut_static} that a $\beta$-dicut sparsifier can also
be obtained by sampling using undirected edge connectivity.
\begin{theorem}[Formal version of \Cref{thm:statdicut_inf}] \label{thm:baldicut}
    Given $\vG=(V,E,\ww)$ a simple weighted directed graph with $G \defeq
    \und(\vG)$, a $\ell$-stretch edge connectivity estimation $\tkk$ of $G$
    and parameters $\eps,\delta,\beta$ satisfying the input conditions, \SDC
    (Algorithm~\ref{alg:baldicut})
    returns in time $O(m)$ a weighted subgraph $\vH = (V,E,\ww_{\vH})$ such that
    with probability at least $1-\delta$, for every directed cut $\vC = E(U,V
    \setminus U)$ with corresponding undirected cut $C$ in $G$,
    \[
        \ww(\vC) - \frac{\eps}{\sqrt{\beta+1}} \sqrt{\ww(\vC) \cdot \ww(C)}
        \le
        \ww_{\vH}(\vC)
        \le
        \ww(\vC) + \frac{\eps}{\sqrt{\beta+1}} \sqrt{\ww(\vC) \cdot \ww(C)},
    \]
    and if $\vC$ is $\alpha$-balanced,
    \[
        \Par{1-\eps\frac{\sqrt{\alpha+1}}{\sqrt{\beta+1}}} \cdot \ww(\vC)
        \le
        \ww_{\vH}(\vC)
        \le
        \Par{1+\eps\frac{\sqrt{\alpha+1}}{\sqrt{\beta+1}}} \cdot \ww(\vC),
    \]
    and the number of edges in $\vH$ satisfies $\nnz(\ww_{\vH}) =
    O(\eps^{-2}\beta \ell n\log\frac{n}{\delta})$.
\end{theorem}

If an undirected graph is an expander, its edge connectivities can be certified
by the degrees (see \Cref{lemma:expconn}).
One can then leverage expander decompositions to devise a fast algorithm for
directed cut sparsifications (see \Cref{thm:dicutapprox_final}).
Such an algorithm lends itself nicely to the dynamic setting by utilizing the
dynamic expander decomposition framework presented in \Cref{ssec:dynexpdecomp}.
We present a fully dynamic dicut sparsifier in the oblivious setting in
\Cref{thm:dyndicut_for}.

\begin{theorem}%
    \label{thm:dyndicut_for}
    There exists a fully dynamic algorithm for maintaining explicitly a
    $(\beta,\eps)$-directed cut sparsifier $\vH$ of a weighted directed graph
    $\vG$ undergoing oblivious edge insertions and deletions with high
    probability for any $\eps \in (0,1)$ and $\beta \ge 1$.
    The algorithm has preprocessing time $O(m \log^7 n)$, amortized update time
    $O(\eps^{-2} \beta \log^7 n)$, and maintains $\vH$ a subgraph of $\vG$ with
    size $O(\eps^{-2} \beta n \log^5 n \log W)$.
\end{theorem}

We also present a de-amortized result in \Cref{thm:dyndicut_worst_for} under the
weaker notion of $\beta$-balanced dicut approximation (\Cref{def:baldicuappr})
using the dynamic framework developed by \cite{AbrahamDKKP16}.
While it is possible to directly apply the black-box reduction developed in
\cite{BernsteinvdBPGNSSS22} for this task, their framework introduces
significant runtime overhead since it is geared towards reduction of dynamic
algorithms for adaptive adversary.
The algorithm by \cite{AbrahamDKKP16} dynamically maintains bundles of
approximate maximum spanning forests (called $\alpha$-MSFs), that can certify
edges with small undirected edge connectivities.

\begin{theorem}[Formal version of \Cref{thm:dyndicut_worst_inf}] 
    \label{thm:dyndicut_worst_for}
    There exists a fully dynamic algorithm for maintaining explicitly a
    $(\beta,\eps)$-balanced directed cut sparsifier $\vH$ of a weighted directed graph
    $\vG$ undergoing oblivious edge insertions and deletions with high
    probability for any $\eps \in (0,1)$ and $\beta \ge 1$.
    Given a fixed parameter $1 \le \gamma \le n^2$ the algorithm has
    preprocessing time $O(\eps^{-2} \beta m \log^5 n \log^3 \gamma)$, worst-case
    recourse $O(\log \gamma)$ and worst-case update time $O(\eps^{-2} \beta
    \log^5 n \log^3 \gamma)$ per update, and maintains
    $\vH$ a subgraph of $\vG$ with size 
    \[
        m_{\vH} = O(\eps^{-2} \beta n \log n \log W \log^3 \gamma +
        \frac{m}{\gamma}).
    \]
\end{theorem}
Note that by choosing $\gamma = \Theta(n^2)$, our dynamic algorithm has a
worst-case update time of $O(\eps^{-2} \beta \log^8 n)$, and explicit
nearly-linear sized sparsifier with $O(\eps^{-2} \beta n \log^4 n \log W)$
edges, which is slightly better than our amortized cut sparsifier.

\subsection{Directed cut approximation preliminaries}
\label{ssec:dicut_prelim}

Before we present our sparsification algorithms, we state and formally prove a
few useful properties of $\beta$-dicut approximations.
For our first property, we show that a $\beta$-dicut approximation also gives
cut approximation of corresponding undirected graphs.
A proof of \Cref{lemma:dicut_to_cut} is provided in \Cref{app:dicutproof}.
\begin{lemma} \label{lemma:dicut_to_cut}
    Suppose $\vH$ is a $(\beta,\eps)$-dicut approximation to $\vG$, then the
    corresponding undirected $H \defeq \und(\vH)$ is a
    $\frac{\sqrt{2}\eps}{\sqrt{\beta+1}}$-cut approximation to $\vG$, 
    i.e., for every undirected cut $C$
    \[
        (1-\frac{\sqrt{2}\eps}{\sqrt{\beta+1}}) \ww_G(C) \le \ww_H(C)
        \le (1+\frac{\sqrt{2}\eps}{\sqrt{\beta+1}}) \ww_G(C).
    \]
\end{lemma}

Next, we show that it also satisfies the union property.
This is crucial for applying the expander decomposition framework.
The proof is deferred to \Cref{app:dicutproof}
\begin{lemma}[Union property] \label{lemma:dicut_union}
    Suppose directed graph $\vG = \bigcup_{i=1}^k s_i \cdot \vG_i$ for
    some $k$ and that $s_1,\ldots,s_k \in \R_{\ge 0}$. 
    Then, suppose $\vH_i$ is a $(\beta,\eps)$-dicut approximation to $\vG_i$
    for every $i \in [k]$, it follows that $\bigcup_{i=1}^k s_i \vH_i$ is a
    $(\beta,\eps)$-dicut approximation to $\vG$.
\end{lemma}

For $\beta$-dicut approximation, there is no transitivity property as the
upperbound may have a large exponential blow-up for directed cuts that are much
smaller than their corresponding undirected cuts.
However, if we have $\beta$-dicut guarantees for each level of approximation we
can still observe the weaker $\beta$-balanced dicut approximation.
\begin{lemma}[Weak transitivity] \label{lemma:dicut_trans}
    Let $\vG_0 \defeq \vG$ be some directed graph, and suppose $\vG_i$ is a
    $(\beta,\eps_i)$-dicut approximation to $\vG_{i-1}$ for all $i \in [k]$ and
    for some $k$ and $\sum_{i \in [k]} \eps_i \le \frac{1}{\sqrt{2}}$.
    Then, $\vG_k$ is a $(\beta,3\sum_{i \in [k]} \eps_i)$-balanced dicut
    approximation to $\vG$.
\end{lemma}
\begin{proof}
    Let $\vC$ be an arbitrary cut of $\vG$ and we suppose $\ww_G(C) =
    \gamma^2(\beta+1) \cdot \ww_{\vG(\vC)}$ for some $\gamma \ge
    \frac{1}{\sqrt{\beta+1}}$.
    Then,
    \[
        \ww_{\vG_1}(\vC) \in (1\pm \gamma \eps_1)\ww_{\vG}(\vC).
    \]
    By the transitivity property of cut approximation and
    \Cref{lemma:dicut_to_cut}, we have
    \[
        \ww_{G_i}(C) \le \exp(\sum_{j \in [i]} \frac{\sqrt{2} \eps_j}{\sqrt{\beta+1}}) \ww_G(C)
        \le \exp(\frac{\sqrt{2} \sum_{j \in [i]} \eps_i}{\sqrt{\beta+1}}) 
        \gamma^2 (\beta+1) \cdot \ww_{\vG}(\vC)
    \]
    where $C$ is the corresponding undirected cut of $\vC$.
    It is then easy to prove inductively the following upper bounds
    \[
        \begin{aligned}
            \ww_{\vG_i}(\vC) &\le \exp(\sqrt{2} \gamma \sum_{j\in [i]} \eps_i) \ww_{\vG}(\vC),
            \\
            \sqrt{\ww_{\vG_i}(\vC) \ww_{G_i}(C)} 
            &\le \sqrt{\exp(\sqrt{2} \gamma \sum_{j \in [i]} \eps_i)
                \ww_{\vG}(\vC) \cdot \gamma^2(\beta+1) \exp(\sqrt{2} \gamma \sum_{j \in [i]} \eps_i)
                \ww_{\vG}(\vC)} 
            \\
            &\le \gamma \sqrt{\beta+1} \exp(\sqrt{2} \gamma \sum_{j \in [i]} \eps) \ww_{\vG}(C)
        \end{aligned}
    \]
    holds for every $i$, where we use that $\gamma \ge \frac{1}{\sqrt{\beta+1}}$.
    Since we are proving only for $\beta$-balanced dicut approximation, we
    assume $\vC$ is $\beta$-balanced and have $\gamma \le 1$.
    By $\sum_{i \in [k]} \eps_i \le \frac{1}{\sqrt{2}}$, we get $e^{\sqrt{2}
    \gamma \sum_{j \in [i]} \eps_j} < 1+3\gamma \eps_i$.

    We can now derive the lower bounds.
    Consider $\gamma \eps_i < \frac{1}{2\sqrt{2}}$ for all $i$, in which case
    $e^{\sqrt{2}\gamma \eps_i} < 1+2\sqrt{2}\gamma \eps < 2$.
    Using the upperbounds above, we can then inductively show that
    \[
        \ww_{\vG_i(\vC)} \ge (1- 2\sqrt{2}\gamma \sum_{j \in [i]} \eps) \ww_{\vG}(\vC).
    \]
    If there is some $\gamma\eps_i \ge \frac{1}{2\sqrt{2}}$, the same lowerbound
    still holds as it is at most 0.
    Finally, replacing back the definition of $\gamma$, we have
    \[
        \ww_{\vG_k}(\vC) \in \Par{\ww_{\vG}(\vC) \pm
        \frac{3 \sum_{i \in [k]} \eps_i}{\sqrt{\beta+1}} \sqrt{\ww_{\vG}(\vC)}}.
    \]
\end{proof}

\subsection{Balanced directed cut sparsifier via edge connectivity}
\label{ssec:dicut_static}
Given a weighted directed graph $\vG = (V,E,\ww)$, we denote its corresponding
undirected graph $G \defeq \und(\vG)$.
For simplicity, we set $E \defeq E_G$.
\begin{proposition} \label{prop:connum}
    For a connected undirected graph $G = (V,E,\ww)$ with edge connectivity $\kk$,
    \begin{equation} \label{eq:connum}
        \sum_{e \in E} \frac{\ww_e}{\kk_e} = n-1.
    \end{equation}
\end{proposition}
We say that $\tkk$ is a $\ell$-stretch edge connectivity estimation of $G$
if $\tkk_e \le \kk_e$ for all $e$ and 
\[
    \sum_{e \in E} \frac{\ww_e}{\tkk_e} \le \ell \cdot \sum_{e \in E}
    \frac{\ww_e}{\kk_e}.
\]

\begin{definition}
    An edge is $k$-heavy if its connectivity is at least $k$; otherwise, it is
    said to be $k$-light.
    The $k$-projection of a cut is the set of $k$-heavy edges in it.
\end{definition}
The key to our analysis is the following cut counting statement in
\Cref{thm:cpc} that is a generation of Theorem 2.3 of \cite{FunHHP11} to
weighted graphs.
Crucially, \Cref{thm:cpc} allows us to bound the number of distinct
$k$-projections polynomially, allowing us to obtain high probability
concentration (see proof of \Cref{lemma:baldicut}).

\begin{theorem}[Generalized cut projection counts, \cite{FunHHP11}] \label{thm:cpc}
    For a weighted undirected graph $G$, any $k \ge \lambda$ and any $\alpha \ge
    1$, the number of distinct $k$-projections in cuts of value at most $\alpha
    k$ in $G$ is at most $n^{\ccpc \alpha}$, where $\lambda$ is the minimum
    value of a cut in $G$ and $\ccpc \ge 2$ is a constant.
\end{theorem}
\begin{proof}[Proof sketch]
    Suppose all edge weights of $G$ are in $[1,W]$.
    We can form a graph $G'$ by rounding down each edge weight in $G$ to its
    nearest $2^k$ value for integer $k \in [0,\floor{\log W}]$ so that $\ww'_e
    > \frac{1}{2} \ww_e$.
    Now, each $k$-projection of a cut $C$ in $G$ is a $\frac{k}{2}$-projection
    of cut $C$ in $G'$.
    If $C$ is of value $\ww(C) \le \alpha k$, we also have $\ww'(C) \le \alpha
    k$.
    Applying Theorem 2.3 of \cite{FunHHP11} on $G'$ (with weighted edges split
    into unweighted multiedges) then gives us this theorem.
\end{proof}

Similar to \cite{FunHHP11}, we partition the edges in $\vG$ according to their
edge connectivities in $G$ into sets $F_0, F_1, \ldots, F_K$, where $K =
\ceil{\log \frac{\max_{e\in E} \kk_e}{\lambda}}$ where $\lambda = \min_{e \in
E(G)} \kk_e$ is the min-cut value of $G$ and
\begin{equation} \label{eq:fi}
    F_i = \{e \in E : 2^i \lambda \le \kk_e \le (2^{i+1} - 1) \lambda\}.
\end{equation}

For a dicut $\vC$ in $\vG$ induced by a partition of the
vertices $(S,V\setminus S)$, we let $C$ be the cut induced by $(S,V\setminus S)$
on $G$.
Let $F_i^{(\vC)} \defeq F_{i} \cup \vC$ and let $f_i^{(\vC)} \defeq
\ww(F_i^{(\vC)})$.
We also let $F_i^{(C)} \defeq F_{i} \cup C$ and let $f_i^{(C)} \defeq
\ww(F_i^{(C)})$.
\Cref{lemma:baldicut} shows that for each $i$, Algorithm~\ref{alg:baldicut}
achieves $\beta$-dicut approximation on subset $F_i$.

\begin{lemma} \label{lemma:baldicut}
    For any fixed $i$, with probability at least $1-\frac{\delta}{2n^2}$,
    every cut $\vC$ in $\vG$ satisfies 
    \begin{equation} \label{eq:baldicut}
        \Abs{f_i^{(\vC)}-\tf_i^{(\vC)}} \le \eps \cdot
        \sqrt{\frac{f_i^{(\vC)} f_i^{(C)}}{\beta+1}}
    \end{equation}
    where $\tf_i^{(\vC)} \defeq \ww_{\vH}(F_i^{(\vC)})$.
\end{lemma}
\begin{proof}
    By a uniform edge scaling, we may assume for simplicity that $\lambda = 1$.
    Assume w.l.o.g.~that $f_i^{(\vC)} > 0$, in which case $f_i^{(C)} > 0$.
    Then, the lowerbound in \eqref{eq:fi} ensures that $f_i^{(C)} \ge 2^i$.
    Let $\gC_{ij}$ be the set of all undirected cuts $C$ such that
    \begin{equation} \label{eq:cij}
        2^{i+j} \le f_i^{(C)} \le 2^{i+j+1} - 1 \quad \text{for } j\ge 0.
    \end{equation}

    Suppose $C \in \gC_{ij}$ and let its corresponding directed cut $\vC$.
    We let $\mu \defeq \E(\tf_i^{(\vC)}) = f_i^{(\vC)}$.
    For a directed edge $e \in F_i^{(C)}$, we let $X_e$ be the random variable
    that takes $X_e = \frac{\ww_e}{p_e}$ with probability $p_e \ge
    \rho \frac{\ww_e}{\kk_e}$ and $X_e = 0$ otherwise.
    Then, $X_e \in [0,\frac{2^{i+1}}{\rho}]$ for all $e$ and $\tf_i^{(\vC)} =
    \sum_{e \in F_i^{\vC}} X_e$.
    We show that each $X_e$ follows a sub-gaussian distribution.
    Given some fixed arbitrary $e$, we omit the subscripts here for brevity.
    Taking the moment generating function, we have
    \[
        \E[e^{t \cdot X}] = 1-p + p \cdot e^{\frac{w}{p} \cdot t},
    \]
    and $e^{t \E X} = \exp(-wt)$.
    Then,
    \[
        \E[e^{t \cdot (X - \E X)}] 
        = \exp(-wt) (1-p + p \cdot e^{\frac{w}{p} \cdot t})
        = \exp\Par{-wt + \ln(1-p+p\exp(\frac{wt}{p}))}.
    \]
    Define $h \defeq wt$ and function 
    $f(h) = -h + \ln(1-p+p\exp(\frac{h}{p}))$.
    We analyze this function on three separate pieces $(-\infty,-1) \cup [-1,1]
    \cup (1,\infty)$.
    We consider first when $h$ is non-negative, in which case $f(h) \le
    \frac{1-p}{p} h$.
    When $h > 1$, as $h < h^2$, $f(h) \le \frac{1-p}{p} h \le \frac{1}{p} h^2$.
    Consider $h \in [-1,1]$.
    A Taylor series at $h=0$ gives $f(h) = \frac{1-p}{2p} h^2 + O(h^3)$.
    Now, by $f(-1) \le f(1) = \ln(1-p+pe^{1/p}) \le \frac{1}{p}$, we have $f(h)
    \le \frac{1}{p} h^2$.
    Finally, when $h < -1$, $f(h) < -h < \frac{1}{p} h^2$.
    Thus, combining these cases gives us
    \[
        \E[e^{t \cdot (X - \E X)}] \le \exp(f(wt)) 
        \le \exp\Par{\frac{w^2}{p} t^2}
        \le \exp\Par{\frac{wk}{\rho} t^2}
        \le \exp\Par{\frac{2^{i+1}w}{\rho} t^2},
    \]
    proving that $X_e$ is $\sqrt{\frac{2^{i+2}\ww_e}{\rho}}$-sub-gaussian.
    Now, by independency, $\sum_{e \in F_i^{(\vC)}} X_e$ is $R$-sub-gaussian for
    \[
        R = \sqrt{\frac{2^{i+2}}{\rho} \sum_{e \in F_i^{(\vC)}} \ww_e} 
        = \sqrt{\frac{2^{i+2}}{\rho} f_i^{(\vC)}}.
    \]
    Then, sub-Gaussianity gives the following concentration inequality (see
    Proposition 2.5.2, \cite{Vershynin18})
    \begin{align*}
        &\pr\Par{ \Abs{\sum_{e \in F_i^{(\vC)}} (X_e - \E X_e)}
        > \eps \sqrt{\frac{f_i^{(\vC)} f_i^{(C)}}{\beta+1}}}
        \\
        \le&
        2\exp\Par{ - \frac{\csbg \eps^2 f_i^{(\vC)} f_i^{(C)}}{(\beta+1)R^2}}
        \le
        2\exp\Par{ - \frac{\csbg \rho \eps^2 f_i^{(\vC)} f_i^{(C)}}{2^{i+3}(\beta+1) f_i^{(\vC)}}}
        \\
        \stackrel{(a)}{\le}&
        2\exp\Par{ - \frac{\ccpc \log\frac{8n}{\delta} f_i^{(C)}}{2^{i-2}}}
        \stackrel{(b)}{\le}%
        2\exp\Par{ - \log\frac{8n}{\delta} \cdot \ccpc \cdot 2^{j+2}}
        \\
        \stackrel{(c)}{\le}&
        \frac{\delta}{4n^{\ccpc \cdot 2^{j+2}}}
    \end{align*}
    for any cut $C \in \gC_{ij}$, where $\csbg$ is some universal constant.
    Here, (a) follows by $\rho = \cbal \cdot \eps^{-2} (\beta+1)
    \log\frac{8n}{\delta}$ where we set $\cbal \ge \frac{32 \ccpc}{\csbg}$, (b)
    follows by the lowerbound in \eqref{eq:cij}, and (c) follows from $\delta <
    1$.

    Since edges in $F_i^{(C)}$ are $2^i$-heavy in $G$, by \Cref{thm:cpc} with
    $k = 2^i$ and $\alpha k = 2^{i+j+1}-1$, the number of distinct $X_i^{(C)}$
    sets for cuts $C \in \gC_{ij}$ is at most $n^{\ccpc \cdot \alpha k/k} <
    n^{\ccpc \cdot 2^{j+1}}$ for some constant $\ccpc \ge 2$.
    By union bounding over all distinct $F_i^{(C)}$, all cuts in $\gC_{ij}$
    satisfy \eqref{eq:baldicut} with probability of at least $1-
    \frac{\delta}{4n^{2^{j+1}}}$.
    Since
    \[
        n^{-2} + n^{-4} + \ldots + n^{-2 \cdot 2^j} + \ldots \le 2n^{-2},
    \]
    the lemma follows by taking a union bound over all $j \ge 0$.
\end{proof}

\begin{algorithm2e}[ht!] \label{alg:baldicut}
    \caption{$\SDCalgo(\vG=(V,E,\ww),\beta,\eps,\delta,\tkk)$}
    \codeInput $\vG=(V,E,\ww)$ a simple weighted directed graph with $G \defeq
    \und(\vG)$, $\beta \in \R_{\ge 0}$, $\eps \in (0,\frac{1}{2})$, $\delta \in
    (0,1)$, $\tkk$ an edge connectivity estimation of $G$\;
    $\rho \defeq \cbal \cdot \eps^{-2} (\beta+1) \log \frac{8n}{\delta}$\;
    \ForEach{$e \in E$}{
        With probability $p_e = \rho \cdot \frac{\ww_e}{\tkk_e}$, set $\ww_e' \gets \frac{1}{p_e} \ww_e$; otherwise
        $\ww_e' \gets 0$\;
    }
    \Return{$\vH = (V,E,\ww')$}\;
\end{algorithm2e}

\begin{proof}[Proof of \Cref{thm:baldicut}]
For any directed cut $\vC$, let $E^{(\vC)} \defeq E \cup \vC$ and let
$e^{(\vC)} \defeq \ww(E^{(\vC)}), \te^{(\vC)} \defeq \ww_{\vH}(E^{(\vC)})$.
We also let $E^{(C)} \defeq E \cup C$ and let $e^{(C)} \defeq \ww(E^{(C)}), 
\te^{(C)} \defeq \ww_{\vH}(E^{(C)})$.

Since there are at most $n^2$ edges in $\vG$ (and $G$), the number of nonempty
sets $F_i$ is also at most $n^2$. 
Then, by a union bound of all values of $i$ where $F_i$ is nonempty, we have
with probability at least $1-\frac{\delta}{2}$,
\[
    \sum_{i=0}^K \Abs{f_i^{(\vC)}-\tf_i^{(\vC)}}
    \le
    \eps \cdot \sum_{i=0}^K \sqrt{\frac{f_i^{(\vC)} f_i^{(C)}}{\beta+1}}
\]
for all cuts $C$.
Under the assumption that this inequality holds, we get
\begin{align*}
    |e^{(C)} - \te^{(C)}| 
    &=
    \Abs{\sum_{i=0}^K f_i^{(C)} - \sum_{i=0}^K \tf_i^{(C)}}
    \\
    &\le
    \sum_{i=0}^K \Abs{f_i^{(\vC)}-\tf_i^{(\vC)}}
    \quad \le
    \eps \cdot \sum_{i=0}^K \sqrt{\frac{f_i^{(\vC)} f_i^{(C)}}{\beta+1}}
    \\
    &\le
    \frac{\eps}{\sqrt{\beta+1}} \cdot \sqrt{\Par{\sum_{i=0}^K f_i^{(\vC)}} \Par{\sum_{i=0}^K f_i^{(C)}}}
    && (\text{Cauchy-Schwarz})
    \\
    &\le
    \frac{\eps}{\sqrt{\beta+1}} \cdot \sqrt{e^{(\vC)} e^{(C)}}
\end{align*}
as required.
Note that when $\vC$ is $\alpha$-balanced, $e^{(C)} \le (\alpha+1) \cdot
e^{(\vC)}$.

We now bound the number of edges in $\vH$.
For all $e \in E$, let $Y_e$ be an independent Bernoulli random variable with
probability $p_e$.
Then, the expected number of edges is 
\[
    \E(\sum_e Y_e) = \sum_e \E Y_e = \sum_e p_e 
    \le \rho \ell n = \cbal \cdot \eps^{-2} (\beta+1) \ell
    n\log\frac{8n}{\delta}.
\]
A tail bound (e.g., Hoeffding's inequality) gives us the desired sparsity
result with a failure probability $\le \frac{\delta}{2}$ for any valid $n$, as
long as $\cbal \ge 8$.
\end{proof}

We note that an expander naturally ensures that the connectivity of an edge is
at least the minimum degree of its incident vertices (see \Cref{lemma:expconn}).
This then implies a simple sampling scheme for directed cut sparsification
almost identical to that of \Cref{lemma:entrysample}.

\begin{lemma} \label{lemma:expconn}
    For $\phi$-expander $G=(V,E,\ww)$ it satisfies that for ever edge $e \in E$,
    the edge connectivity $\kk_e \ge \phi \cdot \min_{v \ni e} \dd_v$.
\end{lemma}
\begin{proof}
    For any edge $e = (u,v)$, let $C_{\star} = (U, V \setminus U)$ be the
    minimum cut that separates $u,v$.
    Assume w.l.o.g. that $u \in U$ and $\Vol(U) \le \Vol(V \setminus U)$.
    Then, by definition of edge connectivity and expander, \[
        \kk_e = \ww(C_{\star}) \ge \phi \cdot \vol(U) 
        \ge \phi \cdot \min(\dd_u,\dd_v).
    \]
\end{proof}

\begin{corollary} \label{cor:expdicutapprox}
    There is an algorithm that given a weighted directed graph $\vG=(V,E,\ww)$
    with $G \defeq \und(\vG)$ a $\phi$-expander,  and $\eps,\delta \in (0,1),
    \beta \ge 1$, it returns in time $O(m)$ a weighted subgraph $\vH$ such that
    with probability at least $1-\delta$, $\vH$ is a $(\beta,\eps)$-dicut
    sparsifier of $\vG$ and the number of edges $\nnz(\ww_{\vH}) =
    O(\eps^{-2}\phi^{-1}\beta n\log\frac{n}{\delta})$.
\end{corollary}
\begin{proof}
    Our algorithm simply performs \SDC (Algorithm~\ref{alg:baldicut}) with the
    same inputs and a connectivity estimation $\tkk$ defined by $\tkk_e =
    \frac{\phi}{\dd_u^{-1}+\dd_v^{-1}}$ for every $e = (u,v) \in E$.
    As $\min_{v \ni e} \dd_v \ge \frac{1}{\dd_v^{-1} + \dd_u^{-1}}$ for any
    $e=(u,v) \in E$. 
    \Cref{lemma:expconn} entails that $\kk_e \ge \tkk_e$.
    It then suffices to bound the total stretch $\ell$ of $\tkk$.
    W.l.o.g., we assume $G$ is connected; otherwise we can bound the total
    stretch for each connected component by $O(\phi^{-1})$.
    We get by definition that $\frac{1}{\tkk_e} = \phi^{-1} (\frac{1}{\dd_u} +
    \frac{1}{\dd_v})$. 
    The total sum satisfies
    \[ 
        \sum_{e \in E} \frac{\ww_e}{\tkk_e}
        = \phi^{-1} \sum_{e \in E} \ww_e (\sum_{v \ni e} \frac{1}{\dd_v})
        = 2\phi^{-1} \sum_{v \in V} \frac{1}{\dd_v} \sum_{e \ni v} \ww_e 
        = 2\phi^{-1} n,
    \]
    which guarantees $O(\phi^{-1})$ total stretch by \eqref{eq:connum}.
\end{proof}

Finally, we obtain our simple directed cut sparsification algorithm by combining
\Cref{cor:expdicutapprox} with expander decompositions \Cref{prop:ex_partition}.
\begin{theorem} \label{thm:dicutapprox_final}
    There is an algorithm that 
    given $\vG=(V,E,\ww)$ a simple weighted directed graph and parameters
    $\eps,\delta \in (0,1)$, $\beta \ge 1$, it returns in time 
    $O\Par{m \log^6(n)\log\Par{\frac n \delta} + m\log(n)\log(W)}$
    a weighted subgraph $\vH$ such that with probability at least $1-\delta$,
    $\vH$ is a $(\beta,\eps)$-dicut sparsifier of $\vG$ and the number of edges
    $\nnz(\ww_{\vH}) = O(\eps^{-2}\beta n \log^3(n) \log(W)
    \log(\frac{n}{\delta}))$.
\end{theorem}
\begin{proof}[Proof sketch]
    The algorithm first perform an expander decomposition on the corresponding
    undirected graph using \Cref{prop:ex_partition} with probability parameter
    $\frac{\delta}{2}$.
    This induces a decomposition of the original directed graph where each
    corresponding undirected graph is an expander.
    It then sample within each subgraph with failure probability
    $\frac{\delta}{2m}$.
    Note that $m$ is a crude upperbound on the total number of expanders in the
    decomposition.
    Thus, we are guaranteed a success probability of $1-\delta$ by a union
    bound.
    \Cref{lemma:dicut_union} then ensures the union of all the
    $(\beta,\eps)$-dicut sparsifiers is a $(\beta,\eps)$-dicut sparsifier of the
    original $\vG$.
    Finally, the expansion $\phi = \frac{1}{O(\log^2)}$, the vertex coverage $J
    = O(\log n \log W)$ and our probability choice gives the sparsity bound of
    \[
        \nnz(\ww_{\vH}) 
        = O(\eps^{-2}\beta \phi^{-1} nJ \log\frac{n}{\delta})
        = O\Par{\eps^{-2}\beta n \log^3(n)\log(W)\log\Par{\frac{n}{\delta}}}.
    \]
\end{proof}

\subsection{Dynamic directed cut sparsification}

Since our $\beta$-dicut approximation definition satisfies union property
(\Cref{lemma:dicut_union}), the reductions presented in \Cref{ssec:dynexpdecomp}
can, again, be applied for dynamic dicut sparsifications.
This allows us to focus on deriving a decremental algorithm on expanders for
this task.
Specifically, we show that Algorithms~\ref{alg:dsp},~\ref{alg:dsd} for
dynamically maintaining $\eps$-degree approximations applies here as well.

\begin{lemma} \label{lemma:decexpdicut}
    Suppose $\vG=(V, E, \ww)$ is a weighted directed graph with weight ratio
    $\frac{\max_{e \in E} \ww_e}{\min_{e \in E} \ww_e} \le 2$ undergoing edge
    deletions only.
    There is an algorithm that given $\eps,\phi \in (0,1)$, $\beta
    \ge 1$ and an oblivious online sequence of edge deletions, preprocesses in
    time $O(m)$, maintains explicitly a reweighted directed subgraph $\vG'$ of
    $\vG$ with $\nnz(\ww') = O(\eps^2 \beta \phi^{-1} N \log n)$ where $N$ is
    the number of non-trivial vertices currently.
    The algorithm has worst-case recourse $O(\eps^{-2} \phi^{-1}\beta \log n)$
    and worst-case update time $O(\eps^{-2}\phi^{-1}\beta \log^2 n)$.
    Further more, with high probability, whenever $G = \und(\vG)$ becomes a
    $\phi$-expander, $\vH$ is a $(\beta,\eps)$-dicut sparsifier of $\vG$.
\end{lemma}
\begin{proof}
    Our proof is analogous to that of \Cref{lemma:decexpdeg}.
    We preprocess using \DSP (Algorithm~\ref{alg:dsp}) and handles edge
    deletions using \DSD (Algorithm~\ref{alg:dsd}) with error parameter $\eps' =
    \Theta(\eps \sqrt{\frac{\beta}{\phi}})$.
    The sparsity, preprocessing time, worst-case recourse and update time
    guarantees all follows directly by the same proof of \Cref{lemma:decexpdeg}.

    Now, all there left to show is the dicut approximation guarantees.
    The crux is, again, that we can ``reuse'' the random choices throughout the
    updates due to the adversary being oblivious.
    Similar to before, we will union bound over at most $n^2$ edge deletions.
    Consider a time after some updates (or initially) when $\vG$ satisfies that
    $G = \und(\vG)$ is a $\phi$-expander.
    The probability of an edge $e = (u,v)$ being sampled is then 
    \[
        p_e = 
        1 - (1-\frac{2L\rho}{\dd_{h(e)}})(1-\frac{2L\rho}{\dd_{t(e)}}) = 
        4L^2\rho^2\frac{(\dd_{h(e)} + \dd_{t(e)} - 1)}{\dd_{h(e)}\dd_{h(e)}}
        \ge 
        \frac{2L\rho}{\min_{v \in e} \dd_v}
    \]
    when $\min_{v \in e} \dd_v > 2L\rho$; otherwise $p_e = 1$ which is as
    large as possible.
    Note that $2L \ge \ww_e$ and $\rho = O(\eps^{-2} \phi^{-1} \beta \log n)$. 
    When combined with \Cref{lemma:expconn}, an argument similar to that of
    \Cref{cor:expdicutapprox} than guarantees $(\beta,\eps)$-dicut approximation
    with our desired high probability guarantees, assuming that $\cds$ is a
    sufficiently large constant.
\end{proof}

We can now complete the proof of our dynamic sparsifier with amortized update time
guarantees.
\begin{proof}[Proof of \Cref{thm:dyndicut_for}]
    We again apply the framework in \Cref{lemma:red_dynexp}.
    Note that $\beta$-dicut approximation satisfies the union property by
    \Cref{lemma:dicut_union}.
    Similar to \Cref{thm:dynspecext_for}, we set the expansion to $\phi =
    1/\Omega(\log^2 n)$, then, the algorithm in \Cref{lemma:decexpspec_ext} has
    preprocessing time $P(m,\eps) = O(m)$, update time $T(N,\eps) \le
    O(\eps^{-2} \beta \log^4 n)$, sparsity $S(N,\eps) \le O(\eps^{-2} \beta N
    \log^3 n)$, where $N$ is the number of non-trivial vertices in the
    decremental directed graphs with corresponding undirected expander, and $n$
    is the number of vertices in the entire graph $\vG$.
    Note that there is no query operation as our sparsifier is explicit.
    Then, our overall guarantees are given by \Cref{lemma:decexpspec_ext}.
    To summarize, our dynamic algorithm has preprocessing time of $O(m \log^7
    n)$, amortized update time $O(\eps^{-2} \beta \log^7 n)$, number of edges in
    the sparsifier $O(\eps^{-2} \beta n \log^5 n \log W)$.
\end{proof}

\subsection{Dynamic algorithm with worst-case update time}

We now present a simple sparsification algorithm by \cite{AbrahamDKKP16} we are
using for our dynamic dicut sparsifier with worst-case guarantees.
\SDM (Algorithm~\ref{alg:dicutmst}) is a half-sparsification algorithm that
repeatedly sparsify the graph by a multiplicative factor of roughly
$\frac{1}{2}$ on the number of edges.
Within each iteration, \SDMO (Algorithm~\ref{alg:dicutmst_once}) uses a bundle
of approximate maximum spanning forests, called $\alpha$-MSFs, to find edges
with sufficiently small connectivity to sample on.
\begin{definition}
    A subgraph $T$ of an undirected graph $G$ is an $\alpha$-MSF for $\alpha \ge
    1$ if for every edge $e = (u,v)$ of $G$, there is a path $\pi$ from $u$ to
    $v$ in $T$ such that $\ww(e) \le \alpha \ww(f)$ for every edge $f$ on $\pi$.
\end{definition}
Note that $T$ need not be a spanning forest.
A maximum spanning forest is a $1$-MSF.

\begin{definition}
    A $t$-bundle $\alpha$-MSF of an undirected graph $G$ is the union $B =
    \bigcup_{i=1}^t T_i$ of a sequence of grpah $T_1 ,\ldots, T_t$ such that,
    for every $1 \le i \le t$, $T_i$ is an $\alpha$-MSF of $G \setminus
    \bigcup_{j=1}^{i-1} T_j$.
\end{definition}

\begin{algorithm2e}[ht!] \label{alg:dicutmst}
    \caption{$\SDMalgo(\vG=(V,E,\ww),\beta,\eps,\delta,\gamma)$}
    \codeInput $\vG=(V,E,\ww)$ a simple weighted directed graph with $G \defeq
    \und(\vG)$, $\beta \in \R_{\ge 0}$, $\eps \in (0,\frac{1}{2})$, $\delta \in
    (0,1)$, $\gamma \in \R_{> 0}$\;
    $I \gets \ceil{\log \gamma}$, $i \gets 0$\; 
    $\vG_0 \gets \vG$, $\vB_0 \gets$ an empty subgraph of $\vG_0$ on $V$\;
    \While{$i < I$ and $|E(G_i)| > 16 \ln \frac{8m}{\delta})$ }{
        \label{line:dicutmst_while}
        $(\vH_i,\vB_i) \gets \SDMOalgo(\vG_i, \beta, \frac{\eps}{3I},
        \frac{\delta}{2^{i+1}})$\;
        $\vG_{i+1} \gets \vH_i \setminus \vB_i$, $i \gets i+1$\;
    }
    $\vH \gets \bigcup_{1\le j \le i-1} \vB_j \cup \vG_i$\;
    \Return{$\vH,\{\vB_j\}_{j=1}^{i-1},\vG_i$}\;
\end{algorithm2e}

\begin{algorithm2e}[ht!] \label{alg:dicutmst_once}
    \caption{$\SDMOalgo(\vG=(V,E,\ww),\beta,\eps,\delta)$}
    \codeInput $\vG=(V,E,\ww)$ a simple weighted directed graph with $G \defeq
    \und(\vG)$, $\beta \in \R_{\ge 0}$, $\eps \in (0,\frac{1}{2})$, $\delta \in
    (0,1)$\;
    $t \gets 4\cbal \cdot \eps^{-2} \alpha (\beta+1) \log \frac{16n}{\delta}$\;
    Let $B$ be a $t$-bundle $\alpha$-MSF of $G$ and $\vB$ its corresponding
    directed subgraph in $\vG$\; 
    Set $\tkk \in \R^{E}$ by $\tkk_e = \frac{t}{\alpha}$ if $e \in \vG \setminus \vB$ and
    $\tkk_e = \ww_e$ otherwise\;
    $\vH \gets \SDCalgo(\vG,\beta,\eps,\frac{\delta}{2},\tkk)$\;
    \Return{$(\vH,\vB)$}\;
\end{algorithm2e}

We summarize the guarantees for \SDMO and \SDM in the static settings in
\Cref{lemma:sdmo,lemma:sdm}.
\begin{lemma} \label{lemma:sdmo}
    Given inputs satisfying the conditions of \SDMO
    (Algorithm~\ref{alg:dicutmst_once}), the algorithm returns with probability
    $1-\delta$ a reweighted directed subgraph $\vH$ of $\vG$ that is a
    $(\beta,\eps)$-dicut approximation of $\vG$ and a subgraph $\vB$ of $\vH$
    such that if $|E_{\vG} - E_{\vB}| \ge 8 \ln\frac{n}{\delta}$ the number of
    edges satisfies 
    \[
        |E_{\vH}| \le |E_{\vB}| + \frac{1}{2}(|E_{\vG}| - |E_{\vB}|).
    \]
\end{lemma}
\begin{proof}
    For each directed edge $e \in E(\vG) \setminus E(\vB)$, by definition of
    $t$-bundle $\alpha$-MSFs, there exists $t$ edge disjoint paths in the
    undirected graph $G$ that connects its incident vertices with minimum edge
    weight at least $\alpha^{-1} \ww_e$.
    As such, the undirected connectivity of $e$ satisfies $\kk_e \ge
    t\alpha^{-1} = 4 \rho = \tkk_e$ for $\rho$ as defined in the invocation of
    \SDC with probability parameter $\frac{\delta}{2}$.
    Then $e$ is sampled with probability exactly $\frac{1}{4}$.
    When $e \in E(\vB)$, as $\frac{\ww_e}{\tkk_e} = 1$ and $\rho > 1$, the edge
    $e$ is also preserved.
    This guarantees that $\vB$ is a subgraph of $\vH$.
    The dicut approximation guarantees follows directly by \Cref{thm:baldicut}.
    As we remove an edge in $E_{\vG} \setminus E_{\vB}$ with probability
    $\frac{3}{4}$, Hoeffding's inequality then gives a failure probability of
    at most $\exp(-\frac{1}{8} k)$ for the inequality on the number of edges,
    where $k = |E_{\vG} - E_{\vB}|$.
    It suffices to have $k \ge 8 \ln \frac{2}{\delta}$ for the failure
    probability to be $\le \frac{\delta}{2}$.
    Combine with the failure probability of at most $\frac{\delta}{2}$ for
    dicut approximation, we get our desired probability bounds.
\end{proof}

\begin{lemma} \label{lemma:sdm}
    Given inputs satisfying the conditions of \SDM
    (Algorithm~\ref{alg:dicutmst}), with probability $1-\delta$ the algorithm 
    runs for $\min\Par{\ceil{\log \gamma}, \ceil{\log\Par{\frac{m}{8
    \ln(8m\delta^{-1})}}}}$
    number of iterations before termination and returns a reweighted directed
    subgraph $\vH$ of $\vG$ that is a $(\beta,\eps)$-balanced dicut sparsifier
    of $\vG$ with
    \[
        O\Par{ \sum_j |B_j| + \frac{m}{\gamma}+\log \frac{m}{\delta}}
    \]
    number of edges, and the number of edges in $\vG_i$ is at most 
    $O(m/\gamma + \log \frac{m}{\delta})$.
\end{lemma}
\begin{proof}
    The probability bound follows by a union bound over events of
    success probabilities $1-\frac{\delta}{2^{i+1}}$ each.
    We then condition on the success of all calls to
    \SDMO~(Algorithm~\ref{alg:dicutmst_once}).

    We show first that $\vH$ is a $(\beta,\eps)$-balance dicut approximation of $\vG$.
    Let $k$ be the last $i$ when the while loop on
    line~\ref{line:dicutmst_while}.
    By our choice of error parameter $\frac{\eps}{3I}$, each $\vH_i$ is a
    $(\beta,\frac{\eps}{3I})$-dicut sparsifier of $\vG_i$ by \Cref{lemma:sdmo}.
    \Cref{lemma:dicut_union} then gives that $(\bigcup_{j=1}^i \vB_i) \cup
    \vG_{i+1}$ is also a $(\beta,\frac{\eps}{3I})$-dicut sparsifier of
    $(\bigcup_{j=1}^{i-1} \vB_i) \cup \vG_i$.
    We can then apply \Cref{lemma:dicut_trans} across $k \le I$ levels to get
    the $(\beta,\eps)$ approximation.

    To show the second term in the total number of iterations, we consider the
    edge reductions per iteration. 
    For any $i \le \ceil{\log m}$, if $|E(G_i)| > 16 \ln \frac{8m}{\delta} \ge 16
    \ln \frac{2^{i+2}}{\delta}$, \Cref{lemma:sdmo} guarantees that
    \[
        |E_{\vG_{i+1}}| \le \frac{1}{2} (|E_{\vG_i}|-|E_{\vB_i}|) 
        \le \frac{1}{2} |E_{\vG_i}|.
    \]
    Then, there also can be at most $\ceil{\log\Par{\frac{m}{8
    \ln(8m\delta^{-1})}}}$ iterations.
    Using the two terms in the minimum, after the while loop on
    line~\ref{line:dicutmst_while} terminates, we are guaranteed to have
    $|E_{\vG_i}| \le \max(\frac{m}{\gamma}, 16 \ln \frac{8m}{\delta})$ that
    further gives us the last two terms in the sparsity bound.
\end{proof}

\begin{remark}
    It is possible to bypass the use of weak transitivity property
    (\Cref{lemma:dicut_trans}) and prove \SDM can generate the stronger
    notion of $(\beta,\eps)$-dicut sparsifier.
    To do so, one can adapt a resparsification analysis framework, similar
    to the one developed by Kyng, Pachocki, Peng, and Sachdeva \cite{KyngPPS17}
    for analyzing resparsification problems for undirected spectral
    sparsification.
    Since adapting such analyses is more involved for cut sparsification, we omit
    this result.
\end{remark}

For the de-amortized regime, we leverage a dynamic spanning forest
algorithm with good worst-case guarantees.
\begin{theorem}[\cite{KapronKM13,GibbKKT15}] \label{thm:kkm}
There is a fully dynamic algorithm for maintaining a spanning
forest $T$ of an undirected graph $G$ with worst-case update time $O(\log^4 n)$.
Every time an edge $e$ is inserted into $G$, the only potential change to $T$ is
the insertion of $e$.
Every time an edge $e$ is deleted from $G$, the only potential change to $T$ is
the removal of $e$ and possibly the addition of at most one other edge to $T$.
The algorithm is correct with high probability against an oblivious adversary.
\end{theorem}

For a weighted undirected graph $G = (V,E,\ww)$, with edge weights, w.l.o.g., in
$[1,W]$ we define for each integer $i$ in $[0,\floor{\log W}]$ a subset $E_i
\defeq \{e \in E | 2^i \le \ww(e) < 2^{i+1} \}$.
Suppose $F_i$ is a spanning forest of $G[E_i]$, the union of these forests form
a 2-MSF of $G$.
\begin{lemma}[Lemma 5.11, \cite{AbrahamDKKP16}]
    $T = \bigcup_{i=0}^{\floor{\log W}} F_i$ is a 2-MSF of $G$.
\end{lemma}

Abraham et al.~\cite{AbrahamDKKP16} presented an algorithm that maintains a
$t$-bundle $\alpha$-MSF $B = \bigcup_{1 \le i \le t} T_i$ by maintaining an
$\alpha$-MSF for each $G \setminus \bigcup_{j=1}^{i-1} T_j$.
The guarantees in \Cref{cor:tbund_worst} follow by noticing that
one update at level $i$ only incurs at most another update at level $i+1$.

\begin{corollary} \label{cor:tbund_worst}
    There is a fully dynamic algorithm for maintaining a
    $t$-bundle 2-MSF $B$ of size $O(tn\log W)$ with worst-case update time
    $O(t\log^4 n)$ against an oblivious adversary.
    Every update in $G$ has a recourse of at most 1 on the graph $G \setminus
    B$.
\end{corollary}

We now show how to turn the static half-sparsification algorithm into a fully
dynamic algorithm that supports $\Theta(n^2)$ number of edge updates.
By keeping 2 instances of this dynamic algorithm on two disjoint subgraphs, one
can process all insertions on one copy until the other is completely empty.
Once this happens, the algorithm then swap to inserting only into the other
copy.
The correctness is ensured by the union property (\Cref{lemma:dicut_union}).
This reduction allows us to make the algorithm support arbitrary length of
updates without changing the asymptotic update time.
Below, we show a fully dynamic version of \SDMO
(Algorithm~\ref{alg:dicutmst_once}) with worst-case update time guarantees.

\begin{lemma} \label{lemma:dynsdmo}
    Suppose $\vG=(V, E, \ww)$ is a weighted directed graph with weight ratio
    $\frac{\max_{e \in E} \ww_e}{\min_{e \in E} \ww_e} \le W$.
    There is an algorithm that given $\eps,\delta \in (0,1)$, $\beta \ge 1$, and an
    oblivious online sequence of at most $4n^2$ edge insertions and deletions,
    preprocesses in time $O(\eps^{-2} \beta M \log^4 n \log \frac{n}{\delta})$, where $M$
    is the initial number of edges, maintains
    explicitly a reweighted directed subgraph $\vH$ that is a
    $(\beta,\eps)$-dicut approximation of $\vG$ and a subgraph $\vB$ of $\vH$
    with number of edges \[
        m_{\vH} \le \frac{1}{2}m + \frac{1}{2} m_{\vB}, \quad
        m_{\vB} = O(\eps^{-2}\beta n \log W \log \frac{n}{\delta}) ,
    \]
    and has worst-case recourse of 1 and worst-case update time
    $O(\eps^{-2} \beta \log^4 n \log \frac{n}{\delta})$ per update.
    Our guarantees hold with probability at least $1-\delta$ over the entire
    sequence of updates.
\end{lemma}
\begin{proof}
    Our algorithm maintains a $t$-bundle 2-MSF $B$ of $G =
    \und(\vG)$ for $t = O(\eps^{-2}\beta \log \frac{n}{\delta'})$ using the
    algorithm in \Cref{cor:tbund_worst}, initialized by processing all initial
    edges at the preprocessing phase of the algorithm.
    Let $\vB$ be the corresponding dynamic digraph of $B$.
    An initial sparsifier $\vH$ is computing using the same sampling procedures
    in \SDMO (Algorithm~\ref{alg:dicutmst_once}).
    For any type of edge update, we first propagate this update into our
    dynamic $t$-bundle 2-MSF instance $B$, which incurs a recourse of at most
    1 on $\vG \setminus \vB$.
    Note that the probability of any other edge being sampled is not affected by
    this change.
    Suppose first that it incurs an edge insertion into $\vG \setminus \vB$.
    Our oblivious adversary setting allows us to just perform a sampling on
    this edge with $\frac{1}{4}$ probability (see proof of \Cref{lemma:sdmo}).
    As for an edge deletion on $\vG \setminus \vB$, we can simply remove it from
    our sparsifier if it is currently present $\vH$.
    The approximation guarantee at each time is then given by \Cref{lemma:sdmo}.
    By choosing $\delta'$ in to be $\frac{\delta}{5n^2}$ in our choice of $t$,
    a union bound over preprocessing and $4n^2$ updates, guarantees
    $(\beta,\eps)$-dicut approximation over the entire sequence of
    updates.

    By our dynamic algorithm, it is immediate that our worst-case recourse is 1
    and our worst-case update time is dominated by the worst-case update time of
    $B$, which is $O(\eps^{-2}\beta \log^4 n \log \frac{n}{\delta})$ given our
    choice of $t$.
    Our preprocessing time is also dominated by the initialization time of $B$,
    which runs in $m \cdot O(t \log^4 n) = O(\eps^{-2} \beta m \log^4 n \log
    \frac{n}{\delta})$.
    Our sparsity guarantee follows by \Cref{lemma:sdmo} and the size of
    $B$ given by \Cref{cor:tbund_worst}.
\end{proof}

Finally, we complete our proof of \Cref{thm:dyndicut_worst_for} by showing a
fully dynamic version of \SDM (Algorithm~\ref{alg:dicutmst}) with worst-case
update time guarantees over $\Theta(n^2)$ updates.
\begin{lemma} \label{lemma:dynsdm}
    Suppose $\vG=(V, E, \ww)$ is a simple weighted directed graph with weight
    ratio $\frac{\max_{e \in E} \ww_e}{\min_{e \in E} \ww_e} \le W$.
    There is an algorithm that given $\eps,\delta \in (0,1)$, $\gamma,\beta \ge
    1$, and an oblivious online sequence of at most $4n^2$ edge insertions and
    deletions, preprocesses in time $O(\eps^{-2} \beta m \log^5 n \log^3
    \gamma)$, maintains explicitly a reweighted directed subgraph $\vH$ that is
    a $(\beta,\eps)$-balanced dicut approximation of $\vG$ with number of edges 
    \[
        m_{\vH} = O(\eps^{-2} \beta n \log n \log W \log^3 \gamma +
        \frac{m}{\gamma}),
    \]
    and has worst-case recourse of $O(\log \gamma)$ and worst-case update time
    $O(\eps^{-2} \beta \log^5 n \log^3 \gamma)$ per update.
    Our guarantees hold with high probability over the entire sequence of
    updates.
\end{lemma}
\begin{proof}
    We maintain $I = \ceil{\log \min(\gamma, \frac{n^2}{C\log n})}$ instances of
    the dynamic version of \SDMO (Algorithm~\ref{alg:dicutmst_once}) for some
    constant $C \ge 16$ that depends on our high probability guarantees, as in
    \SDM (Algorithm~\ref{alg:dicutmst}).
    Since $\frac{\delta}{2^{i+1}}$ is at least $\frac{\delta}{2n^2}$, we simply
    apply this choice of probability over all iterations with $\delta =
    \frac{1}{n^C}$.
    Note that we do not check of the edge number condition in the while loop of
    \SDM, as the number of edges in the dynamic digraph can change
    drastically overtime.

    For each $i \in [I]$, we maintain a dynamic $(\beta,\frac{\eps}{3I})$-dicut
    sparsifier $\vH_i$ of $\vG_i = \vH_{i-1} \setminus \vB_{i-1}$ and a dynamic
    $t$-bundle 2-MSF $\vB_i \subseteq \vH_i$ as in \Cref{lemma:dynsdmo}.
    Our final sparsifier is the digraph $\vH = \bigcup_{i=[I]} \vB_i \cup
    \vG_k$.
    The overall approximation is guaranteed by the weak transitivity
    property \Cref{lemma:dicut_trans} using our choice of $\eps_i =
    \frac{\eps}{3I}$.
    For the sparsity guarantees, we get from an inductive argument that
    $m_{\vG_k} \le \frac{m}{2^I} = \max(\frac{m}{\gamma}, C\log n)$ using the
    sparsity guarantees from \Cref{lemma:dynsdmo}.
    As we further have each $\vB_i$ has $O(\eps^{-2}\beta n \log n \log W \log^2
    \gamma)$ edges, the total number of edges in $\vH$ is $O(\eps^{-2} \beta n \log n
    \log W \log^3 \gamma + \frac{m}{\gamma})$.
    The high probability guarantee follows by a union bound over preprocessing
    and the entire sequence of updates.

    For each edge update, our algorithm propagate this change to $\vH_1$ first.
    Subsequently, all updates in $\vH_i$ is passed down to $\vH_{i+1}$.
    Notice that the recourse of each $\vH_i$ is only 1 by \Cref{lemma:dynsdm}.
    Thus, we only need to process 1 update for each $\vH_i$. 
    This results in an overall worst-case recourse of $O(\log \gamma)$ and
    worst-case update time of $O(\eps^{-2} \beta \log^5 n \log^3 \gamma)$ per
    update.
    At preprocessing phase, we simply execute the entire \SDM and initialize all
    instances of our dynamic half-sparsifier $\vH_i$ with $\vB_i$.
    This results in a preprocessing time of $O(\eps^{-2} \beta m \log^5 n \log^3
    \gamma)$.
\end{proof}

\section{Directed spectral sparsification against an adaptive adversary
via partial symmetrization}
\label{sec:adaptive}

In this section, we consider the problem of fully dynamic degree balance
preserving directed spectral sparsification against an adaptive adversary.
This is a particularly difficult question, as the state-of-the-art adaptive
\emph{undirect} spectral sparsifier requires $O(\log n)$ multiplicative
approximation \cite{BernsteinvdBPGNSSS22}, while any current notions of directed
spectral sparsification do not permit approximation factor $\eps >1$.
As was observed by Kyng, Meierhans, and Probst Gutenberg~\cite{KyngSPG22}, a
weaker notion of approximation, i.e., approximate pseudoinverse
(\Cref{def:app_pinv}), suffices for solving directed Eulerian Laplacian systems.

The key observation in \cite{KyngSPG22} is that a $\beta$-\emph{partial
symmetrization} $\vG^{(\beta)} \defeq \beta \cdot G \cup \vG$ of a Eulerian
directed graph $\vG$ is a good approximate pseudoinverse of $\vG$.

\begin{definition}[Generalized partial symmetrization, \cite{KyngSPG22}]
    For a directed graph $\vG$ with $G \defeq \und(\vG)$ and $\beta \ge 0$, we
    let $\vG^{(\beta)} \defeq \beta \cdot G \cup \vG$ be the
    $\beta$-partial symmetrization of $\vG$.
\end{definition}
We remark that $\vLL_{\vG^{(\beta)}} = \beta\LL_G + \vLL_{\vG}$ and
$\LL_{G^{(\beta)}} = (2\beta+1) \LL_G$ where $G^{(\beta)} =
\und(\vG^{(\beta)})$.
Subsequently, one can sparsify the directed portion of $\vG^{(\beta)}$ much more
aggressively under the notion degree balance preserving spectral approximation,
since
\[
    \normop{\LL_{G^{(\beta)}}^{\frac \dag 2} (\vLL_{\vG} - \vLL_{\vH})
    \LL_{G^{(\beta)}}^{\frac \dag 2}}
    =
    \frac{1}{2\beta+1} \normop{\LL_G^{\frac \dag 2} (\vLL_{\vG} - \vLL_{\vH})
    \LL_G^{\frac \dag 2}}.
\]
Moreover, the graph $\beta G \cup \vH$ is also a good approximate pseudoinverse
of $\vG^{(\beta)}$.
One can further reduce the size of the graph by performing an undirected spectral
sparsification of the undirected portion $\beta G$.
These layers of sparsification allows them to build sparse preconditioner chain
which can be used in solving directed Eulerian Laplacian systems by applying an
iterative solver, such as Preconditioned Richardson (see
\Cref{lemma:preconsolver}).
We refer readers to \cite{KyngSPG22} and \cite{CohenKPPRSV17} for further
expositions of this idea.

\begin{lemma}[Preconditioned Iterations, Lemma 4.4 in \cite{CohenKPPRSV17}]
    \label{lemma:preconsolver}
    If $\ZZ$ is a $\lambda$-approximate pseudoinverse of $\MM$ with respect to
    $\UU$ for $\lambda \in (0,1)$, $\bb \perp \ker(\MM^\top)$ and $N \ge 0$, then
    $\xx_N = \PRalgo(\MM,\ZZ,\bb,N)$ computes $\xx_N = \ZZ_N \bb$ for some
    matrix $\ZZ_N$ only depending on $\ZZ,\MM,$ and $N$, such that $\ZZ_N$ is a
    $\lambda^N$-approximate pseudoinverse of $\MM$ with respect to $\UU$.
\end{lemma}

Following \cite{KyngSPG22}, we consider for $\vG_0 = \vG$, three graphs
$\vG_1,\vG_2,\vG_3$, such that each $\vG_{i+1}$ spectrally approximates $\vG_i$.
Specifically, 
\begin{enumerate}
    \item $\vG_1 = \vG_0^{(\beta)}$ is the $\beta$-partial symmetrization of
        $\vG_0$,
    \item $\vG_2 = \beta \cdot G \cup \vH_2$ is a sparsification of the directed
        portion $\vG$ of $\vG_1$,
    \item $\vG_3 = \beta \cdot H_3 \cup \vH_2$, the final sparsifier, further
        sparsifies the undirected portion of $\beta \cdot G$ of $\vG_2$.
\end{enumerate}

We now recall the definition of a \emph{sparsification quadruple}.
\begin{definition}[Sparsification quadruple, \cite{KyngSPG22}]
    \label{def:quad}
    We call strongly connected Eulerian graphs $\vG_0,\vG_1,\vG_2,\vG_3$ a
    $(\gamma,\beta,\alpha)$-quadruple if
    \begin{enumerate}
        \item $\vLL_{\vG_i}^\dag$ is a $1-\frac{1}{\gamma}$-approximate
            pseudoinverse of $\vLL_{\vG_{i-1}}$ with respect to $\LL_{G_i}$ for
            $i=1,2,3$,
        \item $\frac{1}{\gamma} \LL_{G_{i-1}} \pleq \LL_{G_i} \pleq \gamma
            \LL_{G_{i-1}}$ for $i=1,2,3$,
        \item $\ddi_{\vG_0} = \ddo_{\vG_0} = (1+\beta) \ddi_{\vG_i} =
            (1+\beta)\ddo_{\vG_i}$ for $i=1,2$, and $\ddi_{\vG_0} =
            (1+\frac{\beta}{\alpha}) \ddi_{\vG_3} = (1+\frac{\beta}{\alpha})
            \ddo_{\vG_3}$,
        \item $\vG_3$ is nearly-linear in $n$ and $m(\vG_i) = O(m)$ for $i=1,2$.
    \end{enumerate}
\end{definition}

Our adaptive degree (balance) preserving directed spectral sparsification
algorithm maintains explicitly a sparsification quadruple.
Similar to \Cref{thm:dynspecstar_for}, our sparsifier contains a few extra
vertices and the approximation guarantees are given for the Schur complements
with those vertices removed.
Having such extra vertices does not affect preconditioning, as we observed in
\Cref{lemma:schurprecon}.
\begin{theorem}[Adaptive algorithm for sparsification quadruple, formal version
    of \Cref{thm:dyn_quad_adp_inf}]
    \label{thm:dyn_quad_adp}
    There is an algorithm that given a weighted dynamic directed graph
    $\vG=(V,E,\ww)$ undergoing adaptive edge insertions and deletions, maintains
    a set of 3 directed graphs $\vG_1 = \vG^{(\beta)},\vG_2',\vG_3'$, where $V(\vG_2') = V \cup X$
    and $V(\vG_3') = V \cup X \cup Y$.
    Let $\vG_2 \defeq \Sc(\vG_2',V)$ and $\vG_3 \defeq \Sc(\vG_3',V)$
    Then, for $\gamma = O(\log^{24} n), \beta = O(\log^4 n), \alpha^{-1} =
    O(\log^{16} n)$, it satisfies with high probability
    \begin{enumerate}
        \item whenever $\vG$ is Eulerian, 
            $\vLL_{\vG_i}^\dag$ is a $1-\frac{1}{\gamma}$-approximate
            pseudoinverse of $\vLL_{\vG_{i-1}}$ with respect to $\LL_{G_i}$ for
            $i=1,2,3$,
        \item $\frac{1}{\gamma} \LL_{G_{i-1}} \pleq \LL_{G_i} \pleq \gamma
            \LL_{G_{i-1}}$ for $i=1,2,3$.
        \item $\ddi_{\vG_0} = (1+\beta) \ddi_{\vG_i}$, 
            $\ddo_{\vG_0} = (1+\beta)\ddo_{\vG_i}$ for $i=1,2$, and 
            $\ddi_{\vG_0} = (1+\frac{\beta}{\alpha}) \ddi_{\vG_3}$,  
            $\ddo_{\vG_0} = (1+\frac{\beta}{\alpha}) \ddo_{\vG_3}$,
        \item $m(\vG_2') = O(m)$, $m(\vG_3') = \tO(n \log W)$, and
            $|X|,|Y| \le \tO(n \log W)$.
    \end{enumerate}
\end{theorem}

When the input graph $\vG$ becomes Eulerian, for a constant approximate
pseudoinverse of $\vG_3$, we can solve the system of linear equation using at most
$\polylog(n,\eps^{-1})$ applications of the pseudoinverse and each matrix
$\vLL_{\vG_i}$ for $i=1,2,3$ using \Cref{lemma:preconsolver}.

The rest of this section is organized as follows.
We first show in \Cref{ssec:adp_prelim} how directed or undirected spectral
approximation implies approximate pseudoinverse for each $\vG_i$ to
$\vG_{i+1}$.
We state a modified blackbox reduction to dynamic algorithm on decremental
pruned uniform degree expanders from \cite{BernsteinvdBPGNSSS22} in
\Cref{ssec:adp_exp}.
In \Cref{ssec:adp_quad2to3}, we provide an algorithm for maintaining adaptively
a degree-preserving undirected spectral sparsifier.
We then show in \Cref{ssec:adp_quad1to2} how to sparsify the directed portion
(from $\vG_1$ to $\vG_2$) dynamically against adaptive adversary and prove our
main theorem of this section, \Cref{thm:dyn_quad_adp}

\subsection{Sparsification quadruple preliminaries}
\label{ssec:adp_prelim}

In this section, we consider the approximation factors between each $\vG_i$ and
$\vG_{i+1}$ under the notion of approximate pseudoinverse.

\begin{lemma}[Lemma 4.6, \cite{KyngSPG22}]
    \label{lemma:quad0to1}
    For Eulerian $\vG_1 = \vG_0^{(\beta)}$, the matrix $\vLL_{\vG_1}^\dag$ is an
    $(1-\frac{1}{1+2\beta})$-approximate pseudoinverse of $\vLL_{\vG_0}$ with
    respect to $\LL_{G_1}$ for $G_1 \defeq \und(\vG_1)$.
\end{lemma}
Note that we have some extra factors in \Cref{lemma:eulinvbound} since we define
$\und(\vG) = \vG \cup \rev(\vG)$ instead of $\frac{1}{2}(\vG \cup \rev(\vG))$.
Such changes can be observed in a few other lemmas. %

\Cref{lemma:quad1to2} states that if $\vG_2$ is a good degree balance preserving
spectral approximate to $\vG_1$, then it is also a good approximate
pseudoinverse.
\begin{lemma} \label{lemma:quad1to2}
    For $\beta > 0$, $\eps \in (0,1)$ and Eulerian $\vG$, if $\vG_2$ is a
    $\eps$-degree balance preserving directed spectral approximation of $\vG_1$
    and is Eulerian, then $\vLL_{\vG_2}^\dag$ is a $\frac{\eps \cdot
    (2\beta+1)}{\beta}$-approximate pseudoinverse of $\vLL_{\vG_1}$ with respect to
    $\LL_{\vG_2}$ and $(1-\frac{1}{2\beta+1}) \cdot \LL_{G_1} \pleq \LL_{G_2}
    \pleq (1+2\eps) \cdot \LL_{G_1}$.
\end{lemma}
\begin{proof}
    We have by assumption that
    \[
        \frac{1}{2\beta+1} \normop{\LL_G^\dag (\vLL_{\vG_2} - \vLL_{\vG_1})
        \LL_{G}^\dag}
        =
        \normop{\LL_{G_1}^\dag (\vLL_{\vG_2} - \vLL_{\vG_1}) \LL_{G_1}^\dag}
        =\eps,
    \]
    and that $\vH_2$ and $\vG_2$ are Eulerian by \Cref{lemma:dispec_equiv}.
    Then, \Cref{lemma:dispec_to_spec} gives that 
    \[
        (1-\frac{1}{2\beta+1})\LL_{G_1} \pleq \LL_{G_2} \pleq (1+2\eps)\LL_{G_1},
    \]
    where the lowerbound holds since $\LL_{G_2} \ge 2\beta \LL_G$ as
    $\vLL_{\vH_2} + \vLL_{\vH_2}^\top$ is PSD by Eulerian graph.
    We can bound pseudoinverse approximation factor by a multiplication of two
    terms for any $i=1,2,3$,
    \begin{equation} \label{eq:quadinv}
    \begin{aligned}
        \norm{ \PP_V - \vLL_{\vG_i}^\dag \vLL_{\vG_{i-1}} }_{\LL_{G_i} \to
        \LL_{G_i}}
        &=
        \normop{ \LL_{G_i}^{\frac{1}{2}}(\PP_V - \vLL_{\vG_i}^\dag \vLL_{\vG_{i-1}})
        \LL_{G_i}^{\frac \dag 2} }
        =
        \normop{ \LL_{G_i}^{\frac{1}{2}} \vLL_{\vG_i}^\dag (\vLL_{\vG_i} -
        \vLL_{\vG_{i-1}}) \LL_{G_i}^{\frac \dag 2} }
        \\
        &\le 
        \sqrt{\normop{\LL_{G_i}^{\frac{1}{2}} (\vLL_{\vG_i}^\top)^\dag \LL_{G_i}
        \vLL_{\vG_i}^\dag \LL_{G_i}^{\frac{1}{2}}}} \cdot
        \normop{\LL_{G_i}^{\frac \dag 2} (\vLL_{\vG_i} - \vLL_{\vG_{i-1}})
        \LL_{G_i}^{\frac \dag 2} }.
    \end{aligned}
    \end{equation}
    By \Cref{lemma:eulinvbound}, we get for the first term
    \[
        \sqrt{\normop{\LL_{G_i}^{\frac{1}{2}} (\vLL_{\vG_i}^\top)^\dag \LL_{G_i}
        \vLL_{\vG_i}^\dag \LL_{G_i}^{\frac{1}{2}}}}
        \le 2.
    \]
    Using our bound on undirected Laplacians above, the second term for $i=2$ is
    upperbounded by
    \[
        \normop{\LL_{G_2}^{\frac \dag 2} (\vLL_{\vG_2} - \vLL_{\vG_1})
        \LL_{G_2}^{\frac \dag 2} }
        \le
        \eps \cdot \frac{2\beta+1}{2\beta}.
    \]
    Combining then yields
    \[
        \norm{ \PP_V - \vLL_{\vG_2}^\dag \vLL_{\vG_1} }_{\LL_{G_2} \to
        \LL_{G_2}}
        \le
        \eps \cdot \frac{2\beta+1}{\beta}.
    \]
\end{proof}

Similarly, \Cref{lemma:quad2to3} states that using a good multiplicative
undirected spectral approximation to $G$, one can generate $\vG_3$ that is also
a good approximate pseudoinverse of $\vG_2$.
Here, we recall the definition of multiplicative undirected spectral
approximation.
\begin{definition}[Undirected spectral approximation, multiplicative]
    \label{def:specappr}
    $H=(V,E_H,\ww_H)$ is a $\gamma$-spectral approximation of $G=(V,E,\ww)$ if
    it satisfies for some scaler $\eta$ that
    \[
        \LL_G \pleq \eta \LL_H \pleq \gamma \cdot \LL_G.
    \]
\end{definition}

\begin{lemma} \label{lemma:quad2to3}
    For $\lambda > 1$ and Eulerian $\vG$, if the graph $H$ satisfies $\LL_G \pleq
    \LL_H \pleq \lambda \LL_G$, then $\vG_3 = \frac{1}{2}\beta(\lambda+1)\LL_H +
    \vH_2$ satisfies that $\vLL_{\vG_3}^\dag$ is a $(1-\frac{2}{\lambda+1})$-approximate
    pseudoinverse of $\vLL_{\vG_2}$ with respect to $\LL_{G_3}$ and $\LL_{G_2}
    \pleq \LL_{G_3} \pleq \frac{1}{2}\lambda(\lambda+1) \cdot \LL_{G_2}$.
\end{lemma}
\begin{proof}
    We again use \eqref{eq:quadinv} for the pseudoinverse bound and apply
    \Cref{lemma:eulinvbound} to bound the first term by 2.
    Let $\vG_3 = \frac{\beta}{\alpha}\LL_H + \vH_2$.
    For $i=3$, the second term is bounded by
    \begin{align*}
        \normop{\LL_{G_3}^{\frac \dag 2} (\vLL_{\vG_3} - \vLL_{\vG_2})
        \LL_{G_3}^{\frac \dag 2} }
        &=
        \normop{\LL_{G_3}^{\frac \dag 2} (\frac{\beta}{\alpha} \LL_H - \beta\LL_{G})
        \LL_{G_3}^{\frac \dag 2} }
        \\
        &\le
        \frac{1}{2}\normop{\LL_H^{\frac \dag 2} (\LL_H - \alpha \LL_G)  \LL_H^{\frac \dag
        2}}
        \\
        &\le
        \frac{1}{2}\max\Par{\alpha\lambda-1, 1-\alpha},
    \end{align*}
    which is minimized to $\frac{1}{2}(1-\frac{2}{\lambda+1})$ at $\alpha =
    \frac{2}{\lambda+1}$.
    Note that we get the first inequality above by
    \[
        \LL_{G_3} = \frac{2\beta}{\alpha} \LL_H + \LL_{H_2}
        \pgeq \frac{2\beta}{\alpha}\LL_H.
    \]
    We reach our desired pseudoinverse approximation guarantee by combining the
    two upperbounds.
    The spectral lowerbound of the undirected Laplacian is obtained by
    \[
        \LL_{G_3} = \beta(\lambda+1) \LL_H + \LL_{H_2}
        \ge 2\beta \LL_G + \LL_{H_2}
        = \LL_{G_2}.
    \]
    We get the upperbound by
    \[
        \LL_{G_3} = \beta(\lambda+1) \LL_H + \LL_{H_2}
        \le \frac{\lambda(\lambda+1)}{2}(2\beta \LL_G + \LL_{H_2})
        = \frac{\lambda(\lambda+1)}{2} \LL_{G_2}.
    \]
\end{proof}

\subsection{Dynamic reductions to pruned almost uniform degree expanders}
\label{ssec:adp_exp}

We will consider a variant of expanders where the minimum degree is guaranteed
to be close to the average degree.
We say that an graph $G=(V,E,\ww)$ is \emph{($\phi,\gamma$)-uniform degree
expander} if $G$ is a $\phi$-expander and and that 
$\frac{\max_{v \in V} \dd_v}{\min_{v \in V} \dd_v} \le \gamma^{-1}$, where we denote
$\dd_v = \sum_{e \in E | e \ni v} \ww_e$ as the weighted degree of $v$.
Recall that we assume the subgraphs in a decomposition contain the \emph{entire
set of edges}.
\begin{definition}[Uniform degree expander decomposition]\label{def:uniexp_partition}
We call $(\{\vG_i, \vG_i'\}_{i \in [I]}, \{Y_{v,i}\}_{i \in [I], v \in V})$ a 
\emph{$(\phi, \gamma, J)$-uniform degree expander decomposition}
of a directed graph $\vG=(V,E,\ww)$
if $\{\vG_i\}_{i \in [I]}$ are edge-disjoint subgraphs
satisfying $\bigcup_i \vG_i = \vG$, each $Y_{v,i} \subseteq V(\vG_i)$ for $v \in
V(\vG_i)$, and the following hold.
\begin{enumerate}
    \item \label{item:uniexp:partition:weight} \emph{Bounded weight ratio}: For all
        $i \in [I]$, $\frac{\max_{e \in E(\vG_i)}\ww_e}{\min_{e \in E(\vG_i)}\ww_e}
        \le 2$.
    \item \label{item:uniexp:partition:phi} \emph{Conductance}: For all $i \in
        [I]$, $\Phi(G_i) \ge \phi$ and $\Phi(G_i') \ge \phi$, where $G_i \defeq
        \und(\vG_i)$.
    \item \label{item:uniexp:partition:vertex} \emph{Vertex coverage}: Every
        vertex $v \in V(\vG)$ appears in at most $J$ of the subgraphs.
    \item \label{item:uniexp:partition:deg} \emph{Uniform degree}: For all $i \in
        [I]$, $\frac{\max_{v \in V} [\dd_{G_i'}]_v}{\min_{v \in V} [\dd_{G_i'}]_v}
        \le \frac{1}{\gamma}$.
    \item \label{item:uniexp:partition:contra} \emph{Contraction:} For all $i
        \in [I]$, contracting each $Y_{v,i}$ for $v \in V(\vG_i)$ in $\vG_i'$ gives
        $\vG_i$.
    \item \label{item:uniexp:partition:vol} \emph{Volume}: For all $i \in [I]$,
        $|V(\vG_i')| = O(|V(\vG_i)|)$, $\Vol(Y_{v,i}) = \Theta([\dd_{G_i}]_v)$.
\end{enumerate}
\end{definition}

We provide a dynamic uniform degree expander decomposition algorithm in
\Cref{prop:dynexpdecomp_uni} that is a modification of Theorem 5.8 in
\cite{BernsteinvdBPGNSSS22} for directed graphs, with edge weight bucketing,
and using the state-of-the-art expander decomposition algorithm from
\cite{AgassyDK23}.
Since the modification is minor we omit its proof.
\begin{proposition}[Dynamic uniform degree expander decomposition on directed
    graphs, Theorem 5.8 \cite{BernsteinvdBPGNSSS22}]
    \label{prop:dynexpdecomp_uni}
    There is a dynamic algorithm against an
    adaptive adversary that given a weighted directed graph $\vG=(V,E,\ww)$ with
    bounded weight ratio $\frac{\max_{e \in E} \ww_e}{\min_{e \in E} \ww_e} \le
    W$ at all time,
    preprocesses in time $O(m \log^7 n)$ time.
    The algorithm maintains a $(\phi, \phi^2, O(\log^2 m \log W))$-uniform degree
    expander decomposition $(\{\vG_i, \vG_i'\}_{i \in [I]}, \{Y_{v,i}\}_{i \in
    [I], v \in V})$ of $\vG$ at all time for $\phi = \Theta(\log^{-2} n)$ .
    In addition, the minimum weighted degree of each $G_i\defeq \und(\vG_i)$ can
    decrease by at most a factor of $\phi$.

    The algorithm supports edge insertions and deletions in $O(\log^{13} n)$
    amortized time.
    After each update, the output consists of a list of potential changes to the
    decomposition: (i) edge deletions to some $\vG_i$ and $\vG_i'$ in the decomposition,
    (ii) removing some $\vG_i,\vG_i'$ from the decomposition, and (iii) new
    subgraph added to the decomposition.
\end{proposition}

We now recall the reduction to decremental algorithm on uniform degre expanders
by \cite{BernsteinvdBPGNSSS22} using dynamic expander decomposition algorithm
from above.
The version we present here is a slight generalization: we remove the
perturbation condition and allows for directed graphs.
We again omit its proof as the generalization follows identically to the proof
of \Cref{lemma:red_dynexp}.
\begin{lemma}[Blackbox adaptive reduction with uniform degree promise,
    Theorem 5.2 \cite{BernsteinvdBPGNSSS22}]
    \label{lemma:red_dynexp_uni}
    Assume $\gH$ is a graph sparsification problem that satisfies the union
    and \emph{contraction} property, and there is a decremental algorithm $\gA$
    for $\gH$ on (directed) graphs with edge weights ratio at most $2$ and
    satisfying that their corresponding undirected graph is a
    $(\phi,\phi^2)$-uniform degree expander for $\phi \ge \cded \log^{-2} n$.
    Suppose $\gA$ preprocess in time $P(m,\eps) \ge m$, and maintains a
    sparsifier of size at most $S(n,\eps) \ge n$ with $N(n) \ge 0$ extra
    of vertices in $T(n,\eps)$ amortized update time.
    The maintained sparsifier is is either explicit, or is implicit and has
    query time $Q_1(n,\eps) \ge 1$ for an edge query or $Q_2(n,\eps) \ge
    S(n,\eps)$ for a graph query, where $P,N$ is superadditive
    in the first variable and $S,Q_2$ are subadditive in the first variable.

    There exists a fully dynamic algorithm $\gB$ for $\gH$ on
    weighted (directed) graphs such that given a dynamic weighted simple
    directed graph $\vG$ with weight ratios bounded by $W$, and a parameter
    $\eps$, $\gB$ preprocesses in time of $O(m \log^7 n + P(m,\eps))$ and
    maintains a sparsifier in $\gH(\vG,\eps)$ of size $O(S(n,\eps) \cdot \log^2
    n \log W)$ with $O(N(1) \cdot n \log^2 n \log W)$ extra number of vertices
    in amortized update time
    \[
        O\Par{(T(n,\eps)+ \frac{P(n^2,\eps)}{n^2}) \log^{13} n}.
    \]
    The sparsifier is explicit if $\gA$ maintains an explicit sparsifier.
    Otherwise, $\gB$ supports edge or graph query in time $O(Q_1(n,\eps) \cdot
    \log^2 n \log W)$ or $O(Q_2(n, \eps) \cdot \log^2 n \log W)$ respectively.
    Moreover, $\gB$ works against an adaptive adversary if $\gA$ supports
    adaptive edge updates.
\end{lemma}

\subsection{Adaptive degree-preserving sparsification for undirected graphs}
\label{ssec:adp_quad2to3}

We consider first the dynamic sparsification from $\vG_2$ to $\vG_3$, which, by
\Cref{lemma:quad2to3}, reduces to the sparsification of the undirected graph $G =
\und(\vG)$ against an adaptive adversary.

\begin{theorem} \label{thm:dyn_quad2to3}
    There is a fully dynamic algorithm that given as input a weighted undirected
    graph $G=(V,E,\ww)$ with bounded weight ratio $W$ at all time undergoing
    \emph{adaptive} edge insertions and deletions,
    maintain \emph{explicitly} a graph $\vH$ on vertices $V \cup X$ with $X$
    disjoint from $V$.
    The Schur complement $S = \Sc(\vH,V)$ is an undirected graph and satisfies with
    high probability that $S$ is $O(\log^{12} n)$-degree-preserving spectral
    approximation of $G$.
    The graph $\vH$ has size $\tO(n\log W)$ and extra
    number of vertices $|X| \le O(n \log^3 n \log W)$.
    The algorithm has preprocessing time $\tO(m)$ and amortized update time
    $\tO(1)$.
\end{theorem}

To prove \Cref{thm:dyn_quad2to3}, we use the observation that a
degree-preserving undirected cut sparsifier of an expander is a good spectral to
the expander as well (\Cref{lemma:exp_samedeg}).
This observation allows to utilize an adaptive algorithm for cut sparsification
from \cite{BernsteinvdBPGNSSS22} on decremental uniform degree expanders.

\begin{definition}[Undirected cut approximation, multiplicative]
    \label{def:cutappr}
    $H=(V,E_H,\ww_H)$ is a $\gamma$-cut approximation of $G=(V,E,\ww)$ if for
    all non-trivial ut $(U,V\setminus U)$, it satisfies that
    \[
        \ww(U,V\setminus U) \le \ww_H(U,V\setminus U) \le 
        \gamma \cdot \ww(U,V\setminus U).
    \]
\end{definition}

\begin{lemma}[\cite{ChuzhoyGLNPS20,BernsteinvdBPGNSSS22}]
    \label{lemma:exp_samedeg}
    Let $G,H$ be weighted undirected graphs with the same set of vertices and
    satisfying $\dd_G = \dd_H$.
    If both $G,H$ are $\phi$-expanders for $\phi \in (0,1]$, then 
    \[
        \frac{\phi^2}{4} \LL_G \pleq \LL_H \pleq \frac{4}{\phi^2} \LL_G.
    \]
\end{lemma}

We now show that a degree-preserving spectral sparsifier of an almost uniform
degree expander is easy to compute using cut sparsifiers and our patching scheme
in \Cref{lemma:dynextpatch_mod} in the static setting.
In \Cref{lemma:expspar_star}, our returned graph $\vH$ on vertices $V \cup X$ is
instead \emph{directed}.
This is not a problem: (1) our overall sparsifier $\vG_3$ is ultimately a
directed graph, and (2) the Schur complement $S = \Sc(\vH,V)$ back onto the
original set of vertices is an undirected graph. 
We further notice that \Cref{lemma:expspar_star} require specifically for the
graph $G$ to be a bipartite graph.
This is an important assumption to ensure our star patching does not generate
self-loops after ``unlifting'' the graph.
See \Cref{thm:dynspecstar_for_deg} in \Cref{ssec:expander} for the exact
reduction.
\begin{lemma} \label{lemma:expspar_star}
    Let weighted \emph{bipartite} undirected graph $G = (C \cup R,E,\ww)$ be a
    $\phi$-expander.
    If $\tG$ is a $\gamma$-cut approximation of $G$, then there is a
    deterministic algorithm that computes in time $O(m_G+m_{\tG}+n\log n)$ a
    directed graph $\vH$ on vertices $V \cup \{x,y\}$ for $V = C \cup R$ such that
    the Schur complement $S = \Sc(H,V)$ is an undirected graph that satisfies
    $\dd_S = \dd_G$, and
    $\frac{4\gamma^2}{\phi^2} \cdot S$ is a
    $\frac{16\gamma^4}{\phi^4}$-spectral approximation of $G$.
    The graph $H$ has size $m_H = O(m_{\tG}+n)$, 
\end{lemma}
\begin{proof}
    Notice first that the degrees of $\tG$ satisfies $\dd_G \le \dd_{\tG} \le
    \gamma \dd_G$ by considering cut on each single vertex.
    As $G$ is a $\phi$-expander, we also get for each cut $(U,V\setminus U)$
    with $\vol(U) \le \vol(V\setminus U)$ that
    \[ 
        \ww_{\tG}(U,V\setminus U) \ge \ww_G(U,V\setminus U) \ge \phi \vol_G(U).
    \]
    By $\gamma$-cut approximation, the $\dd_{\tG} \le \phi \dd_G$.
    We now scale down the cut sparsifier $\tG$ by a factor of $\frac{1}{\gamma}$.
    
    To compute the degree-preserving patching, we simply treat the undirected
    bipartite graph as a directed graph with each undirected edge $e = (s,t) \in C
    \times R$ corresponding to a directed edges $(s,t')$ and $(t,s')$ with the
    weight $\frac{1}{2}\ww_e$ each.
    The patching can then be computed using our directed star patching
    (\Cref{lemma:starpatch_algo}). 
    Note that we can separate out two star, one for edges from $C$ to $R'$ and
    another from $R$ to $C'$.
    The two stars are the reverse of each other when $C,C'$ and $R,R'$ are
    contracted back (see \Cref{lemma:blift_schur}).
    For the first star with center $x$, we consider the demands $\dd_1 =
    \frac{1}{2}(\dd_C - \eta [\dd_{\tG}]_C)$ and $\dd_2 = \frac{1}{2}(\dd_R -
    \eta [\dd_{\tG}]_R)$.
    For the second star with center $y$, we consider instead $\dd_2$ to $\dd_1$.
    Note that, since $G$ and $\tG$ are bipartite, $\|\dd_1\|_1 = \|\dd_2\|_1$.
    The degree-preserving condition $\dd_S = \dd_G$ for $S = \Sc(\vH,V)$ follows
    directly from \Cref{lemma:starpatch_algo}.
    Note that $S$ is an undirected graph, since the two stars we constructed
    above are ``reflections'' to each other.
    Moreover, we have by cut lowerbound of $G$ that for any $U \subseteq V$ with
    $\vol_G(U) \le \vol_G(V\setminus U)$ that
    \[ 
        \ww_S(U,V\setminus U) \ge \frac{1}{\gamma} \ww_G(U,V\setminus U) \ge
        \frac{\phi}{\gamma} \vol_G(U) = \frac{\phi}{\gamma} \vol_S(U),
    \]
    which certifies that $S$ is a $\frac{\phi}{\gamma}$-expander.
    Since \Cref{lemma:exp_samedeg} guarantees that
    \[
        \frac{\phi^2}{4\gamma^2} \LL_G \pleq \LL_H \pleq
        \frac{4\gamma^2}{\phi^2} \LL_G,
    \]
    we get our desired spectral approximation guarantee by
    scaling $H$ up by a factor of $\frac{4\gamma^2}{\phi^2}$.

    The runtime follows by the time to compute the weighted degrees and the time
    to compute the two stars.
    The sparsity bound follows by the guarantee that each star has at most
    $O(n)$ edges.
\end{proof}

We can now a decremental cut sparsification algorithm from
\cite{BernsteinvdBPGNSSS22} for spectral sparsification.
Remark that the following \Cref{lemma:dyn_cut_adap} has a slight modification to
allow for edge weights to be within a factor of 2, instead of uniform.
\begin{lemma}[Decremental adaptive cut sparsifier, Theorem 9.5
    \cite{BernsteinvdBPGNSSS22}]
    \label{lemma:dyn_cut_adap}
    Let $G=(V,E,\ww)$ be a decremental weighted graph undergoing \emph{adaptive} edge deletions
    that satisfies $L \le \ww_e \le 2L$ for all $e$ and $\Dmax \ge \dd_v \ge
    \Dmin \ge \frac{80\log n}{\phi}$ for some fixed $\phi,L$, where $\Dmax$ is the maximum weighted
    degree in the initial graph $G$ and $\Dmin$ is a degree lowerbound.
    Then, there is an algorithm that maintains a reweighted subgraph $H$ of $G$
    such that
    $m_H = \tO\Par{n (\frac{\Dmax}{\phi\Dmin})^2}$ and
    at any stage where $H$ has $\min_{v \in V} \dd_v \ge \Dmin$ and $\phi(H) \ge
    \phi$ on non-trivial vertices, $H$ is a $O(\log n)$-approximate cut sparsifier
    of $H$ with high probability.
    The algorithm has preprocessing time $O(m)$ and worst-case update time
    $\tO\Par{(\frac{\Dmax}{\phi\Dmin})^3}$.
\end{lemma}

We are now ready to prove our main statement of this section,
\Cref{thm:dyn_quad2to3}.
\begin{proof}[Proof sketch of \Cref{thm:dyn_quad2to3}]
    To facilitate degree-preserving sparsification using star~patching \\ 
    (\Cref{lemma:starpatch_algo}), we use the same edge decomposition approach
    from the proof of \Cref{thm:dynspecstar_for_deg} in \Cref{ssec:degpre}.
    This reduction only results in a $O(\log n)$ factor overhead on the
    sparsity and extra vertices.
    We remark that after taking the edge decomposition, the resulting subgraphs
    are all undirected bipartite graphs.
    The correctness of this edge decomposition scheme is guaranteed by
    \Cref{lemma:dispec_union} and \Cref{lemma:blift_schur}.

    Combining \Cref{lemma:dyn_cut_adap}, \Cref{lemma:expspar_star} and
    maintaining the star patchings dynamically, we get an adaptive algorithm for
    maintaining explicitly a $O(\log^{12} n)$-degree-preserving spectral
    approximation (using Schur complement) of a decremental pruned
    $\frac{C}{\log^2 n}$-expander.
    This algorithm has worst-case update time $\tO\Par{\frac{\Dmax^3}{\Dmin^3}}$ for
    maximum and minimum non-trivial weighted degrees $\Dmax,\Dmin$.
    Since the dynamic uniform expander decomposition algorithm in
    \Cref{prop:dynexpdecomp_uni} guarantees degree ratio of $\phi^2 = O(\log^4
    n)$, the update time is simply $\tO(1)$.
    The blackbox reduction framework \Cref{lemma:red_dynexp_uni}
    then guarantees our the claimed results.
\end{proof}

\subsection{Adaptive degree-preserving sparsification for partially
symmetrized directed graphs}
\label{ssec:adp_quad1to2}

We consider the last piece of the puzzle -- the dynamic sparsification from
$\vG_1$ to $\vG_2$.
By \Cref{lemma:quad1to2}, it suffices to produce a $\eps$-degree preserving
spectral sparsification $\vH$ of $\vG$ with respect $\LL_{G^{(\beta)}}$.
The exact statement of our guarantees is presented in \Cref{lemma:dyn_quad1to2}.
\begin{lemma} \label{lemma:dyn_quad1to2}
    Let $\vG =(V,E,\ww)$ be a dynamic weighted directed graph undergoing adaptive
    edge updates.
    There is a deterministic dynamic algorithm that for $\eps \in (0,1)$, if
    $\beta \ge \Omega(\eps^{-1}\log^4 n)$, it maintains explicitly a weighted
    directed graph $\vH=(V \cup X, E_{\vH},\ww_{vH})$ satisfying that 
    for Schur complement $\vS = \Sc(\vH,V)$
    $\vG' = \beta \cdot G \cup \vS$ is a
    $\eps$-degree-preserving directed spectral sparsifier of $\vG^{(\beta)}$,
    where $G = \und(\vG)$.
    The graph $\vH$ has size $\tO(n\log W)$ and extra
    number of vertices $|X| \le O(n \log^3 n \log W)$.
    The algorithm has preprocessing time $\tO(m)$ and amortized update time
    $\tO(1)$.
\end{lemma}

Similar to the undirected sparsification before (see \Cref{lemma:expspar_star}),
we start by showing that such a directed graph $\vH$ can be computed using a
patching (\Cref{lemma:dynextpatch_mod}) for expanders.
We again considers a bipartite graph with ``disjoint'' vertices to avoid
self-loops when eliminating the extra vertices.
\begin{lemma} \label{lemma:diexpspar_star}
    Let $\vG = (C \cup R,E,\ww)$ be a weighted bipartite directed graph with $E
    \subseteq C \times R$ so that $G=\und(\vG)$ is a $\phi$-expander.
    There is a deterministic algorithm that given $\eps \in (0,1)$, and $\beta$
    satisfying $\beta \ge 2\eps^{-1}\phi^{-2}$, returns in time $O(m+n\log n)$ a
    graph $\vH$ on vertices $C\cup R \cup \{x\}$ of size $O(n)$ such that for
    $\vG' = \beta \cdot G \cup \Sc(\vH,C\cup R)$ is a $\eps$-degree-preserving
    directed spectral sparsifier of $\vG^{(\beta)}$.
\end{lemma}
\begin{proof}
    Our algorithm simply computes a degree-preserving patching 
    using a star patching for the \emph{entire} weighted in and out degrees.
    The degree-preserving condition of $\vS \defeq \Sc(\vH,V)$ for $V = C \cup
    R$ then follows by \Cref{lemma:starpatch_algo}.
    We let $\ww_{\vS} \in \R_{\ge 0}^{C \times R}$ be the edge weights of the
    star generated after eliminating the centre $x$.

    Consider now the spectral approximation guarantees.
    By \Cref{lemma:blift_spectral}, the operator norm guarantee in
    \Cref{lemma:starpatch_algo} and \Cref{lemma:optoinf}, we have
    \begin{align*}
        \normop{[\DDi]^{\frac \dag 2} \BB^\top (\WW_{\vS}-\WW) \HH [\DDo]^{\frac \dag 2}}
        &=
        \normop{[\DDi]^{\frac \dag 2} (\HH^\top (\WW_{\vS}-\WW) \HH - \TT^\top
        (\WW_H-\WW)\HH) [\DDo]^{\frac \dag 2}}
        \\
        &=
        \normop{[\DDi]^{\frac \dag 2} \TT^\top (\WW_{\vS}-\WW)\HH [\DDo]^{\frac \dag 2}}
        \\
        &\le
        \normop{[\DDi]^{\frac \dag 2} \TT^\top \WW_{\vS}\HH [\DDo]^{\frac \dag 2}}
        +
        \normop{[\DDi]^{\frac \dag 2} \TT^\top \WW \HH [\DDo]^{\frac \dag 2}}
        \\
        &\le
        1+1 = 2,
    \end{align*}
    where the second equality follows by noticing that $\HH^\top \WW_{\vS} \HH =
    \HH^\top \WW \HH = \DDo$.
    Combining this inequality with the fact that $\DD \ge \DDi,\DDo$, $G$ is a
    $\phi$-expander, and Cheeger's inequality (\Cref{lemma:cheeger}) gives us
    \[
        \normop{\LL_G^{\frac \dag 2} (\vLL_{\vS} - \vLL_{\vG}) \LL_G^{\frac \dag 2}}
        \le
        \frac{2}{\phi^2}\normop{\DD_G^{\frac \dag 2} (\vLL_{\vS} - \vLL_{\vG}) \DD_G^{\frac \dag 2}}
        \le \frac{4}{\phi^2}.
    \]
    Finally, since $G^{(\beta)} \defeq \und(\vG^{(\beta)}) = (2\beta+1)G$, the
    overall spectral error with respect to $\LL_{G^{(\beta)}}$ is
    $\frac{4}{\phi^2(2\beta+1)}$, which is no more than $\eps$ as long as $\beta
    \ge 2\eps^{-1}\phi^{-2}$.

    The runtime of the algorithm is governed by the time to compute the degrees
    and the time to compute the patching using a static version of
    the algorithm in \Cref{lemma:dynextpatch_mod}.
\end{proof}

We can now prove \Cref{lemma:dyn_quad1to2}.
\begin{proof}[Proof sketch of \Cref{lemma:dyn_quad1to2}]
    Similar to the proof of \Cref{thm:dyn_quad2to3,thm:dynspecstar_for_deg}, we
    first take the bipartite lift of the graph (see \Cref{lemma:blift_spectral})
    and apply the edge decomposition as in \Cref{thm:dynspecstar_for_deg} from
    \Cref{ssec:degpre} to avoid self-loops when ``unlifting'' the graph.
    Again, the correctness is guaranteed by
    \Cref{lemma:dispec_union,lemma:blift_schur}.
    This reduction only results in a $O(\log n)$ factor overhead on the
    sparsity and extra vertices.

    Unlike \Cref{thm:dyn_quad2to3}, our algorithm for directed sparsification
    does not require the stronger blackbox reduction using
    \Cref{lemma:red_dynexp_uni}, since our sparsification algorithm from
    \Cref{lemma:diexpspar_star} can be trivially made deterministic dynamic --
    it only needs to maintain a dynamic star patching.
    Then, we can combine this sparsification algorithm on expanders to our
    previous blackbox reduction \Cref{lemma:red_dynexp} to obtain our claimed
    results.
\end{proof}

We have all the ingredients to prove \Cref{thm:dyn_quad_adp}
\begin{proof}[Proof of \Cref{thm:dyn_quad_adp}]
    Our algorithm maintains (a) a copy of the algorithm from
    \Cref{lemma:dyn_quad1to2} for sparsification from $\vG_1$ to $\vG_2$, and
    (b) a copy of the algorithm from \Cref{thm:dyn_quad2to3} for sparsification
    from $\vG_2$ to $\vG_3$ with the sparsifier scaled appropriately.
    Note that (a) maintains a graph $\vH_2$ on $V \cup X$ and (b) maintain a graph
    $\vH_3$ on $V \cup X'$ for some subset $X'$ disjoint from $X$.
    By the disjointness, we have the following equivalence in the Schur
    complement, 
    \[
        \Sc(\vG_3,V) = \Sc(\vH_3 \cup \vH_2,V) = \Sc(\vH_3,V) \cup \Sc(\vH_2,V).
    \]
    Then, the union property (\Cref{lemma:dispec_union}) then ensures the Schur
    complement $\Sc(\vH_3 \cup \vH_2,H)$ is a degree-preserving spectral
    approximation to $\vG_2$.

    For (a), we simply fix an approximation factor $\eps = 0.001$ and take
    the minimum possible $\beta \defeq C \log^4 n$ for some constant $C$ given by
    \Cref{lemma:dyn_quad1to2}.

    For (b), the approximation factor is capped at $\lambda = O(\log^{12} n)$.
    Then, by \Cref{lemma:quad1to2,lemma:quad2to3}, the $\gamma$ parameter must
    satisfies $\gamma \ge \frac{1}{2}\lambda(\lambda+1) = O(\log^{24} n)$.
    We also get the $\alpha$ parameter needs to satisfy $\alpha^{-1} =
    \frac{1}{2}\beta(\lambda+1) = O(\log^{16} n)$.
    To conclude, our parameter choices are $\gamma = O(\log^{24} n)$, $\beta =
    O(\log^4 n)$ and $\alpha = \Omega(1/\log^{16} n)$.

    The preprocessing time, update time, extra vertices, and sparsity guarantees
    all follows by combining the respective guarantees from
    \Cref{thm:dyn_quad2to3,lemma:dyn_quad1to2}.
\end{proof}

\printbibliography

\appendix
\section{Missing proofs from Section~\ref{ssec:dispec_prelim}}
\label{app:specproof}

The following is a proof on the kernel of the difference in directed Laplacians.
\begin{proof}[Proof of \Cref{lemma:dispec_ker}]
    For a connected undirected graph $G$, the left and right kernels of $\LL_G =
    \BB^\top \WW \BB$ is $\spn(\vone_V)$.
    We get our kernel guarantees by
    \begin{gather*}
        (\vLL_{\vG} - \vLL_{\vG'}) \vone_V = \BB^\top (\WW - \WW') \HH \vone_V
        = \BB^\top (\ww - \ww') = \vzero_V
        \\
        (\vLL_{\vG} - \vLL_{\vG'})^\top \vone_V = \HH^\top (\WW - \WW') \BB \vone_V
        = \HH^\top (\WW - \WW') \vzero_E = \vzero_V.
    \end{gather*}
\end{proof}

We provide below a proof of \Cref{lemma:dispec_to_spec}.
\begin{proof}[Proof of \Cref{lemma:dispec_to_spec}]
    By \Cref{lemma:dispec_ker,lemma:dispec_equiv}, we may assume w.l.o.g. that
    $G$ is connected.
    Then,
    \[
        \vLL_{\rev(\vG)} - \vLL_{\rev(\vG')}
        = -\BB^\top (\WW - \WW') \TT  
        = \HH^\top (\WW - \WW') \BB
        = (\vLL_{\vG} - \vLL_{\vG'})^\top
    \]
    where we used \Cref{fact:dispec_top} and the degree balance preserving
    assumption that $\BB^\top \ww = \BB^\top \ww'$ in the second equality.
    Then,
    \begin{align*}
        \normop{\LL_G^{\frac \dag 2} (\LL_G - \LL_{G'}) \LL_G^{\frac \dag 2}}
        &=
        \normop{\LL_G^{\frac \dag 2} ((\vLL_{\vG} - \vLL_{\vG'}) +
        (\vLL_{\rev(\vG)} - \vLL_{\rev(\vG')})) \LL_G^{\frac \dag 2}}
        \\
        &\le 
        2 \normop{\LL_G^{\frac \dag 2} (\vLL_{\vG} - \LL_{\vG'}) \LL_G^{\frac \dag 2}}
        \le
        2\eps,
    \end{align*}
    as required.
\end{proof}

The following is a proof of the union property of degree-balance preserving
directed spectral approximations.
\begin{proof}[Proof of \Cref{lemma:dispec_union}]
    Let $\xx,\yy \in \R^V$ be arbitrary vectors.
    Then, we get by \eqref{eq:dispec_def}
    \begin{align*}
        \Abs{\xx^\top (\vLL_{\vG} - \vLL_{\vG'}) \yy}
        &\le
        \sum_{i=1}^k s_i \cdot \Abs{\xx^\top (\vLL_{\vG_i} - \vLL_{\vG_i'}) \yy}
        \le
        \eps \sum_{i=1}^k s_i \cdot 
        \sqrt{\xx^\top \LL_{G_i} \xx \cdot \yy^\top \LL_{G_i} \yy}
        \\
        &\le
        \eps \sqrt{\xx^\top \Par{\sum_i s_i \cdot \LL_{G_i}} \xx \cdot 
        \yy^\top \Par{\sum_i s_i \cdot \LL_{G_i}} \yy}
        &&
        (\mbox{Cauchy-Schwarz})
        \\
        &=
        \eps \sqrt{\xx^\top \LL_G \xx \cdot \yy^\top \LL_G \yy},
    \end{align*}
    as required.
\end{proof}

We provide a proof for \Cref{lemma:blift_spectral} in \Cref{ssec:dispec_prelim}.
\begin{proof}[Proof of \Cref{lemma:blift_spectral}]
Consider the edge-vertex transfer matrix of $\vG$, $\BB_{\vG} = \HH_{\vG} - \TT_{\vG}$.
The edge-vertex transfer matrix of $\vG^\uparrow$ is then
$\BB_{\vG^\uparrow} = \begin{pmatrix} \HH_{\vG} & -\TT_{\vG} \end{pmatrix}$.
Hence, any vector $\xx$ satisfying $\BB_{\vG^\uparrow}^\top \xx = \vzero_V$ must have
$\HH_{\vG}^\top \xx = \vzero_V, \TT_{\vG}^\top \xx = \vzero_V$, giving us
preservation of both the difference between in and out degrees and the
sum of in and out degrees $\vG$, i.e.,
\[
    \BB_{\vG}^\top \xx = \vzero_V,\; |\BB_{\vG}|^\top \xx = \vzero_V,
\]
where we used that $\BB_{\vG} = \HH_{\vG} - \TT_{\vG}$ and $|\BB_{\vG}| = \HH_{\vG} + \TT_{\vG}$.
Taking $\xx \defeq \ww' - \ww$ then gives the the first two claims.
We remark that the directed graph $\vG$ need not be Eulerian.

We proceed to prove the third claim. For ease of notation, we omit $\vG$ in the subscripts of the matrices and
denote $\AA \defeq \AA_{\vG}$, $\AA_\uparrow \defeq \AA_{\vG^\uparrow}$ for all
matrices $\AA$.
We also use the following equivalent definition of operator norms with the convention
that the fraction is 0 if the numerator is 0:
\[
    \normop{\AA} = 
    \max_{\xx,\yy} \frac{|\xx^\top \AA \yy|}{\norm{\xx}_2
    \norm{\yy}_2},
\]
where the $\max$ is over $\xx, \yy$ of compatible dimensions.
Also, we let $\QQ \in \R^{(V \cup \bV) \times V}$ be defined by 
\[\QQ\ee_v = \begin{pmatrix} \ee_v \\ \ee_{\bv} \end{pmatrix} \]
for all $v \in V$, where $\bv$ is identified with $v$. 
Notice that $\HH = \HH_\uparrow\QQ$, $\TT = \TT_\uparrow\QQ$ and $\BB = \BB_\uparrow\QQ$.
Then,
\[
    \QQ^\top \LL_\uparrow \QQ
    =
    \QQ^\top \BB_\uparrow^\top \WW \BB_\uparrow \QQ
    =
    \BB^\top \WW \BB
    =
    \LL.
\]
Finally, for any non-trivial vectors $\xx,\yy \in \R^V$ satisfying $\xx,\yy \perp
\vone_V$, and defining 
\[\aa \defeq \LL^{\frac \dagger 2}\xx,\; 
\bb \defeq \LL^{\frac \dagger 2}\yy,\; \hxx \defeq \LL_\uparrow^{\frac \dagger 2}\QQ \aa,\; 
\hyy \defeq \LL_\uparrow^{\frac \dagger 2}\QQ \bb,\]
 we have
\begin{align*}
  \frac{|\xx^\top \LL^{\frac \dagger 2} \BB^\top (\WW'-\WW) \HH 
    \LL^{\frac \dagger 2} \yy|}{\norm{\xx}_2 \norm{\yy}_2} &= 
    \frac{|\aa^\top \BB^\top (\WW'-\WW) \HH \bb|}{\norm{\aa}_{\LL}
    \norm{\bb}_{\LL}}\\
    &=
    \frac{|\aa^\top \QQ^\top \BB_\uparrow^\top (\WW'-\WW) \HH_\uparrow \QQ
    \bb|}{\norm{\QQ \aa}_{\LL_\uparrow} \norm{\QQ \bb}_{\LL_\uparrow}} \\
    &=
    \frac{|\hxx^\top \LL_\uparrow^{\frac \dagger 2} \BB_\uparrow^\top 
        (\WW'-\WW) \HH_\uparrow \LL_\uparrow^{\frac \dagger 2} \hyy|}
        {\norm{\hxx}_2 \norm{\hyy}_2},
\end{align*}
giving us the desired operator norm bound.
We note that $\QQ \aa,\QQ \bb \perp \ker(\LL_\uparrow)$.

Finally, we record a consequence of this proof we will later use. Recall that we have shown
\[\HH = \HH_{\uparrow}\QQ,\quad
\TT = \TT_{\uparrow}\QQ,\quad
\BB = \BB_{\uparrow}\QQ,\quad
\QQ^\top \LL_{\uparrow} \QQ = \LL.\]
Therefore, suppose that for some $\eps > 0$ and fixed vectors $\xx,\yy \in \R^V$,  $\ww' \in \R^E_{>0}$ satisfies
\begin{equation}\label{eq:lift_imply_1}
\Abs{\xx^\top \QQ^\top \BB_{\uparrow}^\top (\WW' - \WW) \HH_{\uparrow} \QQ \yy} \le
\eps \cdot \norm{\QQ\xx}_{\LL_{\uparrow}} \norm{\QQ\yy}_{\LL_{\uparrow}}.
\end{equation}
Then, we also have the bound in the unlifted graph $G$,
\begin{equation}\label{eq:lift_imply_2}
 \Abs{\xx^\top \BB^\top (\WW' - \WW) \HH \yy} \le
\eps \cdot \norm{\xx}_{\LL} \norm{\yy}_{\LL}.
\end{equation}
\end{proof}

\section{Dynamic expander decomposition proofs}\label{app:dynexpdecomp}

We provide a proof to our dynamic expander decomposition guarantees.

\begin{proof}[Proof of \Cref{prop:dynexpdecomp}]
    We maintain a decomposition of the graph $G$ into subgraphs
    $G_1, G_2, \ldots, G_k$ where each $G_i$ contains at most $2^i$ edges and $k
    = \ceil{2(\log(n) - \log(2))} \le (2\log(n)+1)$.
    Initially, we simply set $G_i = G$ for $i = \ceil{\log m}$ and all other
    graphs as empty.
    For each subgraph $G_i$, we maintain an dynamic expander decomposition
    $\{G_{i,j}\}_j$ of $G_i$ where each $G_{i,j}$ is a $\phi$-expander.
    We also associated for each $G_{i,j}$ a instance of the expander pruning
    algorithm.

    For edge deletions, suppose some edge $e$ currently in $G_i$ is to be
    deleted, we perform the expander pruning algorithm in \Cref{lemma:expprun}
    on the subgraph $G_{i,j}$ that has edge $e$ in it.
    The cut edges in $G_{i,j}$ induced by the pruned set $P$ are then moved from
    $G_i$ into $G_1$ and treated as edge insertions.
    Note that after $\frac{\phi m_{ij}}{10}$ deletions, the entire remaining
    subgraph $G_{i,j}$ is removed from $G_i$ and moved into $G_1$. 

    Suppose now that a set of edges $F$ is to be inserted into $G_i$.
    If, $|E(G_i) \cup F| > 2^i$, our algorithm insert the
    enture $E(G_i) \cup F$ into $G_{i+1}$ and set $G_i$ to be empty.
    Otherwise, we remove the previous expander decomposition and recomputed one
    with the current set of edges of $G_i$ using the statice algorithm of
    \Cref{prop:ex_partition}.
    Again, we re-initialize for each expander an instance of the expander
    pruning algorithm.
    An edge insertion to $G$ is treated as a single edge insertion to $G_1$.

    \Cref{prop:ex_partition,lemma:expprun} guarantee that each $G_{i,j}$ is
    $\frac{\phi}{12}$ expander for $\phi$ in \Cref{prop:ex_partition}.
    Our algorithm ensures each present edge in $G$ is included in the some
    subgraph $G_{i,j}$ and that $\bigcup_{i,j} G_{i,j} = G$.
    Moreover, since the weight of each edge remain unchanged, we are also
    guarantee a bound weight ratio of 2 for all expanders.
    Further, for every $i$ we have that there is a $(\log |E_{G_i}| + 1)(\log W
    + 3) = (i+1)(\log W + 3)$ partition for the coverage of the vertices
    $V(G_i)$. 
    Thus, the overall coverage is $\sum_{i=1}^{2\log n+1} (i+1)(\log W + 3) \le
    \log(n) (2\log^2 n + 3)(\log W + 3)$.

    Consider now the runtime and recourse of our algorithm.
    Suppose we insert a set of edges $F$ into $G_i$ and have $|E(G_i) \cup F|
    \le 2^i$, then re-initialize of $G_i$ takes, by
    \Cref{prop:ex_partition,lemma:expprun}, $O(2^i \log^7 n)$ time.
    For each $G_i$, it takes at least $2^{i-1}$ edge insertions before another
    re-initialize.
    Thus, the amortized recourse and cost per edge insertion (including ones
    incurred by edge deletions) are $O(1)$ and $O(\log^7 n)$.
    As for edge deletions, \Cref{lemma:expprun} gives that each edge deletion
    incurs an amortized (not just average) $O(\phi^{-1}) = O(\log^2 n)$ edge
    insertions, a $O(i \phi^{-1}) = O(\log^3 n)$ amortized recourse and a $O(i
    \phi^{-2}) = O(\log^5 n)$ amortized cost.
    When combined with the amortized recourse and cost of edge insertion, we get
    amortized $O(\log^3 n)$ recourse and $O(\log^7 n)$ time for an edge
    deletion.
\end{proof}

\section{Rounding}\label{app:rounding}

In this section, we provide a proof to \Cref{lemma:rounding}, our guarantee on
\RO.
\restaterounding*

We need the following lemma on bounding the matrix quadratic form of directed
cycles.
\begin{lemma}[Lemma 8 of \cite{SachdevaTZ24}] \label{lemma:cycle}
For a directed graph $\vG = (V,E,\ww)$, let $\cc \in \{0,\pm v\}^E$ be a circulation
on $\vG$ (i.e., $\BB_{\vG}^\top \cc = 0$) such that $|\supp(\cc)| = L$ for $L,v
> 0$, for $e \in \supp(\cc)$, $v \le \ww_e$, and the set of
edges in $\supp(\cc)$ form a cycle.\footnote{The cycle edges are not necessarily
in the same direction. See \cite{ChuGPSSW18,SachdevaTZ24}.} 
If each edge $e = (u,w) \in \supp(\cc)$ satisfies for $G \defeq \und(\vG)$ that $v \cdot
\ER_G(u,w) \le \rho$, then 
\[
\BB_{\vG}^\top \CC \HH_{\vG} \LL_G^\dag \HH_{\vG}^\top \CC \BB_{\vG} 
\le L^2 \rho \cdot \BB_{\vG}^\top \CC \BB_{\vG}. 
\]
\end{lemma}

\begin{proof}[Proof of \Cref{lemma:rounding}]
Throughout the proof we drop the subscripts $\vG$, $G$ from $\BB$, $\HH$, $\LL$
for simplicity.
We assume w.l.o.g. that the corresponding undirected $G$ is connected.
The algorithm sets $\yy$ to be the unique flow on the edges of tree $T$ that
satisfies $\BB^{\top}\yy = \dd.$
Such a vector $\yy$ can be constructed in $O(n)$ time by proceeding sequentially
from an arbitrary leaf in, e.g., depth-first order.
Let $\zz \in \R^E$ be any vector that satisfies $\BB^\top \zz = \dd$ so that
$\yy - \zz$ is a circulation on $G$.

It remains to show that this weight update does not introduce significant error
in the directed Laplacians and no tree edge is significantly congested.
For every edge $e \notin T,$ we let $\cc^{(T,e)} \in \{0, \pm 1\}^E$ denote the
(signed) incidence vector of the unique cycle in $T \cup {e}$.
We observe that $\zz-\yy$ can be expressed uniquely as $\sum_{e \notin T} \zz_e
\cc^{(T,e)}$, so for each $e \in T$, $\yy_e = \zz_e - \sum_{f \in T | e \in
\supp(\cc^{T,f})} \zz_f$, and
\[
\normop{\LL^{\frac \dag 2} \BB^\top (\YY-\ZZ) \HH \LL^{\frac \dag 2}} 
\le \sum_{e \notin T}  \frac{|\zz_e|}{\ww_e} \normop{\ww_e \cdot \LL^{\frac \dag 2}
\BB^\top \CC^{(T,e)} \HH \LL^{\frac \dag 2}}.
\]
Since $T$ is a maximum weighted spanning tree, it satisfies that for any
$e \in \supp(\cc^{T,f})$, we have $\ww_f \le \ww_e$.
Then,
\[
    |\ww_e ^{-1} \yy_e| \le |\ww_e^{-1} \zz_e| + \sum_{f \in T | e \in
    \supp(\cc^{T,f})} |\ww_f^{-1} \zz_f|
    \le
    \norm{\WW^{-1} \zz}_1,
\]
giving the second guarantee.
For the spectral error, it suffices to show that each operator norm in the right-hand side is bounded by $n$.
Note that
\begin{align*}
\normop{\ww_e \LL^{\frac \dag 2} \BB^\top \CC^{(T,e)} \HH \LL^{\frac \dag 2}} 
&= 
\sqrt{\normop{(\LL^{\frac \dag 2} (\ww_e \cdot \BB^\top \CC^{(T,e)}) \HH
\LL^{\frac \dag 2})(\LL^{\frac \dag 2} (\ww_e \cdot \HH^\top \CC^{(T,e)} \BB) \LL^{\frac \dag 2})^{\top}}} 
\\
&= 
\sqrt{\normop{\LL^{\frac \dag 2} (\ww_e \cdot \BB^\top \CC^{(T,e)} \HH)
\LL^{\dag} (\ww_e \cdot \BB^\top \CC^{(T,e)} \HH)^{\top} \LL_G^{\frac \dag 2}}}.
\end{align*}
We will bound the norm of the last matrix in the above expression.
Observe that $\ww_e \cdot \BB^\top \CC^{(T,e)} \HH$ is just the directed
Laplacian of the cycle with uniform edge weights $\ww_e$. 
Denote it $\vMM$ for brevity. 
Again, using the fact that $T$ is a maximum weighted spanning tree and by noting
$L \le n$ and $\rho \le 1$, \Cref{lemma:cycle} gives
\[
    \vMM \LL^{\dag} \vMM^\top \preceq n^2 \MM,
    \quad \text{where } \MM \defeq \BB^\top \CC^{(T,e)} \BB.
\]
Notice that the corresponding undirected graph of Laplacian $\MM$ is a
downweighted subgraph of $G \defeq \und(\vG)$, we have the third result by
\[
    \LL^{\frac \dag 2} \vMM \LL^{\dag} \vMM^\top \LL^{\frac \dag 2}
    \preceq n^2 \cdot \LL^{\frac \dag 2} \MM \LL^{\frac \dag 2}
    \preceq n^2 \cdot \II_V.
\]
\end{proof}

\section{Missing proofs from Section~\ref{ssec:dicut_prelim}}
\label{app:dicutproof}

We show in the following that dicut approximation implies undirected cut approximation.
\begin{proof}[Proof of \Cref{lemma:dicut_to_cut}]
    For any non-empty proper subset of vertices $U \subsetneq V$, we have
    $\ww_{\vG}(U,V\setminus U) + \ww_{\vG}(V\setminus U,U) = \ww_G(U, V\setminus
    U)$.
    As square root is concave, we get
    \[
        \sqrt{\ww_{\vG}(U,V\setminus U)} + \sqrt{\ww_{\vG}(V\setminus U,U)} \le
        \sqrt{2 \ww_G(U, V\setminus U)}
    \]
    The same equality also holds for $\vH$.
    Then, the desired cut approximation guarantee holds by adding the two dicut
    inequalities for $(U,V\setminus U)$ and $(V\setminus U, U)$.
\end{proof}

Below is a proof of the union property of dicut approximations.
\begin{proof}[Proof of \Cref{lemma:dicut_union}]
    Let $\vC$ be an arbitrary dicut from $U$ to $V \setminus U$.
    By assumption, we have for every $i$
    \[
        \ww_{\vH_i}(U,V\setminus U) \in
        \Par{\ww_{\vG_i}(U,V\setminus U) \pm \frac{\eps}{\sqrt{\beta+1}}
        \sqrt{\ww_{\vG_i}(U,V\setminus U) \cdot \ww_{G_i}(U,V\setminus U)}},
    \]
    where $G_i \defeq \und(\vG_i)$.
    Note that 
    \[
        \ww_{\vG}(U,V\setminus U) = \sum_{i=1}^k s_i \cdot
        \ww_{\vG_i}(U,V\setminus U), \quad
        \ww_{G}(U,V\setminus U) = \sum_{i=1}^k s_i \cdot
        \ww_{G_i}(U,V\setminus U).
    \]
    Then, by Cauchy-Schwarz,
    \begin{align*}
        \sum_{i=1}^k s_i \cdot \sqrt{\ww_{\vG_i}(U,V\setminus U) \cdot \ww_{G_i}(U,V\setminus U)}
        &\le 
        \sqrt{(\sum_{i=1}^k s_i \cdot \ww_{\vG_i}(U,V\setminus U)) (\sum_{i=1}^k s_i \cdot
        \ww_{G_i}(U,V\setminus U))}
        \\
        &=
        \sqrt{\ww_{\vG}(U,V\setminus U) \ww_{G}(U,V\setminus U)}.
    \end{align*}
    Thus, summing over all $i$ and taking in the inequality above, we have
    \[
        \ww_{\vH}(U,V\setminus U) \in
        \Par{\ww_{\vG}(U,V\setminus U) \pm \frac{\eps}{\sqrt{\beta+1}}
        \sqrt{\ww_{\vG}(U,V\setminus U) \cdot \ww_{G}(U,V\setminus U)}},
    \]
    which, by definition, guarantees that $\vH$ is a $(\beta,\eps)$-dicut
    approximation to $\vG$.
\end{proof}

\end{document}